\documentclass[a4paper,UKenglish,cleveref, autoref]{lipics-v2019}

\usepackage{amsmath, amsthm, amssymb}
\usepackage{graphicx}
\usepackage{subcaption}
\usepackage{gensymb}
\usepackage{mathtools}
\usepackage{cleveref}

\DeclarePairedDelimiter\ceil{\lceil}{\rceil}
\DeclarePairedDelimiter\floor{\lfloor}{\rfloor}

\newcommand\numberthis{\addtocounter{equation}{1}\tag{\theequation}}

\title{Pushing Lines Helps: Efficient Universal Centralised Transformations for Programmable Matter} 

\titlerunning{Pushing Lines Helps: Efficient Universal Centralised Transformations for Programmable Matter}

\author{Abdullah Almethen}{Department of Computer Science, University of Liverpool, UK}{A.almethen@liverpool.ac.uk}{}{}
\author{Othon Michail}{Department of Computer Science, University of Liverpool, UK}{Othon.Michail@liverpool.ac.uk}{}{}
\author{Igor Potapov}{Department of Computer Science, University of Liverpool, UK}{Potapov@liverpool.ac.uk}{}{}

\authorrunning{A. Almethen, O. Michail  and I. Potapov}

\Copyright{Abdullah Almethen, Othon Michail and Igor Potapov}



\category{}

\supplement{}

\acknowledgements{}

\nolinenumbers 

\EventEditors{}
\EventNoEds{2}
\EventLongTitle{}
\EventShortTitle{}
\EventAcronym{}
\EventYear{}
\EventDate{}
\EventLocation{}
\EventLogo{}
\SeriesVolume{}
\ArticleNo{}

\begin{document}

\maketitle

\begin{abstract}
In this paper, we study a discrete system of entities residing on a two-dimensional square grid. Each entity is modelled as a node occupying a distinct cell of the grid. The set of all $n$ nodes forms initially a connected shape $A$. Entities are equipped with a linear-strength pushing mechanism that can push a whole line of entities, from 1 to $n$, in parallel in a single time-step. A target connected shape $B$ is also provided and the goal is to \emph{transform} $A$ into $B$ via a sequence of line movements. Existing models based on local movement of individual nodes, such as rotating or sliding a single node, can be shown to be special cases of the present model, therefore their (inefficient, $\Theta(n^2)$) \emph{universal transformations} carry over. Our main goal is to investigate whether the parallelism inherent in this new type of movement can be exploited for efficient, i.e., sub-quadratic worst-case, transformations. As a first step towards this, we restrict attention solely to centralised transformations and leave the distributed case as a direction for future research. Our results are positive. By focusing on the apparently hard instance of transforming a diagonal $A$ into a straight line $B$, we first obtain transformations of time $O(n\sqrt{n})$ without and with preserving the connectivity of the shape throughout the transformation. Then, we further improve by providing two $O(n\log n)$-time transformations for this problem. By building upon these ideas, we first manage to develop an $O(n\sqrt{n})$-time universal transformation.  Our main result is then an $ O(n \log n) $-time universal transformation. We leave as an interesting open problem a suspected $\Omega(n\log n)$-time lower bound.    
\end{abstract}

\section{Introduction}
\label{sec:intro}

As a result of recent advances in components such as micro-sensors, electromechanical actuators, and micro-controllers, a number of interesting systems are now within reach. A prominent type of such systems concerns collections of small robotic entities. Each individual robot is equipped with a number of actuation/sensing/communication/computation components that provide it with some autonomy; for instance, the ability to move locally and to communicate with neighbouring robots. Still, individual local dynamics are uninteresting, and individual computations are restricted due to limited computational power, resources, and knowledge. What makes these systems interesting is the collective complexity of the population of devices. A number of fascinating recent developments in this direction have demonstrated the feasibility and potential of such collective robotic systems, where the scale can range from milli/micro \cite{BG15,GKR10,KCL12,RCN14,YSS07} down to nano \cite{DDL09,Ro06}. 

This progress has motivated the parallel development of a theory of such systems. It has been already highlighted \cite{MS18} that a formal theory (including modelling, algorithms, and computability/complexity) is necessary for further progress in systems. This is because theory can accurately predict the most promising designs, suggest new ways to optimise them, by identifying the crucial parameters and the interplay between them, and provide with those (centralised or distributed) algorithmic solutions that are best suited for each given design and task, coupled with provable guarantees on their performance. As a result, a number of sub-areas of theoretical computer science have emerged such as mobile and reconfigurable robotics \cite{ABD13,BKR04,CFPS12,CKLL09,DFSY15,DDG18,DGMR06,DDGR14,DGR15,DGRSS16,DiLuna2019,ALFNVG8,FPS12,KKM10,MICHAIL2019,SMO18,SY10,YUY16,YSS07}, passively-mobile systems \cite{AADFP06, AAER07,MS16a,MS18} including the theory of DNA self-assembly \cite{Do12,RW00,Wi98,WCG13}, and metamorphic systems \cite{DP04,DSY04a,DSY04b,NGY00,WWA04}; connections are even evident with the theory of puzzles \cite{BDFL17,De01,HD05}. A latest ongoing effort is to join these theoretical forces and developments within the emerging area of ``Algorithmic Foundations of Programmable Matter'' \cite{FRRS16}. \emph{Programmable matter} refers to any type of matter that can \emph{algorithmically} change its physical properties. ``Algorithmically'' means that the change (or \emph{transformation}) is the result of executing an \emph{underlying program}.

In this paper, we embark from the model studied in \cite{DP04,DSY04a,DSY04b,MICHAIL2019}, in which a number of spherical devices are given in the form of a (typically connected) shape $A$ lying on a two-dimensional square grid, and the goal is to transform $A$ into a desired target shape $B$ via a sequence of valid movements of individual devices. In those papers, the considered mechanisms were the ability to rotate and slide a device over neighbouring devices (always through empty space). We here consider an alternative (linear-strength) mechanism, by which a line of one or more devices can translate by one position in a single time-step.   

As our main goal is to determine whether the new movement under consideration can \emph{in principle} be exploited for sub-quadratic worst-case transformations, we naturally restrict our attention to centralised transformations. We generally allow the transformations to break connectivity, even though we also develop some connectivity-preserving transformations on the way. Our main result is a universal transformation of $ O(n \log n) $ worst-case running time that is permitted to break connectivity. Distributed transformations and connectivity-preserving universal transformations are left as interesting future research directions.

\subsection{Our Approach}
\label{subsec:approach}

In \cite{MICHAIL2019}, it was proved that if the devices (called \emph{nodes} from now on) are equipped only with a rotation mechanism, then the decision problem of transforming a connected shape $A$ into a connected shape $B$ is in $\mathbf{P}$, and a constructive characterisation of the (rich) class of pairs of shapes that are transformable to each other was given. In the case of combined availability of rotation and sliding, universality has been shown \cite{DP04,MICHAIL2019}, that is, any pair of connected shapes are transformable to each other. Still, in these and related models, where in any time step at most one node can move a single position in its local neighbourhood, it can be proved (see, for instance, \cite{MICHAIL2019}) that there will be pairs of shapes that require $\Omega(n^2)$ steps to be transformed to each other. This follows directly from the inherent ``distance'' between the two shapes and the fact that this distance can be reduced by only a constant in every time step. An immediate question is then \emph{``How can we come up with more efficient transformations?''} 

Two main alternatives have been explored in the literature in an attempt to answer this question. One is to consider parallel time, meaning that the transformation algorithm can move more than one node (up to a linear number of nodes if possible) in a single time step. This is particularly natural and fair for distributed transformations, as it allows all nodes to have their chances to take a movement in every given time-step. For example, such transformations based on pipelining \cite{DSY04b, MICHAIL2019}, where essentially the shape transforms by moving nodes in parallel around its perimeter, can be shown to require $O(n)$ parallel time in the worst case and this technique has also been applied in systems (e.g., \cite{RCN14}).

The other approach is to consider more powerful actuation mechanisms, that have the potential to reduce the inherent distance faster than a constant per sequential time-step. These are typically mechanisms where the local actuation has strength higher than a constant. This is different than the above parallel-time transformations, in which local actuation can only move a single node one position in its local neighbourhood and the combined effect of many such movements at the same time is exploited. In contrast, in higher-strength mechanisms, it is a single actuation that has enough strength to move many nodes at the same time. Prominent examples in the literature are the linear-strength models of Aloupis \emph{et al.} \cite{ABD13,ACD08}, in which nodes are equipped with extend/contract arms, each having the strength to extend/contract the whole shape as a result of applying such an operation to one of its neighbours and of Woods \emph{et al.} \cite{WCG13}, in which a whole line of nodes can rotate around a single node (acting as a linear-strength rotating arm). The present paper follows this approach, by introducing and investigating a linear-strength model in which a node can push a line of consecutive nodes one position (towards an empty cell) in a single time-step. 

In terms of transformability, our model can easily simulate the combined rotation and sliding mechanisms of \cite{DP04,MICHAIL2019} by restricting movements to lines of length 1 (i.e., individual nodes). It follows that this model is also capable of universal transformations, with a time complexity at most twice the worst-case of those models, i.e., again $O(n^2)$. Naturally, our focus is set on exploring ways to exploit the parallelism inherent in moving lines of larger length in order to speed-up transformations and, if possible, to come up with a more efficient in the worst case universal transformation.

As reversibility of movements is still valid for any line movement in our model, we adopt the approach of transforming any given shape $A$ into a spanning line $L$ (vertical or horizontal). This is convenient, because if one shows that any shape $A$ can transform fast into a line $L$, then any pair of shapes $A$ and $B$ can then be transformed fast to each other by first transforming fast $A$ into $L$ and then $L$ into $B$ by reversing the fast transformation of $B$ into $L$.

We start this investigation by identifying the diagonal shape $D$ (which is considered connected in our model and is very similar to the staircase worst-case shape of \cite{MICHAIL2019}) as a potential worst-case initial shape to be transformed into a line $L$. This intuition is supported by the $O(n^2)$ individual node distance between the two shapes and by the initial unavailability of long lines: the transformation may move long lines whenever available, but has to pay first a number of movements of small lines in order to construct longer lines. In this benchmark (special) case, the trivial lower and upper bounds $\Omega(n)$ and $O(n^2)$, respectively, hold. Moreover, observe that a sequential gathering of the nodes starting from the top right and collecting the nodes one after the other into a snake-like line of increasing length is still quadratic, because, essentially, for each sub-trip from one collection to the next, the line has to make a ``turn'', meaning to change both a row and a column, and in this model this costs a number of steps equal to the length of the line, that is, roughly, $1+2+\ldots+(n-1)=\Theta(n^2)$ total time.  

We first prove that by partitioning the diagonal into $\sqrt{n}$ diagonal segments of length $\sqrt{n}$ each, we can first transform each segment in time quadratic in its length into a straight line segment, then push all segments down to a ``collection row'' $y_0$ in time $O(n\sqrt{n})$ and finally re-orient all line segments to form a horizontal line in $y_0$, paying a linear additive factor. Thus, this transformation takes total time $O(n\sqrt{n})$, which constitutes our first improvement compared to the $\Omega(n^2)$ lower bound of \cite{MICHAIL2019}. We then take this algorithmic idea one step further, by developing two transformations building upon it, that can achieve the same time-bound while \emph{preserving connectivity} throughout their course: one is based on \emph{folding} segments and the other on \emph{extending} them.

As the $O(\sqrt{n})$ length of uniform partitioning into segments is optimal for the above type of transformation, we turn our attention into different approaches, aiming at further reducing the running time of transformations. Allowing once more to break connectivity, we develop an alternative transformation based on \emph{successive doubling}. The partitioning is again uniform for individual ``phases'', but different phases have different partitioning length. The transformation starts from a minimal partitioning into $n/2$ lines of length 2, then matches them to the closest neighbours via shortest paths to obtain a partitioning into $n/4$ lines of length $4$, and, continuing in the same way for $\log n$ phases, it maintains the invariant of having $n/2^i$ individual lines in each phase $i$, for $1\leq i\leq \log n$. By proving that the cost of pairwise merging through shortest paths in each phase is linear in $n$, we obtain that this approach transforms the diagonal into a line in time $O(n\log n)$, thus yielding a substantial improvement. 
Observe that the problem of transforming the diagonal into a line seems to involve solving the same problem into smaller diagonal segments (in order to transform those into corresponding line segments). Then, one may naturally wonder whether a recursive approach could be applied in order to further reduce the running time. We provide a negative answer to this, for the special case of uniform recursion and at the same time obtain an alternative $O(n\log n)$ transformation for the diagonal-to-line problem.

Our final aim is on attempting to generalise the ideas developed for the above benchmark case in order to come up with \emph{equally efficient universal transformations}. We successfully generalise both the $O(n\sqrt{n})$ and the $O(n\log n)$ approaches, obtaining universal transformations of worst-case running times $O(n\sqrt{n})$ and $O(n\log n)$, respectively. We achieve this by enclosing the initial shape into a square bounding box and then subdividing the box into square sub-boxes of appropriate dimension. For the $O(n\sqrt{n})$ bound, a single such partitioning into sub-boxes of dimension $\sqrt{n}$ turns out to be sufficient. For the $O(n\log n)$ bound we again employ a successive doubling approach through phases of an increasing dimension of the sub-boxes, that is, through a new partitioning in each phase. Therefore, our ultimate theorem (followed by a constructive proof, providing the claimed transformation) states that: \emph{``In this model, when connectivity need not necessarily be preserved during the transformation, any pair of connected shapes $A$ and $B$ can be transformed to each other in sequential time $O(n\log n)$''}. 

Table \ref{tab:AT} summarises the running times of all the transformations developed in this paper.

\begin{table}[h]
	\normalsize
	\setlength{\tabcolsep}{13pt}
	\begin{center}
		\begin{tabular}{l  l   c c  }
			\hline
			Transformation  & Problem & Running Time & Lower Bound \\ \hline
			\emph{DL-Partitioning}  & Diagonal & $O(n\sqrt{n})$ & $\Omega(n)$ \\
			\emph{DL-Doubling}  & Diagonal & $O(n \log n)$  & $\Omega(n)$ \\
	        \emph{DL-Recursion}  & Diagonal & $O(n \log n)$  & $\Omega(n)$ \\			
			\emph{DLC-Folding}  & Diagonal Connected & $O(n\sqrt{n})$  & $\Omega(n)$ \\ 
			\emph{DLC-Extending}  & Diagonal Connected & $O(n\sqrt{n})$ & $\Omega(n)$ \\
			\emph{U-Box-Partitioning}  & Universal & $O(n\sqrt{n})$ & $\Omega(n)$ \\ 
			\emph{U-Box-Doubling}  & Universal & $O(n \log n)$ & $\Omega(n)$ \\ \hline
		\end{tabular}
	\end{center}
	\caption{A summary of our transformations and their corresponding worst-case running times (the trivial lower bound is in all cases $\Omega(n)$). The Diagonal, Diagonal Connected, and Universal problems correspond to the {\sc DiagonalToLine}, {\sc DiagonalToLineConnected}, and {\sc UniversalTransformation} problems, respectively (being formally defined in Section \ref{subsec:problems}).}
	\label{tab:AT}
\end{table}

Section \ref{sec:prel} brings together all definitions and basic facts that are used throughout the paper. In Section \ref{sec:Transforming_the_Diagonal}, we study the problem of transforming a diagonal shape into a line, without and with connectivity preservation. Section \ref{sec:universal} presents our universal transformations. 
In Section \ref{sec:conclusions}, we conclude and discuss further research directions that are opened by our work.

\section{Preliminaries and Definitions}
\label{sec:prel}

The transformations considered here run on a two-dimensional square grid. Each cell of the grid possesses a unique location addressed by non-negative coordinates $(x,y)$, where $x$ denotes columns and $y$ indicates rows. A \emph{shape} $S$ is a set of $n$ \emph{nodes} on the grid, where each individual node $u \in S$ occupies a single cell $cell(u)  = (x_{u}, y_{u})$, therefore we may also refer to a node by the coordinates of the cell that it occupies at a given time. 
Two distinct nodes $(x_1,y_1)$, $(x_2,y_2)$ are \emph{neighbours} (or \emph{adjacent}) iff $x_2-1 \leq x_1\leq x_2+ 1$ and $y_2-1 \leq y_1\leq y_2+ 1$ (i.e., their cells are adjacent vertically, horizontally or diagonally). A shape $S$ is \emph{connected} iff the graph defined by $S$ and the above neighbouring relation on $S$ is connected.
Throughout, $n$ denotes the number of nodes in a shape under consideration.

A line, $ L \subseteq S $, is defined by one or more consecutive nodes in a column or row. That is, $L = (x_{0},y_{0}), (x_{1},y_{1}),\ldots ,(x_{k},y_{k})$, for $ 0 \leq k \leq n, \: k \in \mathbb{Z}$, is a line iff $ x_{0}= x_1 =\cdots = x_{k} $ and $ |y_{k}-y_{0}| = k $, or $ y_{0}=y_1=\cdots=y_{k} $ and $ |x_{k}-x_{0}| = k $.
A \emph{line move}, is an operation by which all nodes of a line $L$ move together in a single step, towards an empty \emph{cell} adjacent to one of $L$'s endpoints. A line move may also be referred to as \emph{step} (or \emph{move} or \emph{movement}) and time is discrete and measured in number of steps throughout. 
A move in this model is equivalent to choosing a node $u$ and a direction $d\in \{up, down, left, right\}$ and moving $u$ one position in direction $d$. This will additionally push by one position the whole line $L$ of nodes in direction $d$, $L$ (possibly empty) starting from a neighbour of $u$ in $d$ and ending at the first empty cell.  More formally and in slightly different terms: A line $ L = (x_{1},y), (x_{2},y) , \ldots , (x_{k},y) $ of length $k$, where $1 \leq k \le n $, can push all $k$ nodes rightwards in a single step to positions $ (x_{2},y), (x_{3},y) , \ldots , (x_{k+1},y) $ iff there exists an empty cell to the right of $L$ at $(x_{k+1}, y)$. The ``down'', ``left'', and ``up'' movements are defined symmetrically, by rotating the whole system 90\degree, 180\degree, and 270\degree\  clockwise, respectively.

\begin{definition}[A permissible line move]\label{def:permissible_line_move} 
Let $ L $ be a line of $k$ nodes, where $1 \leq k \le n $. Then, $L$ can move as follows (depending on its original orientation, i.e. horizontal or vertical):
\begin{enumerate}
	\item \emph{Horizontal.} Can push all $k$ nodes rightwards in a single step from $ (x_{1},y), (x_{2},y) , \ldots ,$ $(x_{k},y) $ to positions $ (x_{2},y), (x_{3},y) , \ldots , (x_{k+1},y) $ iff there exists an empty cell to the right of $L$ at $(x_{k+1}, y)$. Similarly, it can push all $k$ nodes to the left to occupy $ (x_{0},y), (x_{1},y) , \ldots ,$ $(x_{k-1},y) $, iff there exists an empty cell at $(x_{0}, y)$.
	
	\item \emph{Vertical.} Can push all $k$ nodes upwards in a single step from $ (x,y_{1}), (x,y_{2}) , \ldots ,$ $(x,y_{k}) $ into $ (x,y_{2}), (x,y_{3}) , \ldots , (x,y_{k+1}) $, iff there exists an empty cell above $L$ at $(x, y_{k+1})$. Similarly, it can push all $k$ nodes down to occupy $ (x,y_{0}), (x,y_{1}) , \ldots , (x,y_{k-1}) $, iff there exists an empty cell below $L$ at $(x, y_0)$.
\end{enumerate}	

\end{definition}

The following definitions from \cite{MICHAIL2019} shall be useful for our study. Let us first agree that we colour black any cell occupied by a node (as in Figure \ref{fig:Definitione_Colored_Sahpe_}).

\begin{definition}\label{def:Hole}
	A hole $ H $ is a set of empty cells that are unoccupied and enclosed by nodes $ u \in S $, where $H = \{h_{i} \mid h_{i} = (x_{i}, y_{i})  ,\ (x_{i}, y_{i}) \in \mathbb{Z}^{2} , i \geq 0  \}$, such that every infinite spanning single path starting from the hole $h \in H$ and moving, only vertically and horizontally, certainly passes through a black cell of a node $ u \in S$. 		    		 
\end{definition}

\begin{definition}\label{def:Compact_Connected_Shapes} 
	A compact shape $ S' $ contains all cells of $ S $ and the hole $H$, such that $ S' $ has no holes,  $ S' = S \cup H $. For instance, the black and grey cells in Figure~\ref{fig:Definitione_Colored_Sahpe_} are forming a connected compact shape $S'$.		   
\end{definition}  

\begin{definition} \label{def:Perimeter} 
	A perimeter (border) of $ S $ is defined as a polygon of unit length line segments, which surrounds the minimum-area of the interior of $ S' $, such as the blue line in Figure~\ref{fig:Definitione_Colored_Sahpe_}. 			
\end{definition}

\begin{definition} \label{def:Surrounded_Shape} 
	A surrounded layer of $ S $ is consisting of cells that are not occupied by nodes of $S$ and contribute to the perimeter by at least one of its sides or corners (e,g the red cells in Figure~\ref{fig:Definitione_Colored_Sahpe_}). Normally, each cell in a 2D grid owns four line-segment sides and four corners.     
\end{definition}

\begin{definition} \label{def:External_Surface_of_Connected_Shape} 
	The external surface of $S$ is another shape $W$, which includes all cells occupied by nodes of $S$ and adjacent to the perimeter vertically or horizontally. The external surface of $S$ is also a connected shape itself, and that is proved in \cite{MICHAIL2019}.           
\end{definition}
\begin{figure}[h!t] 	
	\centering
	\includegraphics[angle=90, scale=0.5]{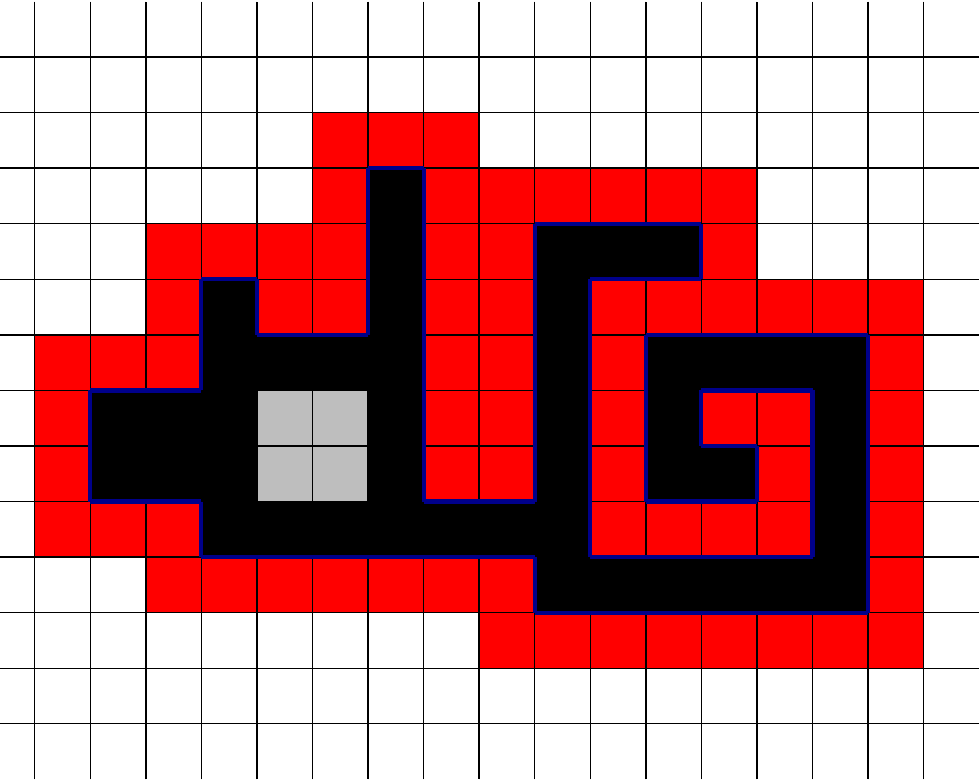}
	\caption{ All nodes of $S$ occupy the black cells, while the grey cells define a hole. The blue line depicts the perimeter of $S$, and the surrounded layer is indicated by the red cells.}
	\label{fig:Definitione_Colored_Sahpe_}	    	
\end{figure} 

\begin{proposition} \label{prop:Connectivity_Of_Surrounded_Layer} 
	The surrounded layer of any connected shape $S$ is always another connected shape.           	
\end{proposition}
\begin{proof}
	It follows from \cite{MICHAIL2019}. Since $S$ is connected, then the perimeter of $S$ is connected too, and hence, forms a cycle. Each segment of the perimeter is contributed by two cells, belonging to the external surface and the surrounded layer. Now, if one walks on the perimeter (vertically or horizontally) or turns (left or right) clockwise or anticlockwise at any segment, one of the following cases will occur: 
	\begin{bracketenumerate}
		\item Pass through two adjacent vertical or horizontal cells on the surrounded layer and the external surface of $ S $.
		\item Stay put at the same position (cell) on the external surface and move through three neighbouring cells connected perpendicularly on the surrounded layer of $ S $.
		\item Stay put at the same position (cell) on the surrounded layer and pass through three neighbouring cells connected perpendicularly on the external surface of $ S $.
		\item Stay put at the same position (cell) on the surrounded layer and pass through two neighbouring cells connected diagonally on the external surface of $ S $ and one cell of the surrounded layer.
	\end{bracketenumerate}     	
	Subsequently, all cases above preserve connectedness of the surrounded layer and the external surface of $ S $. 
\end{proof}  

As already mentioned, we know that there are related settings in which any pair of connected shapes $A$ and $B$ of the same order (``order'' of a shape $S$ meaning the number of nodes of $S$ throughout the paper) can be transformed to each other\footnote{We also use $A\rightarrow B$ to denote that shape $A$ can be transformed to shape $B$.} while preserving the connectivity throughout the course of the transformation.\footnote{In this paper, whenever transforming into a target shape $B$, we allow any placement of $B$ on the grid, i.e., any shape $B^\prime$ obtained from $B$ through a sequence of rotations and translations.} This, for example, has been proved for the case in which the available movements to the nodes are rotation and sliding \cite{DP04,MICHAIL2019}. 
We now show that the model of \cite{DP04,MICHAIL2019} is a special case of our modeld, including universal transformations, are also valid transformations in the present model.

\begin{proposition} \label{Simulation} 
The rotation and sliding model of \cite{DP04,MICHAIL2019} is a special case of the present model.    
\end{proposition}
\begin{proof} 
We establish a technique to prove that our model is capable of simulating \emph{rotation and sliding models} of a two-dimension square grid appeared in \cite{DP04,MICHAIL2019}. First, the sliding operation is equivalent in all those models, that is, if a node  $u$ is located at a cell of the grid, $u = (x,y)$, then $u$ can slid right to any empty cell at $ (x+1, y+1)$ over a horizontal line of length 2, such as in Figure \ref{fig:Simulation} (a). The ``down'', ``left'', and ``up'' movements are defined symmetrically, by rotating the whole system 90\degree, 180\degree, and 270\degree\  clockwise, respectively. By Definition \ref{def:permissible_line_move}, our model is exploited \textit{sliding} rule to push a line of 1 node into an adjacent empty cell, whether vertically or horizontally. For rotation, all mentioned models perform a single operation to rotate a node $u_{1} = (x,y)$ around another $u_{2} = (x,y-1)$ by a $90\degree$ clockwise iff there exits two empty cells at $ (x+1,y) $ and $ (x+1,y+1) $, see Figure \ref{fig:Simulation}(b). Analogously, this holds for all possible rotations by again rotating the whole system 90\degree, 180\degree, and 270\degree\  clockwise, respectively. Still, the rotation mechanism is also adopted by this model following Definition \ref{def:permissible_line_move}, and actually it costs twice for a single rotation to take place, compared with others. Subsequently, it implies that all transformations established there (with their running time at most doubled, including universal transformations, are also valid transformations in the present model).
\begin{figure}
		\centering
		\subcaptionbox{}
		{\includegraphics[scale=0.7]{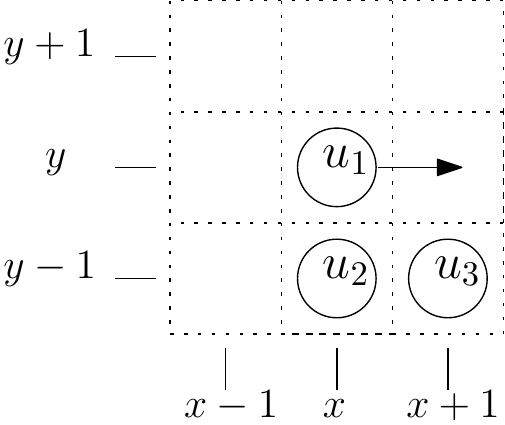}}	 \qquad \qquad \qquad
		\subcaptionbox{}
		{\includegraphics[scale=0.7]{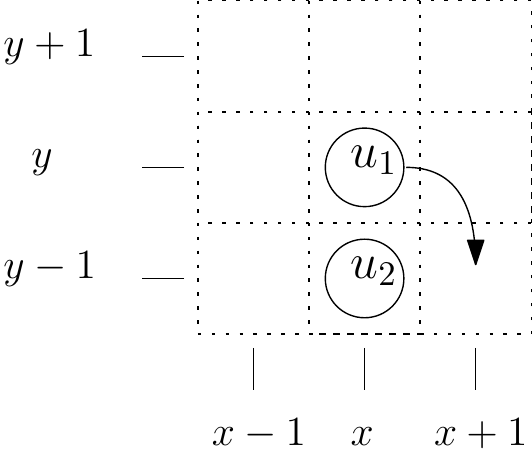}}	
		\caption{(a) An example of sliding $u_1$ over $u_2$ and $u_3$ to an empty cell to the left. (b) Rotate $u_{1}$ a $90\degree$ clockwise around $u_{2}$. }
		\label{fig:Simulation}	
	\end{figure}    
 
\end{proof}  
  Consider a line $ L $ of $ k $ nodes occupying $(x_{1},y), (x_{2},y), \ldots , (x_{k},y) $ and empty cells at $(x_{k+1},y), $ $ (x_{k+2},y), \ldots , (x_{k+o},y) $ for all $ o $, where $k=o \ge 1$.  Now, we aim to transfer all $k$ nodes to the right to fill in all $k$ empty cells. As usual, any transformation of a single motion would move all $k$ nodes in a total of $O(k^{2})$ moves, since each $k$ transfers a distance of $\Delta = k$. On the contrary, the $k$ nodes move altogether in parallel $\Delta = k$ distance to occupy the empty cells in a total of $O(k)$ steps, by exploiting the linear-strength of this model. Now, we are ready to show a more beneficial property in the following lemma:

\begin{lemma} \label{lem:Transfer_Line_H_to_V}
	The minimum number of line moves by which a line of length $ k $, $1\leq k \leq n$, can completely change its orientation\footnote{From vertical to horizontal and vice versa.}, is $2k-2$. 
\end{lemma}
\begin{proof}
	Assume a 2D square grid contains only a line $L_{1}$ of $ k $ nodes at $(x_{1},y_{1}), \ldots , (x_{1},y_{k})$ for all $k$, where $ k \geqslant 1 $, and empty cells on $(x_{2},y_{1}), \ldots , (x_{k},y_{1})$, as depicted in Figure \ref{fig:Exam_Line_Move}. Now, the first bottommost node, $ u_{1} \in L_{1} $, moves one step right to occupy an empty cell on $ (x_{1},y_{1}) $. Consequently, a new empty cell is created at $ (x_{1},y_{1})$, and  the length of $L_{1}$ decreases by $ k - 1$. By the linear-strength, $L_{1}$ pushes all $k-1$ nodes down altogether in parallel in a single-time-move to occupy $ (x_{1},y_{1}), \ldots , (x_{1},y_{k-1})  $. Hence, two line steps have been performed to move $u_{1}$ and push down all $k-1$ nodes. Therefore, it implies that each of $k-1 \in L_{1} $ requires two steps to transfer into the bottommost row, $y_1$. Thus, we conclude that any line of length ($k$ nodes) requires at most twice of its length, $ 2k - 2 $, to change its orientation, which is bounded above by $O(k)$. 
	
	For the matching lower bound consider w.l.o.g. a horizontal line. A complete change of orientation requires the nodes of the line to occupy $k$ consecutive rows (from just one originally). The bound follows by observing that in any two consecutive steps, at most one new row can become occupied.       
\end{proof}

\begin{figure}[h!t]	
	\centering
	\includegraphics[scale=0.7]{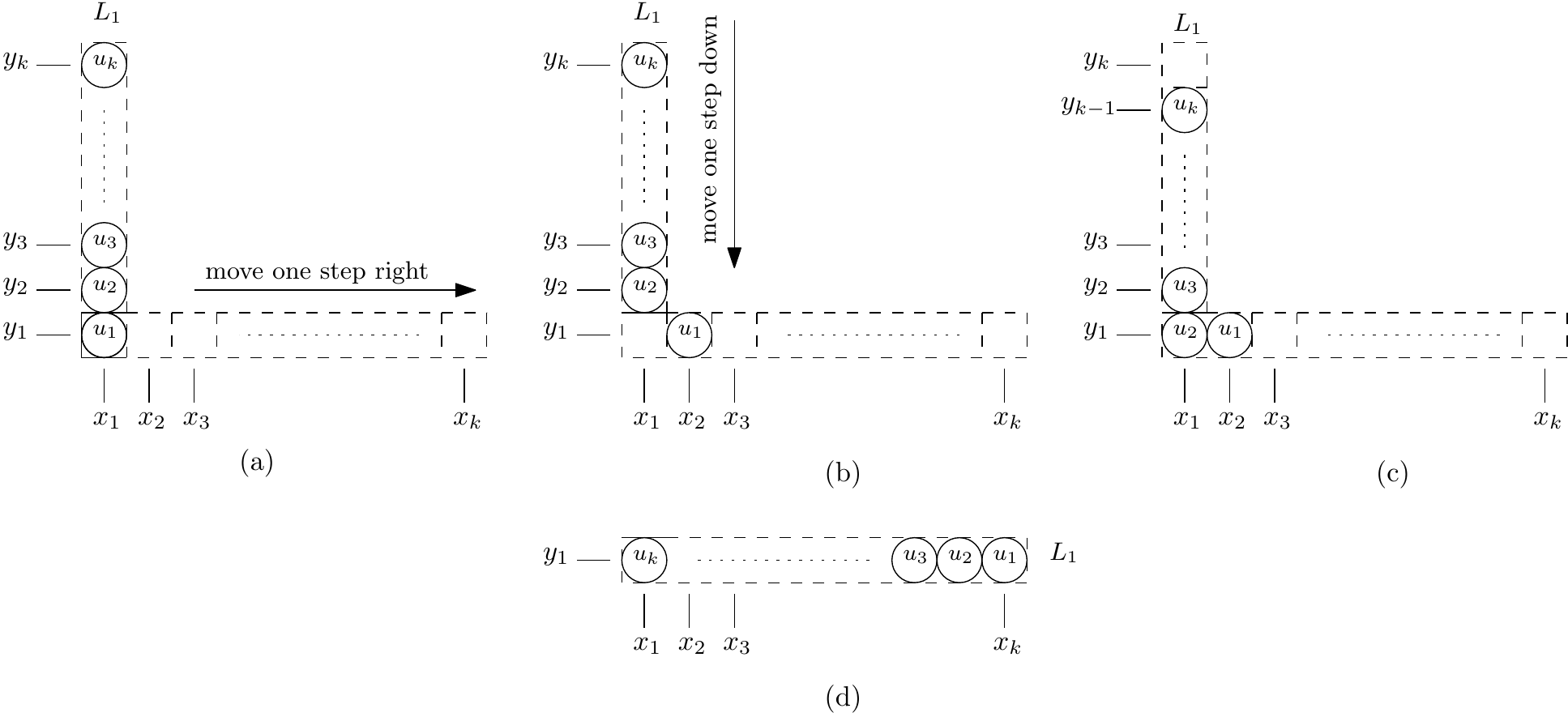}
	\caption{A line of $k$ nodes changes orientation by two consecutive steps per node. (a) Move $u_{1}$ one step to the right. (b) and (c) All $k-1$ nodes push down altogether in single step. In (d), the line has finally transformed from vertical to horizontal after $2k$ steps.}   
	\label{fig:Exam_Line_Move}	         
\end{figure} 

A property that typically facilitates the development of universal transformations, is reversibility of movements. To this end, we next show that line movements are reversible.    

\begin{lemma} [Reversibility] 	\label{lem:Transferability_LineMove}   
	Let $(S_{I}, S_{F})$ be a pair of connected shapes of the same number of nodes $ n $. If  $S_{I} \rightarrow  S_{F} $ (``$\rightarrow$'' denoting ``can be transformed to
	via a sequence of line movements'') then $ S_{F} \rightarrow S_{I} $.
\end{lemma} 
\begin{proof}
	First, we should prove that each single line step is reversible. Figure~\ref{fig:Transferability_LineMove_LineMove} (a) shows a simple connected shape of four nodes forming two horizontal and one vertical lines at $ L_1 =\{ u_{1}, u_{2}\}$, $ L_2 =\{ u_{3}, u_{4}\}$ and $ L_3 =\{ u_{2}, u_{3}\}$, respectively. Assume this configuration has no more space to the left, beyond the dashed line; therefore, $ L_1 $ moves one step to occupy the empty cell $(i+2,j+1)$. Now, all $ L_1 $ nodes are moving altogether to fill in positions $ (x_{i+1},y_{j+1}) $ and $ (x_{i+2},y_{j+1}) $, as depicted in Figure \ref{fig:Transferability_LineMove_LineMove} (b). Consequently, the previous line step creates another empty cell at $ (x_{i},y_{j+1}) $, which then gives the ability to $ L_1 $ to move back reversibly to fill in this new empty cell on the left. Thus, we conclude that any single line step is reversible, which implies that any finite sequence of line steps are reversible too.
	 
	Still, reversibility is valid for any line movement in our model by transforming any given shape $ S_{I} $ into a line $ S_{L} $ (vertical or horizontal). This is sufficient because if any shape $ S_{I} $ transforms into a line $ S_{L}$, then any pair of shapes $ S_{I} $ and $ S_{F} $ can then be transformed to each other, via an intermediate connected shape $S_{L} $, by first transforming  $ S_{I} $ into $ S_{L} $ and then  $S_{L} $ to  $S_{F} $, by reversing the transformation of  $S_{F} $ to $S_{L} $.	
	\begin{figure}
		\centering
		\subcaptionbox{}
		{\includegraphics[scale=0.8]{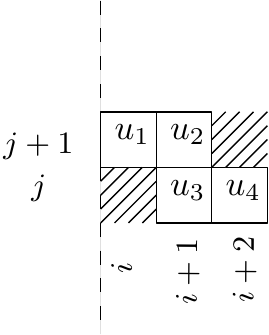}}	  \qquad \qquad
		\subcaptionbox{}
		{\includegraphics[scale=0.8]{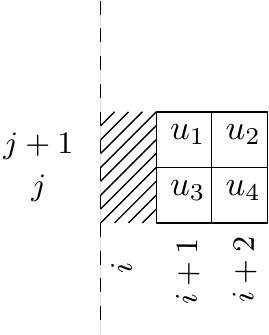}}
		\caption{An example of a reversible line step.}
		\label{fig:Transferability_LineMove_LineMove}	
	\end{figure}
\end{proof}

\subsection{Nice shapes}
\label{subsecNiceShapes}

A family of shapes, denoted $ NICE $, is introduced in this study to probably act as efficient intermediate shapes of the transformation. A \emph{nice} shape is, informally, any connected shape that contains a particular line called the \emph{central line} (denoted $L_{C}$). Intuitively, one may think of $L_{C}$ as a supporting (say horizontal) base of the shape, where each node $u$ not on $L_C$ must be connected to $L_C$ through a vertical line.
\begin{definition}[Nice Shape] \label{def:nice_shape}
		A connected shape $S \in NICE $ if there exists a central line $ L_{C}\subseteq S$, such that every node $ u \in S\setminus L_{C}$ is connected to $  L_{C} $ via a line perpendicular to $L_{C}$ (Figure~\ref{fig:Nice_Shapes} shows some examples of a \textit{nice} shape  and non \textit{nice} shapes). 
\end{definition}

\begin{figure}[h!t]	
	\centering
	\includegraphics[scale=0.3]{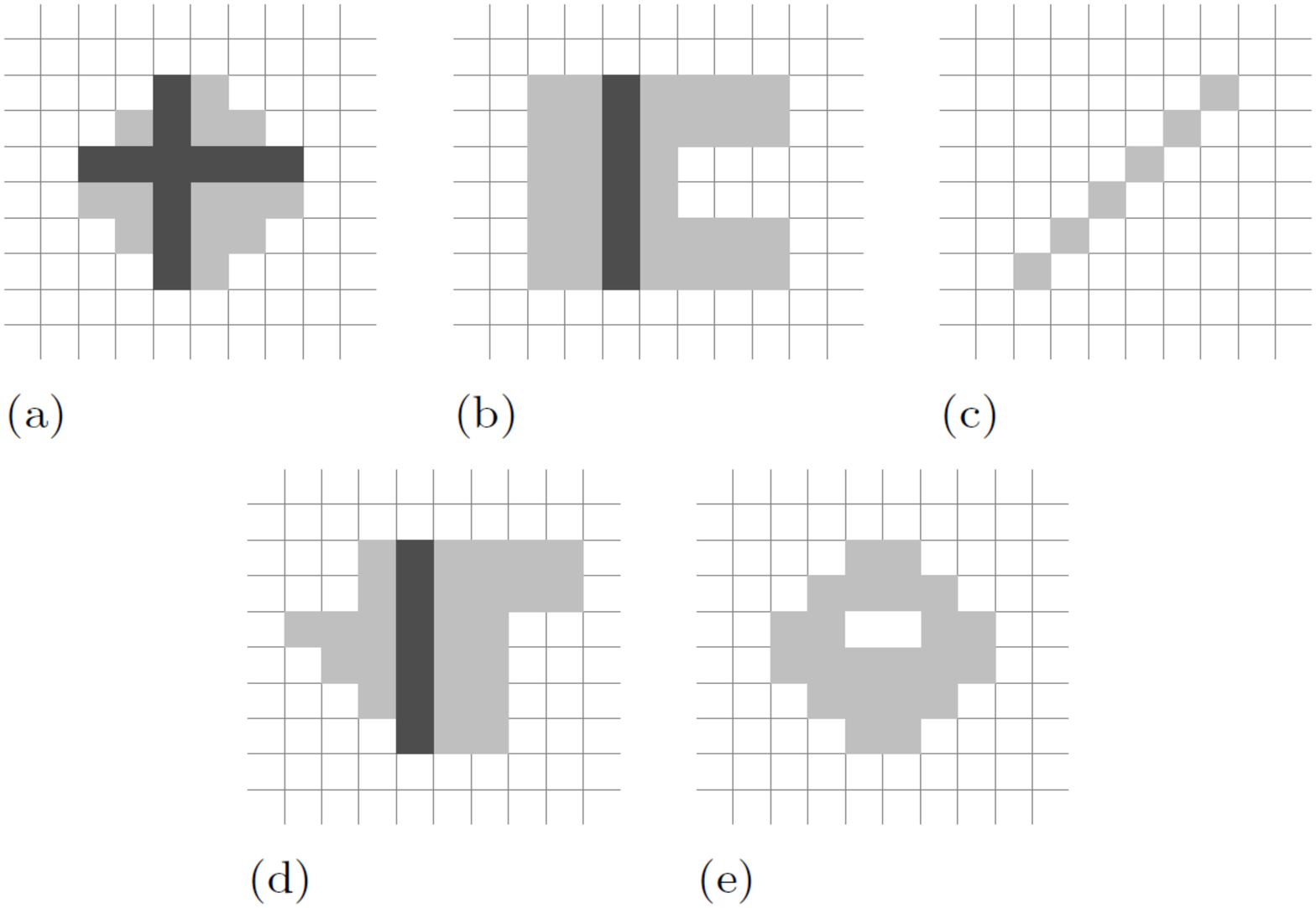}
	\caption{ The central line $L_C$ occupies black cells of \textit{nice} shapes in (a), (b) and (d). In (c), the shape is not \textit{nice} (due to the lack of $L_{C}$). (e) is also not \textit{nice} because of the hole, white cells inside the shape, which prevents the formation of $L_{C}$.}   
	\label{fig:Nice_Shapes}	         
\end{figure}  

By reversibility (Lemma \ref{lem:Transferability_LineMove}), a \textit{nice} shape is exploited as an intermediate shape that simplifies the universal transformations.  The following proposition shows that: 

\begin{proposition}  \label{prop:Niceshape_to_line}
	Let $S_{Nice}$  be a \textit{nice} shape and $S_{L}$ a straight line, both of the same order $n$. Then $S_{Nice} \rightarrow S_{L} $ (and $S_{L} \rightarrow S_{Nice} $) in $O (n)$ steps.
\end{proposition} 
\begin{proof}
	By Definition \ref{def:nice_shape}, any $S_{Nice} $ of $ n $ nodes must have a central line $L_{C}$ of $k$ nodes, $1 \le k \le n$, where other $n -k$ nodes are connecting to $L_{C}$ via a line. By Lemma \ref{lem:Transfer_Line_H_to_V}, each of the  $n -k$ nodes requires at most two line steps to be included into $L_{C}$ in a  maximum total of $2(n -k)$ steps for all $n -k$. Then, $S_{Nice} $ requires at most $O(2n -2k)$ steps to transform into $ S_{L} $. As a consequence, the connected  pair of shapes  $(S_{Nice}, S_{L}) $ of the same order are transformable to each other by Lemma \ref{lem:Transferability_LineMove}, in $O(n)$ steps. 
\end{proof}

\subsection{Problem Definitions}
\label{subsec:problems}

We now formally define the problems to be considered in this paper.\\

\noindent\textbf{{\sc DiagonalToLine}.} Given an initial connected diagonal line $S_{D}$ and a target vertical or horizontal connected spanning line $S_{L}$ of the same order, transform  $S_{D}$  into $S_{L}$, without necessarily preserving the connectivity during the transformation. \\

\noindent\textbf{{\sc DiagonalToLineConnected}.} Restricted version of {\sc DiagonalToLine} in which connectivity must be preserved during the transformation. \\

\noindent\textbf{{\sc UniversalTransformation}.} Give a general transformation, such that, for all pairs of shapes $(S_I,S_F)$ of the same order, where $S_I$ is the initial shape and $S_F$ the target shape, it will transform $S_I$ into $S_F$, without necessarily preserving connectivity during its course. \\

\section{Transforming the Diagonal into a Line}
\label{sec:Transforming_the_Diagonal}

We begin our study from the case in which the initial shape is a diagonal line $S_{D}$ of order $ n $. Our goal throughout the section is to transform $S_{D}$ into a spanning line $S_{L}$, i.e.,f= solve the {\sc DiagonalToLine} and/or {\sc DiagonalToLineConnected} problems. We do this, because these problems seem to capture the worst-case complexity of transformations in this model.  

Consider a $S_{D}$ of $ n $ nodes occupying $ (x_{1},y_{1}), (x_{2},y_{2}), \ldots, (x_{n},y_{n}) $, such as the diagonal line of 6 nodes depicted in Figure~\ref{fig:Nice_Shapes} (c). Observe that the diagonal comprises some special properties which cannot be found in other connected shapes, that is, each single node of $S_{D}$ enjoys a unique $ x $- and $ y $- axis, in other words $n = \#rows = \# columns$. Below, we give three $O(n \sqrt{n})$-time strategies to transform the diagonal into a line. In Section \ref{subsec:DL-Transformability}, the algorithm allows the shape to break its connectivity  throughout transformation, whilst connectedness is preserved in Sections \ref{subsec:DLC-Folding} and \ref{subsec:nrootn_S2}.
 
\subsection{An $O(n\sqrt{n})$-time Transformation}
\label{subsec:DL-Transformability}

We start from {\sc DiagonalToLine} (i.e., no requirement to preserve connectivity). Our strategy is as follows. We divide the diagonal into several segments, as in Figure~\ref{fig:S1_nrootn_partition} (a). Then in each segment, we perform a trivial (inefficient, but enough for our purposes) line formation by moving each node independently to the leftmost column in that segment (Figure~\ref{fig:S1_nrootn_partition} (b)). Therefore, all segments are transformed into lines (Figure~\ref{fig:S1_nrootn_partition} (c)). Then, we transfer each line segment all the way down to the bottommost row of the diagonal $S_{D}$ (Figure~\ref{fig:S1_nrootn_partition} (d)). Finally, we change the orientation of all line segments to form the target spanning line (Figure~\ref{fig:S1_nrootn_partition} (e)).
 
\begin{figure}
	\centering
	\subcaptionbox{}
	{\includegraphics[scale=0.55]{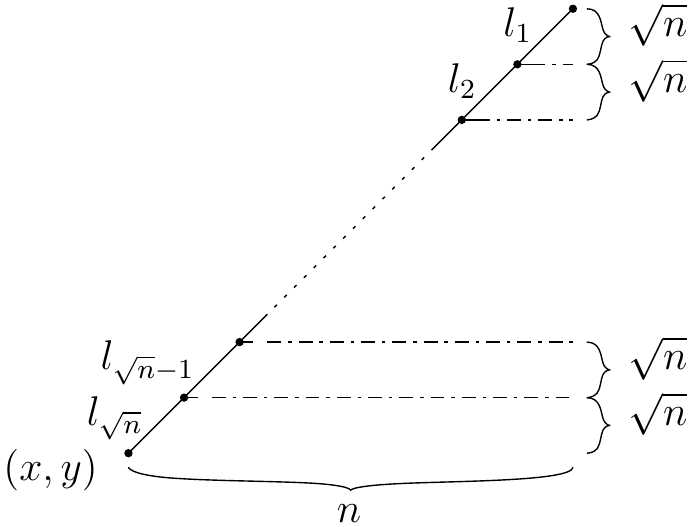}}	 \qquad
	\subcaptionbox{}
	{\includegraphics[scale=0.55]{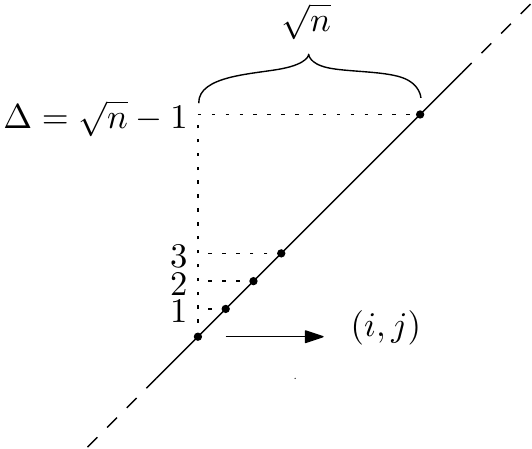}}	\qquad
	\subcaptionbox{}
	{\includegraphics[scale=0.55]{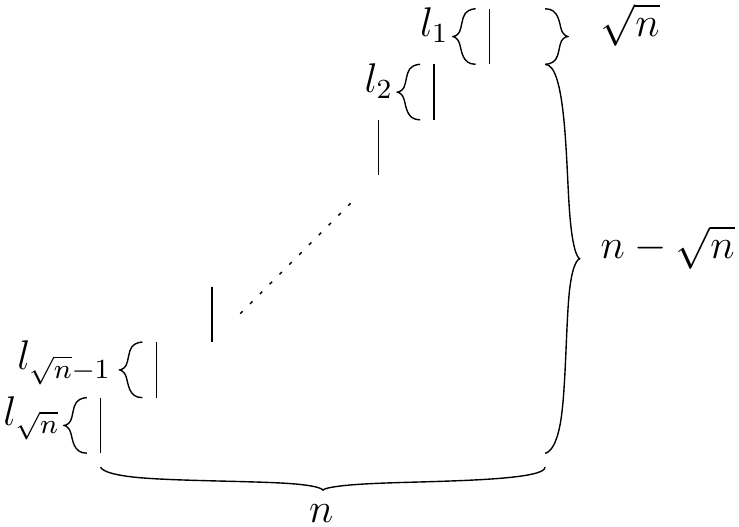}}	\\
	\subcaptionbox{}
	{\includegraphics[scale=0.60]{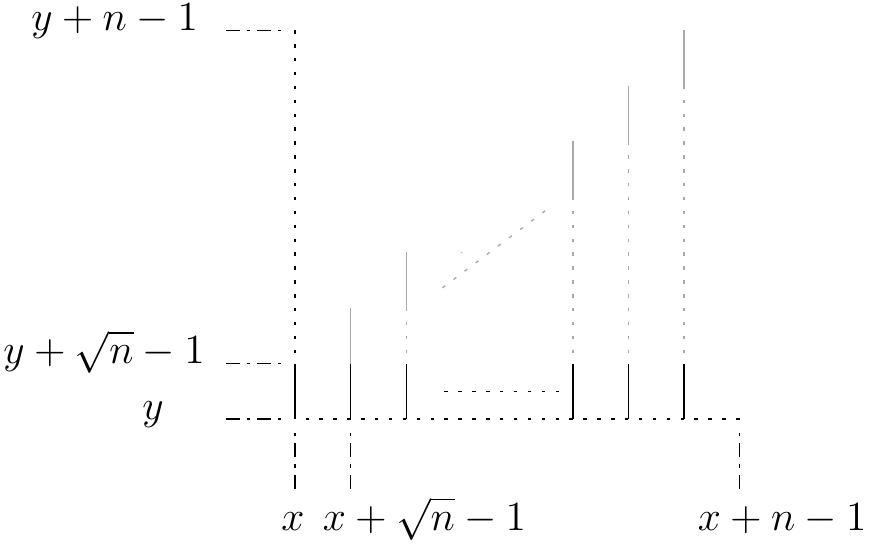}}  \qquad  \quad
	\subcaptionbox{}
	{\includegraphics[scale=0.60]{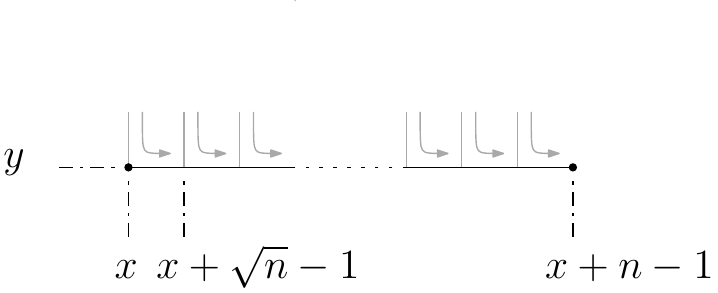}}
	\caption{(a) Dividing the diagonal into $\sqrt{n}$ segments of length $\sqrt{n}$ each (integer $\sqrt{n}$ case). (b) A closer view of a single segment, where $1,2,3, \ldots , \sqrt{n} -1$ are the required distances for the nodes to form a line segment at the leftmost column (of the segment). (c) Each line segment is transformed into a line and transferred towards the bottommost row of the shape, ending up as in (d). (e) All line segments are turned into the bottommost row to form the target spanning line.}
	\label{fig:S1_nrootn_partition}	
\end{figure}

More formally, let $S_{D}$ be a diagonal, occupying  $ (x,y), (x+1,y+1), \ldots, (x+n-1,y+n-1) $, such that $x$ and $y$ are the leftmost column and the bottommost row of $S_{D}$, respectively. $S_{D}$ is divided into $ \ceil{\sqrt{n}} $ segments, $l_{1}, l_{2}, \ldots , l_{\ceil{\sqrt{n}}} $, each of length $ \floor{\sqrt{n}} $, apart possibly from a single smaller one. Figure~\ref{fig:S1_nrootn_partition} (a) illustrates the case of integer $\sqrt{n}$ and in what follows, w.l.o.g., we present this case for simplicity. This strategy (called \emph{DL-Partitioning}) consists of three phases:

\begin{itemize}
	\item \textbf{Phase 1}: Transforms each diagonal segment $l_{1}, l_{2}, \ldots , l_{\sqrt{n}} $ into a line segment. Notice that segment $l_{k}$, $1 \le k \le \sqrt{n}$, contains $ \sqrt{n} $ nodes occupying positions $(x+h_k,y+h_k), (x+h_k+1,y+h_k+1), \ldots, (x+h_k+\sqrt{n}-1,y+h_k+\sqrt{n}-1)$, for $h_k=n-k\sqrt{n}$; see Figure~\ref{fig:S1_nrootn_partition} (b). Each of these nodes moves independently to the leftmost column of $l_{k}$, namely column $x+h_k$, and the new positions of the nodes become $(x+h_k,y+h_k), (x+h_k,y+h_k+1), \ldots, (x+h_k,y+h_k+\sqrt{n}-1)$.
	By the end of Phase 1, $ \sqrt{n} $ vertical line segments have been created (Figure~\ref{fig:S1_nrootn_partition} (c)).
	
	\item\textbf{Phase 2}: Transfers all $ \sqrt{n} $ line segments from Phase 1 down to the bottommost row $y$ of the diagonal $S_{D}$. 
	Observe that line segment $l_{k}$ has to move distance $h_k$ (see Figure~\ref{fig:S1_nrootn_partition} (d)). 
	
	\item\textbf{Phase 3}: Turns  all $ \sqrt{n} $ line segments into the bottommost row $y$ (Figure~\ref{fig:S1_nrootn_partition} (e)). In particular, line $l_k$ will be occupying positions $(x+h_k,y), (x+h_k+1,y), \ldots, (x+h_k+\sqrt{n}-1,y)$.

\end{itemize}

Now, we are ready to analyse the running time of all phases of \emph{DL-Partitioning}.

\begin{theorem}
	Given an initial diagonal of $ n $ nodes, \emph{DL-Partitioning} solves the {\sc DiagonalToLine} problem in $O(n\sqrt{n})$ steps.
\end{theorem}
\begin{proof}
	From the above reasoning, the proof follows by analysing all phases of \emph{DL-Partitioning}. In the first phase, the cost of the trivial line formation is the run of all distances for $ \ceil{\sqrt{n}} $ nodes to be gathered at the leftmost column of a single segment $l_{k}$ for all $k$, where $1 \le k \le \sqrt{n}$. Observe that in Figure~\ref{fig:S1_nrootn_partition} (b), the $ \ceil{\sqrt{n}} $ nodes of $l_{k}$ have to move distances of $\Delta= 0, 1 ,2.3 , \ldots , (\sqrt{n}-1)$. Therefore, the total run of all distances in $l_{k}$ (\textit{except the bottommost node which stays still in place}),  is:
	\begin{align*}
	t_{1} &= 1 + 2 + \ldots + (\sqrt{n}-1)= \sum_{i=1}^{\sqrt{n}-1} i \\
	&=  \frac{\sqrt{n}(\sqrt{n}-1)}{2} = \frac{n -\sqrt{n}}{2}\\
	&= O (n),    \numberthis  \label{eqn101}
	\end{align*}
	From \eqref{eqn101}, the total run $T_{1}$ for all $ \sqrt{n} $ segments is given by:
	\begin{align*}
	T_{1} &=   t_{1} \cdot  \ceil{\sqrt{n}} \\
	&=  \frac{n -\sqrt{n}}{2} \cdot  \ceil{\sqrt{n}} =  \frac{n\sqrt{n} -n}{2} = \frac{n(\sqrt{n} -1)}{2}\\
	&= O (n\sqrt{n}).  \numberthis \label{eqn102}
	\end{align*} 	
	Next, in phase 2, all $ \sqrt{n} $ line segments transfer down into the bottommost of the diagonal except the one already there. As mentioned above, any line segment $l_{k}$ has to transfer a distance of $n-k\sqrt{n}$ to reach the $ y $ bottommost row in a total $T_{2}$ as follows:
	\begin{align*}
	T_{2} &=\sum_{k=1}^{\sqrt{n}-1} (n - k \sqrt{n}) = (n\sqrt{n} - n) - \sum_{k=1}^{\sqrt{n}-1} k \sqrt{n} \\
	&=(n\sqrt{n} - n) - \sqrt{n} \sum_{k=1}^{\sqrt{n}-1} k =(n\sqrt{n} - n) - \sqrt{n} \bigg(\frac{\sqrt{n}(\sqrt{n}-1)}{2}\bigg) \\
	&=(n\sqrt{n} - n) - n \bigg(\frac{\sqrt{n}-1}{2}\bigg) = \frac{n(\sqrt{n} -1)}{2}  \\
	&= O (n \sqrt{n}).     \numberthis \label{eqn103}
	\end{align*}
	The last phase is basically to turn (\textit{re-orientate})  every line segments of $ \sqrt{n} $ nodes into the $y$ bottommost row of length $n$. Now, each line segment $l_{k}$ contributes a single node at the bottommost row $y$, therefore, we have $ \sqrt{n} $ nodes are sitting on row $y$, and the other $n -  \sqrt{n} $ nodes are waiting to be pushed (\textit{turned}) into the bottommost row of length $ \Delta = n$, in order to form the goal spanning line. By Lemma \ref{lem:Linear_rearrangement_Of_Small_Lines} and  \ref{lem:Transfer_Line_H_to_V}, the transformation starts \textit{re-orientating} each $l_{k}$ towards the bottommost row $y$, until it is being completely filled in by $n$  nodes of all segments. Recall that turning each node of the $n -  \sqrt{n} $ into the bottommost row $y$ costs two line steps, therefore, the total $T_{3}$ for all $n - \sqrt{n}$ is:
	\begin{align*}
	T_{3} &= 2 \cdot  (n - \sqrt{n}) = 2n - 2\sqrt{n} = 2(n - \sqrt{n})\\
	&= O (n). \numberthis \label{eqn104}
	\end{align*}      
	Altogether, the sum of \eqref{eqn102}, \eqref{eqn103} and \eqref{eqn104}  gives the total cost $T$, in steps, for all phases,
	\begin{align*}
	T &= T_{1} + T_{2} +T_{3}\\
	&= \frac{n(\sqrt{n} -1)}{2} + \frac{n(\sqrt{n} -1)}{2} + 2(n - \sqrt{n})\\
	& = n(\sqrt{n} -1) + 2n -2\sqrt{n} = n\sqrt{n} + n -2\sqrt{n} = n (\sqrt{n} + 1) - 2\sqrt{n}\\
	& = O (n \sqrt{n}). 
	\end{align*}
	
\end{proof}

\subsection{Preserving Connectivity through Folding}
\label{subsec:DLC-Folding}

In this section, our purpose is to transom the diagonal $S_{D}$ into a line $S_{L}$ by exploiting the line mechanism, with preserving connectivity of the shape throughout transformations. This strategy, called \emph{DLC-Folding}, is as follows; we partition $S_{D}$ into several segments of the same lengths, as we did previously in \emph{DL-Partitioning} (see Figure \ref{fig:S1_nrootn_partition} (a)). Then, in each segment, we perform three operations: \emph{turn}, \emph{push} and \emph{turn}. From the topmost segment of  $S_{D}$, form a \emph{line} segment  by a trivial brute-force line formation, that is, move each node to the bottommost row in that segment. After that, push the \emph{line} segment towards the leftmost column of $S_{D}$ by a distance of  $ \Delta = \sqrt{n} $. Lastly, we fold this \emph{line} segment diagonally to align above the next following diagonal segment, therefore two parallel diagonal segments are formed, as illustrated in Figure \ref{fig:Zoom_in_Folding}. The strategy keeps folding segments above each other in this way, until finishing at the bottommost segment of $S_{D}$. In the end, $S_{D}$  transforms into a \textit{nice shape}, which is trivially converted into a line. \cref{fig:Exmpel_folding,fig:Exmpel_folding1,fig:Exmpel_folding2,fig:Exmpel_folding3,fig:Exmpel_folding4} demonstrate the above transformations on a diagonal of 25 nodes.

\begin{figure}[h!t]
	\centering
	\includegraphics[scale=0.7]{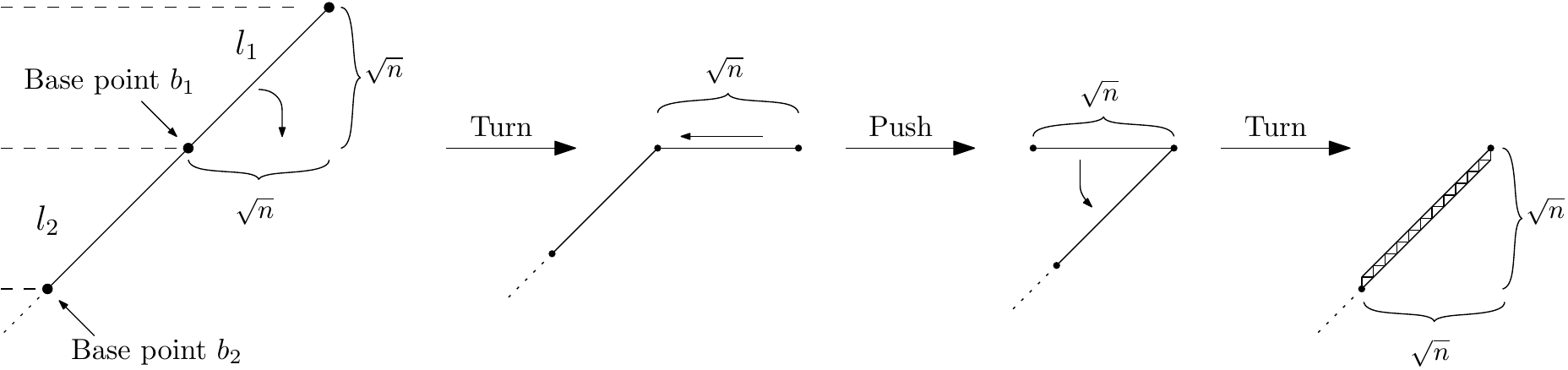}
	\caption{ Three operations (\emph{turn}, \emph{push} and \emph{turn}) for folding the topmost segment of a diagonal line.}
	\label{fig:Zoom_in_Folding}	
\end{figure} 

\begin{figure}[h!t]
	\centering
	\includegraphics[scale=0.6]{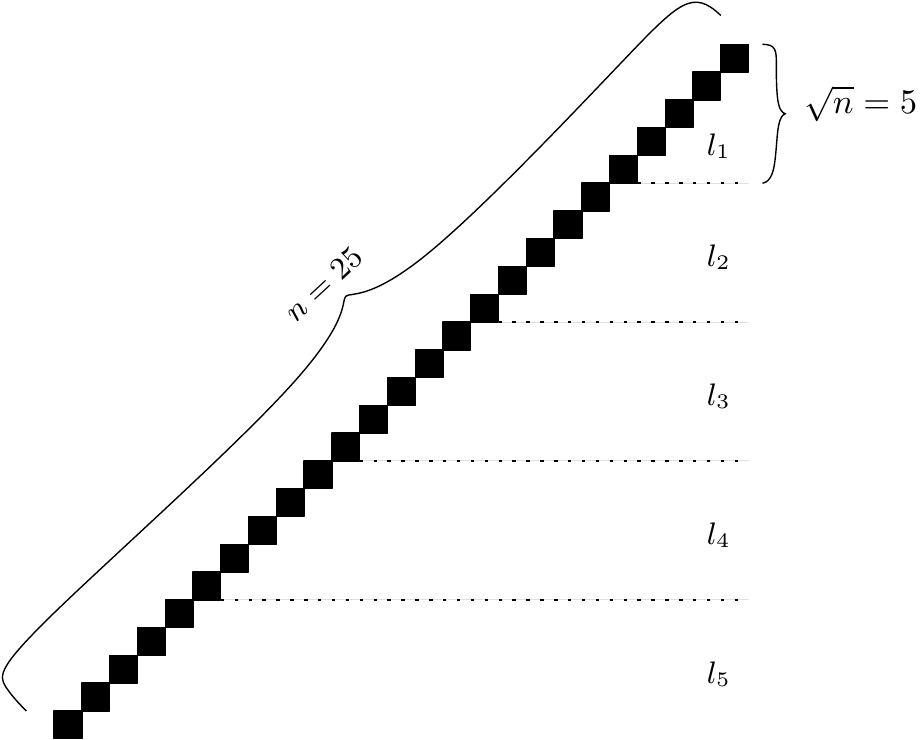}
	\caption{A diagonal line of $ 25 $ nodes.}
	\label{fig:Exmpel_folding}	
\end{figure}   	
\begin{figure}[h!t]
	\centering
	\includegraphics[scale=0.6]{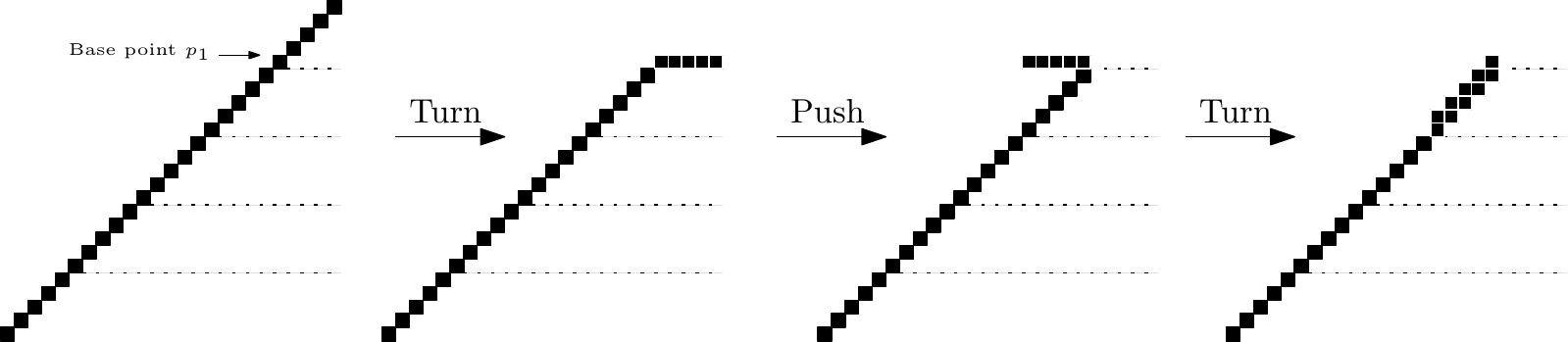}
	\caption{\emph{turn}, \emph{push} and \emph{turn} of the first phase.}
	\label{fig:Exmpel_folding1}	
\end{figure}   	
\begin{figure}[h!t]
	\centering
	\includegraphics[scale=0.6]{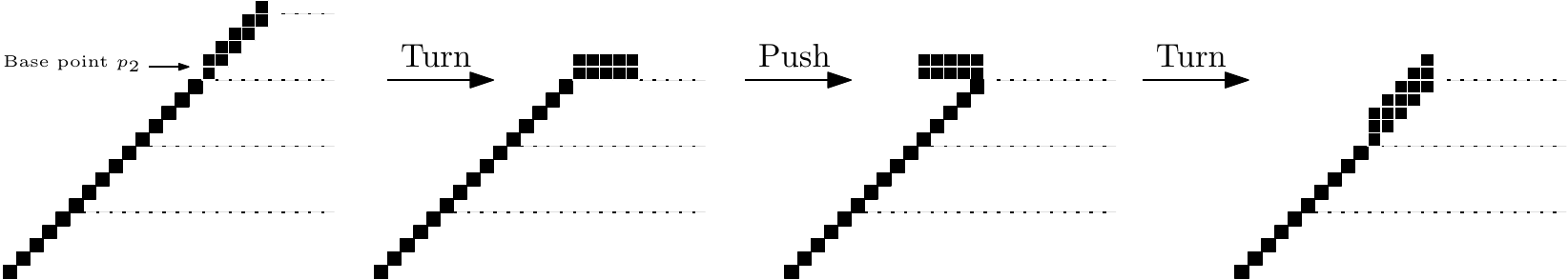}
	\caption{\emph{turn}, \emph{push} and \emph{turn} of the second phase.}
	\label{fig:Exmpel_folding2}	
\end{figure}   	
\begin{figure}[h!t]
	\centering
	\includegraphics[scale=0.6]{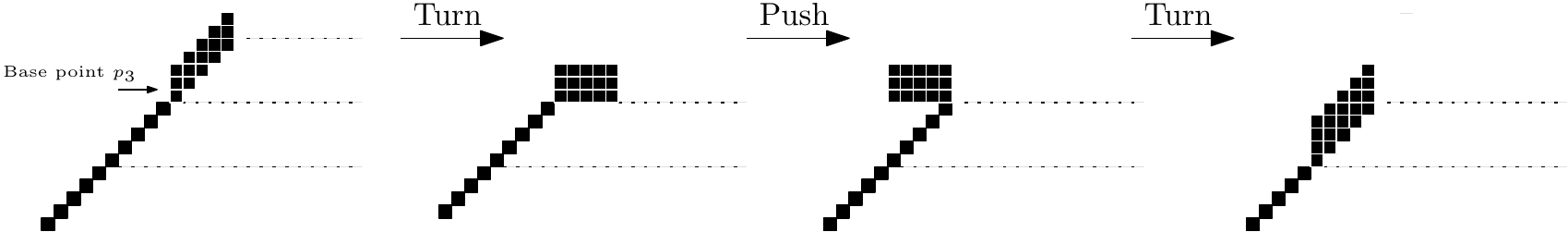}
	\caption{\emph{turn}, \emph{push} and \emph{turn} of the third phase.}
	\label{fig:Exmpel_folding3}	
\end{figure} 
\begin{figure}[h!t]
	\centering
	\includegraphics[scale=0.6]{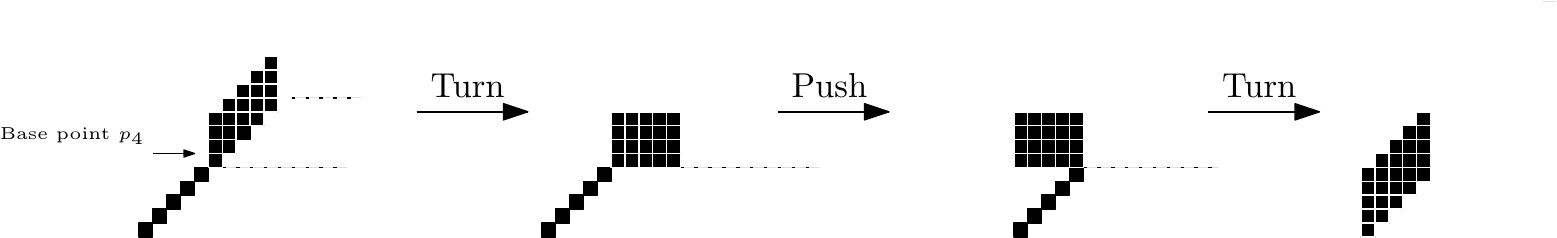}
	\caption{\emph{turn}, \emph{push} and \emph{turn} of the forth phase. Notice the resulting figure in the far right represents a \textit{nice shape}.}
	\label{fig:Exmpel_folding4}	
\end{figure}   

\subsubsection{Formal Description}
\label{subsubsec:DLC-Folding}   

Let $S_{D}$ be a diagonal of $n$ nodes occupying  $ (x,y), (x+1,y+1), \ldots, (x+n-1,y+n-1) $, such that $x$ and $y$ are the leftmost column and the bottommost row of $S_{D}$, respectively. $S_{D}$ is then divided into $ \sqrt{n} $ segments, $l_{1}, l_{2}, \ldots , l_{\sqrt{n}} $, each of which has a length of $ \sqrt{n} $, where $l_{1}$ and $l_{\sqrt{n}}$ are the topmost and bottommost segments of $S_{D}$, respectively. Observe that segment $l_{k}$, $1 \le k \le \sqrt{n}$, consists of $\sqrt{n} $ nodes occupying $(i,j), (i+1,j+1), \ldots, (i+\sqrt{n}-1,j+\sqrt{n}-1)$, where $ i = x+h_k$ and $  j= y+h_k$, for $h_k=n-k\sqrt{n}$. Here, $l_{k}$ has a base point $b_{k} = (i,j)$, which is the bottommost node of $l_{k}$. Moreover, \emph{DLC-Folding} completes its course in $k$ phases, in each phase $k$, we apply three operations (\emph{turn}, \emph{push} and \emph{turn}) to fold the corresponding segment(s) around the base point $b_{k}$. The first segment $l_{1}$ folds around $b_{1}$ in the first phase as follows \emph{(due to symmetry, it is sufficient to demonstrate one orientation)}: 
\begin{alphaenumerate} 
	\item  \textit{turn}. Moves all $\sqrt{n}$ nodes into the bottommost row of $l_{1}$ (brute-force line formation). Notice that the $l_{1}$ nodes change their positions from $(i,j), (i+1,j+1), \ldots, (i+\sqrt{n}-1,j+\sqrt{n}-1)$ into $(i,j), (i+1,j), \ldots, (i + \sqrt{n}-1 ,j)$. By the fist operation, a horizontal line segment of length $\sqrt{n}$ have been formed, and the base point $b_{1}$ keeps place at $(i,j) $.  \label{itm:turn}
	
	\item \emph{push}. Pushes the $l_{1}$ line segment $ \sqrt{n} $ steps towards the leftmost column of the diagonal $S_{D}$, i.e., the $y$ column. All $\sqrt{n}$ nodes of $l_{1}$ transfer altogether into $(i -\sqrt{n},j ), (i +1-\sqrt{n},j), \ldots, (i + \sqrt{n}-1 - \sqrt{n}, j)$.  \label{itm:push}
	
	\item \textit{turn}. Converts the line segment $l_{1}$ into diagonal again by moving its $ \sqrt{n} $ nodes down to align above the $ \sqrt{n} $ nodes of the next following diagonal segment $l_{2}$. This takes place by transferring them into positions  $(i -\sqrt{n},j -\sqrt{n}+1), (i+1 -\sqrt{n} ,j-\sqrt{n} +2), \ldots, (i + \sqrt{n}-1 - \sqrt{n}, j)$, \emph{(except the bottommost node-base point $b_{1} $ which stays still in place at ($i -\sqrt{n},j $))}. By the end of this phase, two parallel diagonal segments $l_{1}$ and $l_{2}$ have been created, as in Figure \ref{fig:Zoom_in_Folding}. \label{itm:turn_3} 
\end{alphaenumerate}

The two parallel segments  $l_{1}$ and $l_{2}$ are consisting of $2\sqrt{n}$ nodes and comprising of $2\sqrt{n}$ vertical lines. Consequently, a new \emph{connected} shape has been formed and is defined below. 

\begin{definition}  \label{def:Ladle_shape}
	A $Ladle$ is a connected shape of  $n$ nodes consisting of two parts, $\mathcal{D}$ and $ \mathcal{S}$. For a given phase $k$, where $2 \le k \le \ceil{\sqrt{n}}$, we have  $Ladle_{k} = \mathcal{D}_{k} + \mathcal{S}_{k}$, where $ \mathcal{D}_{k} $ and $ \mathcal{S}_{k} $ are connected via a base point $b_{k} = (i, j)$, such that:\\
	
	\item - $\mathcal{D}_{k} $, is a diagonal  line containing  $ n - k \sqrt{n} +1$ nodes occupying $ (x,y), (x+1,y+1), \ldots, (i, j) $, such that $x$ and $y$ are the leftmost column and the bottommost row of $Ladle_{k}$, respectively, where $ \sqrt{n} < i =j < n -\sqrt{n} +1 $.  $\mathcal{D}_{k}$ is connected  to $ \mathcal{S}_{k}$ via its topmost node at $(i,j)$.\\
	
	\item - $ \mathcal{S}_{k}$, is parallelogram consists of $k$ parallel diagonal segments of size $k\sqrt{n}$ nodes formed $\sqrt{n}$ lines. $ \mathcal{S}_{k} $ is connected to $\mathcal{D}_{k}$ via its bottommost node at $(i,j)$, as depicted in Figure~\ref{fig:Ladle_shape}. 
	\begin{figure}[h!t]
		\centering
		\includegraphics[scale=1.2]{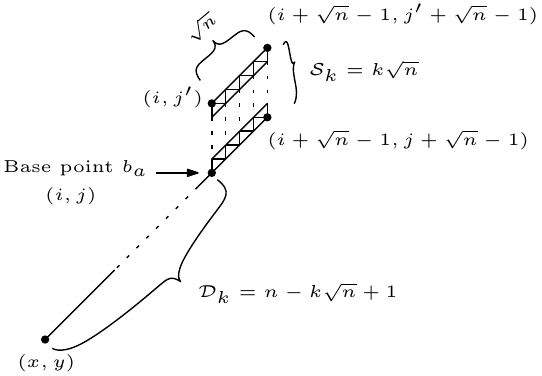}
		\caption{A $Ladle$ shape in phase $k$, where $j\prime = j+k-1$.}
		\label{fig:Ladle_shape}	
	\end{figure}   		
\end{definition}

 Throughout this section, we prove that at any phase $k$ of \emph{DLC-Folding}, there are $k$ parallel \emph{(diagonal)} segments containing $ k\sqrt{n} $ nodes and forming  $\sqrt{n}$ vertical lines. We now prove that the three operations (\emph{turn}, \emph{push} and \emph{turn}) transforms $S_{D}$ into a $Ladle$ by the end of the first phase of \emph{DLC-Folding}.
 
 \begin{lemma} \label{lem:CreateLadel}
 	Let $S_{D}$ be a diagonal of order $ n $ partitioned into $ \sqrt{n}  $ segments $l_{1},l_{2},...,l_{\sqrt{n}} $. \emph{DLC-Folding} converts $S_{D}$ into a $Ladle$ by the end of the first phase. 
 \end{lemma}
\begin{proof}
	Consider a diagonal $S_{D}$ of $n$ nodes as defined previously, which is partitioned into $\sqrt{n}$ segments of length $\sqrt{n}$ each. Now, perform the three operations (\emph{turn}, \emph{push} and \emph{turn}) described above on the topmost segment of $S_{D}$ \emph{(it is sufficient due to symmetry)}, we will obtain a connected shape consists of two parts, a diagonal line whose nodes occupy $(x,y), (x+1,y+2), \ldots, (x+ n - \sqrt{n}-1, y+n-\sqrt{n}-1)$ and two parallel diagonal segments of $2\sqrt{n}$ nodes. Both are connected via the base point $(x+ n - \sqrt{n}-1, y+n-\sqrt{n}-1)$, which is the topmost node of the diagonal part and the bottommost of the two parallel diagonal segments. As a result, a one can easily find out that the new shape constructed by the end of the first phase meets all conditions mentioned in Definition \ref{def:Ladle_shape}, therefore, it is a $Ladle$.
\end{proof}

The following lemma shows that the three operations of \emph{DLC-Folding} hold  in any phase $k$, where $2 \le k < \sqrt{n}$, such that in phase $k+1$, the size of  $\mathcal{S}_{k}$ increases by $ \sqrt{n} $, and conversely the $\mathcal{D}_{k}$ length decreases by $ \sqrt{n} $.

\begin{lemma}	\label{lem:Folding_increses_SectionSize} 
	Consider a $Ladle$ of $n$ nodes in phase $k$, where $1 < k \le \sqrt{n} $. Then, in phase $k+1$, \emph{DLC-Folding} increases the size of $ \mathcal{S}_{k}$ by $ \sqrt{n}$  and  decreases the length of $ \mathcal{D}_{k} $  by $\sqrt{n} $.     		   	 	
\end{lemma} 
\begin{proof}
	The size of the $Ladle  = |n|$ must be the same each phase and all time over transformations. In phase $k$, a $Ladle_{k} $ consists of two parts, $ \mathcal{D}_{k} = |n - k \sqrt{n} +1|$ and $ \mathcal{S}_{k} = |k\sqrt{n}| $, where both are connected via a common node $(i,j)$ (see Definition \ref{def:Ladle_shape}). Now, perform the three operations (\emph{turn}, \emph{push} and \emph{turn}) on the $ \mathcal{S}_{k}$ part that contains $k$ segments of length $\sqrt{n}$  aligned diagonally on top of each other. Let assume that the $k$ segments form $\sqrt{n}$ vertical lines since the same argument applies symmetrically to a different orientation. First, we move all vertical $\sqrt{n}$ lines  downwards to the bottommost row $i$ of $ \mathcal{S}_{k}$, which shall form  $k$ horizontal lines by completely filling in the $k$ bottom rows of $ \mathcal{S}_{k} $. Therefore, those horizontal lines create a rectangle, as depicted in Figure \ref{fig:Ladel_Process} (a). Then, the second operation pushes the $k$ vertical lines of the rectangle $\sqrt{n} $ steps horizontally towards the $x$ leftmost column of the $Ladle_{k}$, as in Figure \ref{fig:Ladel_Process} (b). Lastly, the strategy completes folding by translating the vertical $\sqrt{n} $ lines downwards, until each of them stays above a node of the next following segment \emph{(notice that every vertical line is moved except the rightmost one)}, see Figure \ref{fig:Ladel_Process} (c).  By the end of phase $k+1$, a new  $Ladle $ has been created, which is consisting of $ \mathcal{D}_{k+1} = |n - (k-1) \sqrt{n} +1| $ and  $ \mathcal{S}_{k+1} = |(k+1)\sqrt{n}|$  connected via the common $b_{k+1}$ base point at $(i- \sqrt{n}, j -\sqrt{n})$. Hence, we conclude that in phase $k+1$, the size of $ \mathcal{S}_{k+1} $ increased by $\sqrt{n}$ nodes, while the length of  $ \mathcal{D}_{k+1} $ decreased by $\sqrt{n}$, and this holds trivially and inductively for any phase $k$, where $1 < k \le \sqrt{n} $. 
	\begin{figure}
		\centering
		\subcaptionbox{}
		{\includegraphics[scale=1.2]{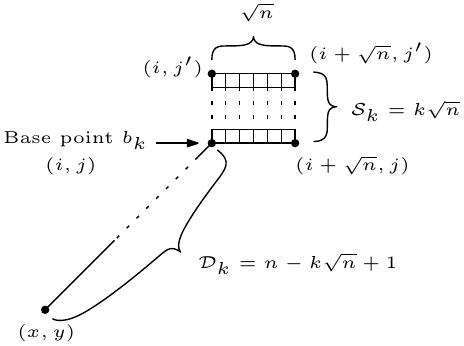}}	 \qquad
		\subcaptionbox{}
		{\includegraphics[scale=1.2]{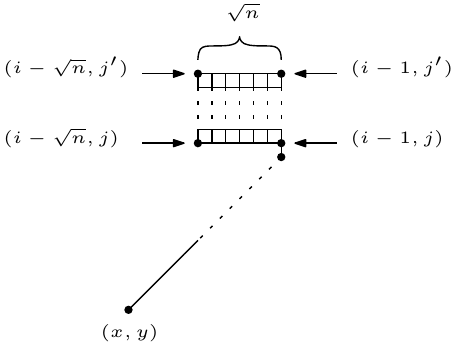}}	\qquad
		\subcaptionbox{}
		{\includegraphics[scale=1.2]{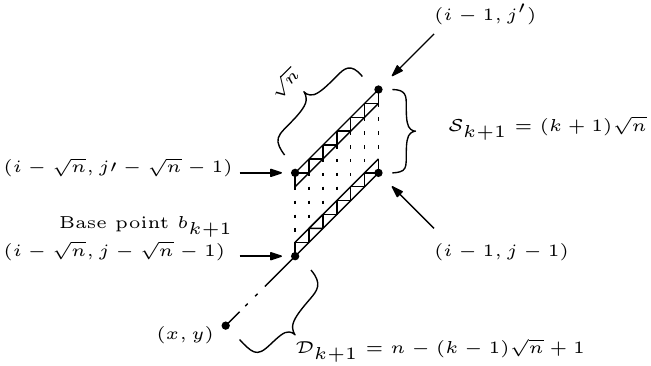}}			
		\caption{Folding a $Ladle_{k}$ over phase $k$, , where $j\prime = j+k-1$, see Lemma \ref{lem:Folding_increses_SectionSize} for further explanation. }
		\label{fig:Ladel_Process}	
	\end{figure}
\end{proof}	

Now, we prove that \emph{DLC-Folding} transforms $S_{D}$ into a \textit{nice} shape in $\sqrt{n}$ phases. 

\begin{lemma} \label{lem:Fold_Digonal_to_NICE}
		Given a diagonal $S_{D}$ of order $ n $ partitioned into $ \sqrt{n} $ segments,  \emph{DLC-Folding} converts $S_{D}$ into a \textit{nice}  shape in $ \sqrt{n}  $ phases. 
\end{lemma}
\begin{proof}
	By following Lemma \ref{lem:CreateLadel}, $S_{D}$ converts into a $Ladle_{2}$ which consists of two parts $\mathcal{D}_{2} = |n-2\sqrt{n}+1|$ and $ \mathcal{S}_{2} = |2\sqrt{n}| $. Then, by Lemma  \ref{lem:Folding_increses_SectionSize}, through the final phase $k = \sqrt{n}$,  all segments are being folded diagonally over each other, therefore the diagonal part of the $Ladle$ will be exhausted $\mathcal{D}_{\sqrt{n}} = \phi $, whilst the parallelogram part acquires all $n$ nodes, $ \mathcal{S}_{\sqrt{n}} = |n| $. The resulting shape of $\sqrt{n}$ vertical (horizontal) lines at the end of the final phase is no longer $Ladle$ and complies perfectly with all standards and properties of \textit{nice shapes}, see Definition \ref{def:nice_shape}; as a result, it is a \textit{nice} shape.
\end{proof}

At this point, we are ready to analyse the running time of \emph{DLC-Folding}  that preserves connectivity over its course.

\begin{lemma}	\label{lem:Fold_FirstSection_of_Stair_To_Line}   	
	Given a diagonal $S_{D}$ of order $ n $ partitioned into $ \sqrt{n}  $ segments, \emph{DLC-Folding} folds the topmost (bottommost) segment in $O(n)$ steps.  
\end{lemma}  
\begin{proof}
     In the first phase, we perform  \emph{turn}, \emph{push} and \emph{turn} on the first topmost (bottommost) segment of $S_{D}$ of length $ \sqrt{n} $. The first operation (\emph{turn}) is a brute-force line formation that is trivially computed by:
	\begin{align*}
	  1 + 2 + ... + (\sqrt{n}-1) = \frac{\sqrt{n}(\sqrt{n}-1)}{2} = \frac{n -\sqrt{n}}{2}, \numberthis  \label{eqn1111}
	\end{align*}
	 Then,  the second operation \emph{pushes} a line of $ \sqrt{n} $ length in $ \sqrt{n} $ line steps. Again, the last operation costs as much as the first \emph{turn}, namely $\frac{n -\sqrt{n}}{2}$. Altogether, the total cost required to fold the first segment thoroughly is at most:
	 \begin{align*}
	 t_{1} &= \frac{n -\sqrt{n}}{2} + \sqrt{n} + \frac{n -\sqrt{n}}{2} = n -\sqrt{n}+ \sqrt{n} = n \\
	 &= O(n).    
	 \end{align*}  
\end{proof}

\begin{lemma}   \label{lem:Folding_T_a_to_a+1}
	By the end of phase $k$, for all $1 < k \le \sqrt{n}$, \emph{DLC-Folding} folds $Ladle_{k}$ in $ O(n) $ steps.
\end{lemma} 
\begin{proof}	
	In phase $k$, $Ladle_{k}$ holds two parts, $ \mathcal{D}_{k} = |n - k \sqrt{n}+1|  $ and  $ \mathcal{S}_{k} = |k \sqrt{n}|$. In the first operation of \emph{DLC-Folding},  by exploiting the linear mechanisms, all $\sqrt{n}$ lines of $ \mathcal{S}_{k}$ translate in  a distance equals to \eqref{eqn1111}, namely $ \frac{n -\sqrt{n}}{2} $.  Now, the $\sqrt{n}$ lines have moved and formed another $k$ lines in a different orientation. Therefore, in the second operation, we push those $k$ lines $\sqrt{n}$ steps  in a total of:
   \begin{align*}
   k\sqrt{n} = O(n),  \numberthis  \label{eqn112}
   \end{align*}
  And the third move is completing the process by (folding) the $\sqrt{n}$ lines diagonally above the next segment, with the same cost of \eqref{eqn1111}: 
  \begin{align*}
  \frac{n -\sqrt{n}}{2} = O(n),  \numberthis  \label{eqn113}
  \end{align*}

   With this, by summing \eqref{eqn1111} , \eqref{eqn112} and \eqref{eqn113}, the total steps performed by the end of phase  $k$ is given by:
  \begin{align*}
  t_{k} &=   \frac{n -\sqrt{n}}{2} + k\sqrt{n} + \frac{n -\sqrt{n}}{2} =  n -\sqrt{n}+ k\sqrt{n} \\  
  &=  O(n).     \numberthis  \label{eqn114}
  \end{align*}
  
   This holds trivially from phase 2 and inductively for every phase $k$, for all $1 < k \le \sqrt{n}$.
\end{proof} 

Altogether, Proposition \ref{prop:Niceshape_to_line} and Lemmas \ref{lem:Fold_FirstSection_of_Stair_To_Line} and \ref{lem:Folding_T_a_to_a+1}, the running time of \emph{DLC-Folding} is,

\begin{theorem} \label{theo:Sahpe_To_Line_with_connectivty_Strategy_1} 
	Given an initial connected diagonal of $ n $ nodes, \emph{DLC-Folding} solves the {\sc DiagonalToLineConnected} problem in $O(n\sqrt{n})$ steps.    	
\end{theorem}
\begin{proof}
    By Lemma \ref{lem:Fold_FirstSection_of_Stair_To_Line}, \emph{DLC-Folding} creates a $Ladle$ in a total of:
	\begin{align*}
	T_{1} &=  \frac{n -\sqrt{n}}{2}.   \numberthis  \label{eqn115}
	\end{align*}
	 
   Now, Lemma \ref{lem:Folding_T_a_to_a+1} provides the running time of phase $k$, for all $1 < k \le \sqrt{n}$, therefore the total run  for all phases is computed by \emph{(except the first phase)}:
	 \begin{align*}
	 T_{2} &= \sum_{i=1}^{\sqrt{n}-1}  n -\sqrt{n}+ i\sqrt{n} = n\sqrt{n} -2n -\sqrt{n} +  \sum_{i=1}^{\sqrt{n}-1} i\sqrt{n}\\ 
	 &= n\sqrt{n} -2n -\sqrt{n} + \sqrt{n } \sum_{i=1}^{\sqrt{n}-1} i= n\sqrt{n} -2n -\sqrt{n} + n \bigg(\frac{\sqrt{n}-1}{2}\bigg) \\
	 &= n\sqrt{n} -2n -\sqrt{n} +  \bigg(\frac{n\sqrt{n}-n}{2}\bigg) = \dfrac{n\sqrt{n} -5n-2\sqrt{n}}{2} \\
	 &=  O(n\sqrt{n}).     \numberthis  \label{eqn116}
	 \end{align*}	
		
	 Then, the total cost for all phases of \emph{DLC-Folding} is given by summing \eqref{eqn115} and \eqref{eqn116} :
	\begin{align*}
	T_{3} &=  T_{1} + T_{2} \\
	 &= \frac{n -\sqrt{n}}{2}  + \dfrac{n\sqrt{n} -5n-2\sqrt{n}}{2}  = \dfrac{2n - 2\sqrt{n} + 2n\sqrt{n} - 10n - 4\sqrt{n}}{4} \\
	 &= \dfrac{2n\sqrt{n} -8n - 6\sqrt{n}}{4} = \dfrac{n\sqrt{n} -4n - 3\sqrt{n}}{2} \\
	&=  O(n\sqrt{n}).     \numberthis  \label{eqn117}
	\end{align*}
	
	Finally, the resulting shape of \emph{DLC-Folding} is a \textit{nice} shape, which transforms into a line $S_{L}$ in $O(n)$ steps(see Proposition \ref{prop:Niceshape_to_line}), then the total cost $T$ required to transform $S_{D}$ into $S_{L}$, is bounded above by:
	\begin{align*}
	T &=  T_{3} + O(n) \\
	&= O(n\sqrt{n}) + O(n) \\
	&=  O(n\sqrt{n}).     
	\end{align*}
\end{proof}

\subsection{Preserving Connectivity through Extending}
\label{subsec:nrootn_S2}

The current transformation, called \emph{DLC-Extending}, is another strategy to transform the diagonal $S_{D}$ into a line $S_{L}$ in $O(n\sqrt{n})$ steps, with preserving connectivity throughout transformations. As mentioned earlier, $S_{D}$ is partitioned into $ \sqrt{n} $ segments of length $ \sqrt{n} $ each. The implementation of this strategy lasts for $ \sqrt{n} $  phases equal to the number of segments. In each phase, we perform two types of movements, $ turn $ and $ push $. Generally speaking, \emph{DLC-Extending} starts building the spanning line by first performing a line formation on the \emph{topmost} \emph{(bottommost)} diagonal segment of $S_{D}$; after that, convert the rest of the diagonal segments into lines and include them to the main spanning line, sequentially one after the other. Notice that  \emph{DLC-Extending} extends the main spanning line gradually every phase by a length of $ \sqrt{n} $. Figure \ref{fig:EX_DLC-Extending} shows the performance of  \emph{DLC-Extending} on a diagonal of 25 nodes.

\begin{figure}
	\centering
	
	{\includegraphics[scale=0.6]{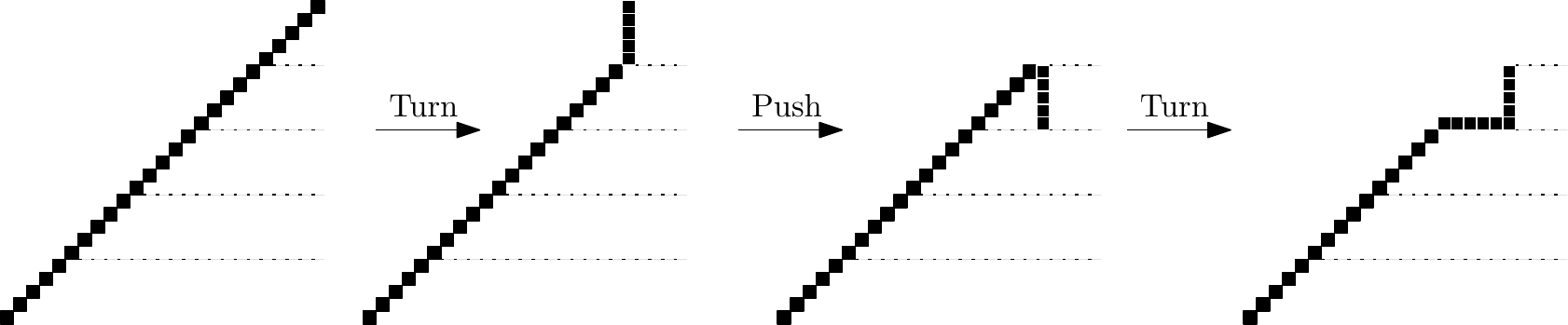}}	  \vspace{20px}

	{\includegraphics[scale=0.6]{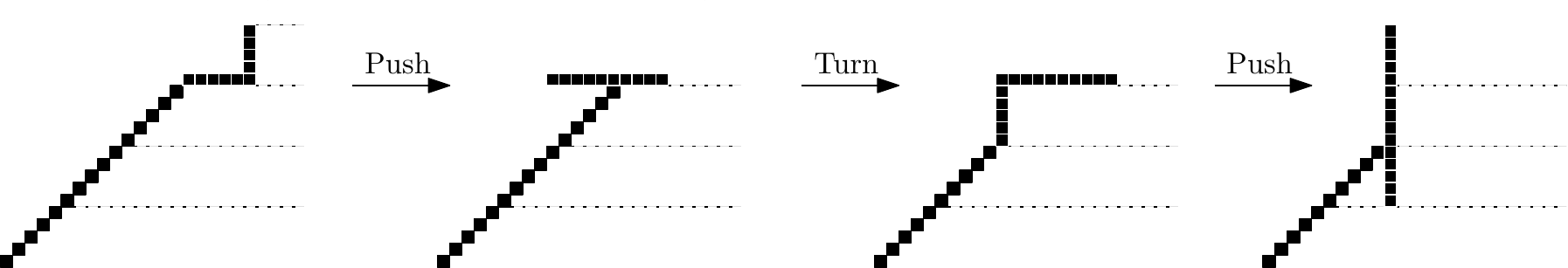}}     \vspace{20px}
	
	{\includegraphics[scale=0.6]{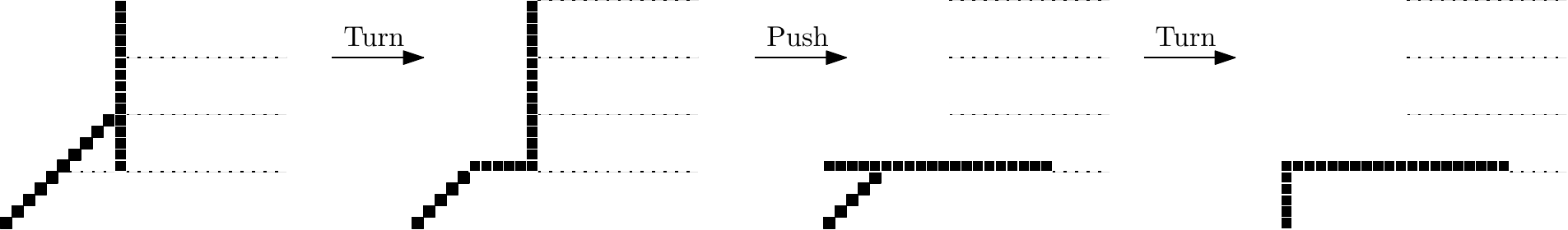}}	  \vspace{20px}
	
	{\includegraphics[scale=0.6]{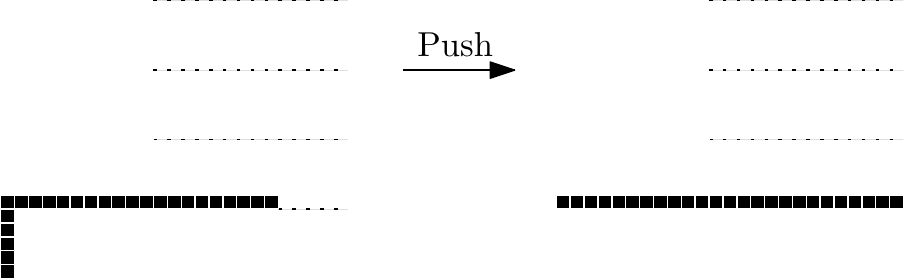}}	  
	\caption{Shows all transformations of  \emph{DLC-Extending} on a diagonal of 25 nodes.}
	\label{fig:EX_DLC-Extending}	
\end{figure} 

\subsubsection{Formal Description}
\label{subsubsec:DLC-Extending}   

Consider an initial diagonal $S_{D}$ of $n$ nodes defined and partitioned as in Section \ref{subsubsec:DLC-Folding}. In phase $k$, for all $1 \le k \le \sqrt{n}$, we perform two line operations, \emph{turn} and \emph{push}, on a diagonal segment $l_{k} \in S_{D}$ of length  $\sqrt{n} $ nodes. Assume $ l_k$ occupies $(i,j), (i+1,j+1), \ldots, (i+\sqrt{n}-1,j+\sqrt{n}-1)$,  where $ i = x+h_k$ and $ j= y+h_k$, for $h_k=n-k\sqrt{n}$.  Here, the bottommost node $b_{k} = (i,j)$ is the base point of $l_{k}$. Due to symmetry, we only show transformations starting vertically from the topmost segment of $S_{D}$. In the first phase, \emph{DLC-Extending} forms a line at the topmost segment $l_{1}$, by performing a brute-force line formation to transfer all nodes into the leftmost column of $l_{1}$. Consequently, all nodes move from $(i,j), (i+1,j+1), \ldots, (i+\sqrt{n}-1,j+\sqrt{n}-1)$ into $(i,j), (i,j+1), \ldots, (i  ,j + \sqrt{n}-1)$. Secondly, it \emph{pushes} this \emph{line} segment $ \sqrt{n} $ steps towards the bottommost row of $S_{D}$ to occupy $(i,j-\sqrt{n}), (i,j+1-\sqrt{n}), \ldots, (i  ,j + \sqrt{n}-1 -\sqrt{n})$. By the end of the first phase, \emph{DLC-Extending} shall construct a spanning line of length $ \sqrt{n} $.  For example, Figure~\ref{fig:FirstSection_of_Stair_To_Line} demonstrates the two operations (\emph{turn} and \emph{push}) of the first phase.  

\begin{figure}[h!t]
	\centering
	\includegraphics[scale=0.6]{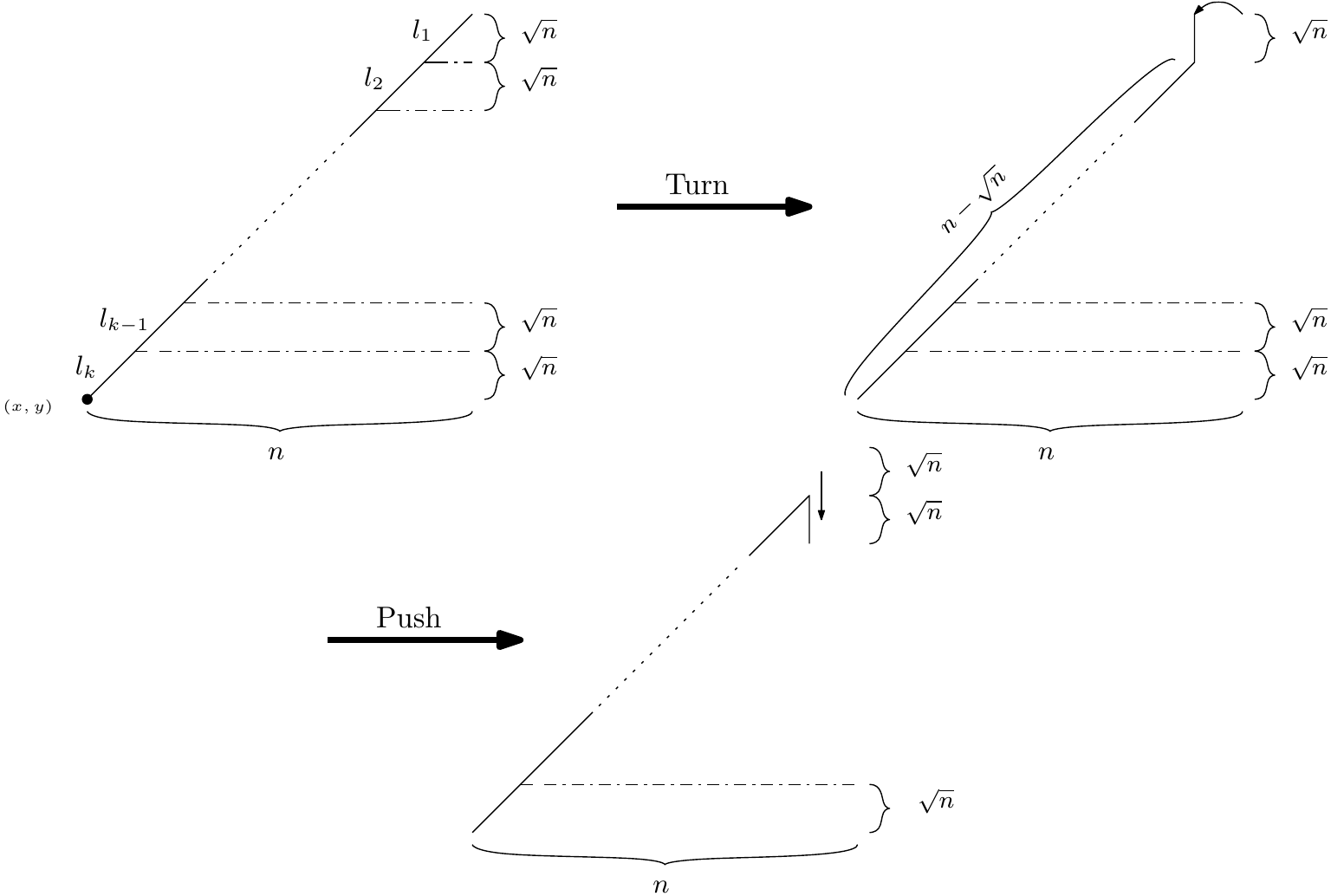}
	\caption{The process of the two operations (\emph{turn} and \emph{push}) in the first phase.}
	\label{fig:FirstSection_of_Stair_To_Line}	
\end{figure}

\begin{figure}[h!t]
	\centering
	\includegraphics[scale=0.6]{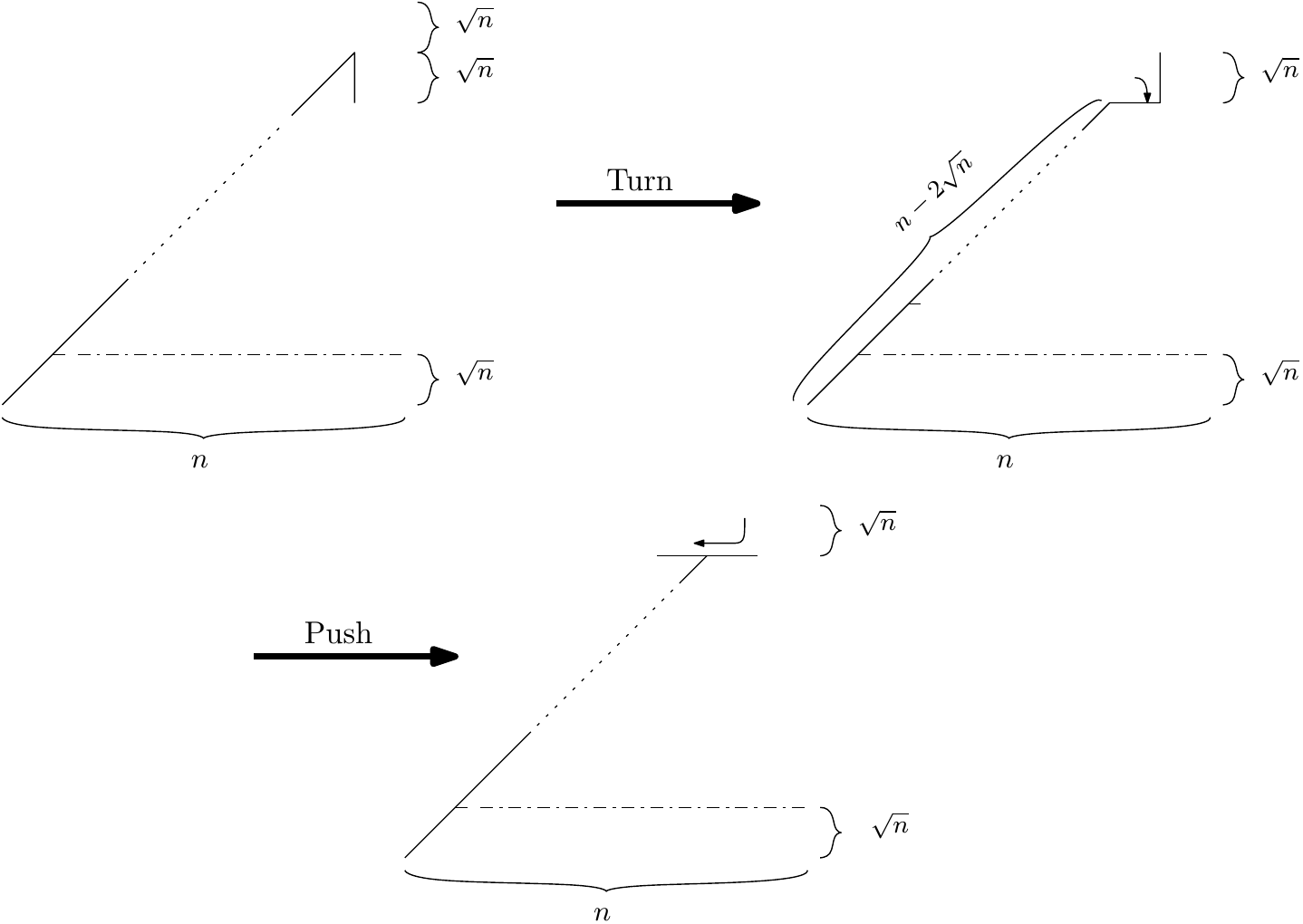}
	\caption{The process of the two operations (\emph{turn} and \emph{push}) in the second phase.}
	\label{fig:FirstSection_of_Stair_To_Line_2}	
\end{figure} 

In the second phase, \emph{DLC-Extending} turns $l_{2}$ into a \emph{line} by sending all $l_{2}$ nodes towards its bottommost row(\emph{in a sequential order}). Observe that $l_{1}$ and $l_{2}$ are now connected perpendicularly via $ (i,j - \sqrt{n}) $ (Figure \ref{fig:FirstSection_of_Stair_To_Line_2}), which then provides the ability to \emph{push}  $l_{1}$  into  $l_{2}$  a distance of $ \sqrt{n} $ vertically towards the leftmost column of $S_{D}$. By the end of this phase, the length of the spanning line is extended by $ \sqrt{n} $. As a result, we have obtained a specific \emph{connected} shape, called $\mathcal{T}\_shape$ and defined as follows;

\begin{definition}
	A $\mathcal{T}\_shape$ is a connected shape of $ n $ consisting of two parts, $ \mathcal{D} $ and $ \mathcal{S} $. Both are connected via a common intersection point $(i, j)$. Figure~\ref{fig:T_shap} shows an example of $\mathcal{T}\_shape$ in phase $k$ for all $ 1 < k < \sqrt{n} $, such that $ \mathcal{T}\_shape_{k} = \mathcal{D}_{k} + \mathcal{S}_{k} $, where\\
	\begin{alphaenumerate}
		\item - $\mathcal{D}_{k} $, is a diagonal line containing  $ n - k \sqrt{n} +1$ nodes and occupying $ (x,y), (x+1,y+1), \ldots, (i, j) $, such that $x$ and $y$ are the leftmost column and the bottommost row of $\mathcal{T}\_shape_{k}$, respectively, where $ \sqrt{n} < i =j < n -\sqrt{n} +1 $.  \\
		
		\item - $ \mathcal{S}_{k}$, is a horizontal or vertical line of length $k\sqrt{n}$ nodes occupying  $(i-\sqrt{n}-1, j), (i+1-\sqrt{n}-1, j), \ldots ,  (i\prime, j) $ or  $ (i, j-\sqrt{n}-1),(i, j+1-\sqrt{n}-1), \ldots, (i, j\prime) $, respectively,  where $i\prime = i+k\sqrt{n}-1$ and $j\prime = j+k\sqrt{n}-1$.
	\end{alphaenumerate}	
	\begin{figure}[h!t]
		\centering 
		\includegraphics[scale=1]{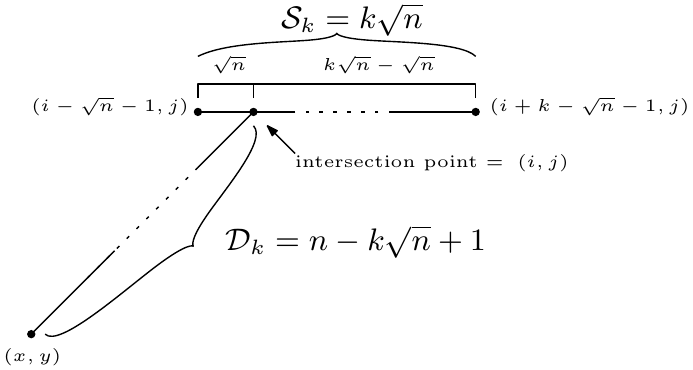}
		\caption{An example of a $\mathcal{T}\_shape$ in phase $k$.}
		\label{fig:T_shap}	
	\end{figure}   	
	\label{def:T_shape}
\end{definition}  

We turn now to prove correctness of \emph{DLC-Extending}. First, we show that the first and second phase converts a diagonal $S_{D}$ into a $\mathcal{T}\_shape$. 

\begin{lemma}  \label{lem:FirstSection_of_Stair_To_Line}   
	Let $S_{D}$ be a diagonal of order $ n $ partitioned into $ \ceil{\sqrt{n}}  $ segments $l_{1},l_{2},...,l_{\sqrt{n}} $. By the end of the second phase, \emph{DLC-Extending}  converts $S_{D}$ into a $\mathcal{T}\_shape$.
\end{lemma} 
\begin{proof}
	 By performing the above two main operations (\emph{turn} and \emph{push}) on the two topmost segments, $l_{2}$ and $l_{2}$, respectively, we shall acquire a connected shape of two parts: 1) A diagonal at $(x,y), (x+1,y+1), \ldots , (i, j)$, where $i = x+n - 2\sqrt{n} -2$ and $j =  y+n - 2\sqrt{n} -2$, and 2) A horizontal or vertical line of length $2\sqrt{n}$ occupying  $(i - 2\sqrt{n}, j), \ldots , ( i + 2\sqrt{n}, j) $ or  $ (i, j - 2\sqrt{n} ), \ldots, (i, j + 2\sqrt{n}) $, respectively. Observe that the two parts will have the same common intersection point at $(i, j)$ in all cases. Therefore, the resulting shape of the second phase meet all properties of Definition \ref{def:T_shape}, and hence, we conclude that this shape is a $\mathcal{T}\_shape$. 
\end{proof} 

The following lemma shows that the two operations of \emph{DLC-Extending} hold in any phase $k$, for all $ 2 \le k < \sqrt{n} $.

\begin{lemma}  \label{lem:SEconPhase_of_Stair_To_Line}   	
	Let a $\mathcal{T}\_shape_{k}$ of $n$ nodes be in phase $k$, for all $ 2 \le k < \sqrt{n}$. Therefore, in phase $k+1$, \emph{DLC-Extending} increases and decreases the length of $ \mathcal{S}$ and $ \mathcal{D} $ by $ \sqrt{n} $, respectively.         	
\end{lemma} 
\begin{proof}
	By Definition \ref{def:T_shape}, $\mathcal{T}\_shape_{k}$ in phase $k$ holds two segments, a diagonal $ \mathcal{D}_{k} = |n -k\sqrt{n}| $ occupies $ (x,y), \ldots, (i,j) $ and \emph{vertical (horizontal)} line $ \mathcal{S}_{k} = |k\sqrt{n}| $  at  $ (i - \sqrt{n}, j), \ldots , (i + k - \sqrt{n}, j)$, where both  intersect at a common point  $(i,j)$, as shown in Figure~\ref{fig:T_a_to_a+2} top-left. Assume that $\mathcal{S}_{k}$ is \emph{horizontal} (\textit{due to symmetry, it is sufficient to focus on only one orientation}). Therefore, we \emph{turn} the topmost diagonal segment  of $\mathcal{D}_{k}$ by performing a \emph{line formation} to collect all $ \sqrt{n} $ nodes at the leftmost column in that segment. Therefore, they move from cells $ (i - \sqrt{n}+1,j- \sqrt{n}+1), \ldots , (i-1,j-1) $  into  $ (i - \sqrt{n}+1,j- \sqrt{n}+1), \ldots , (i - \sqrt{n}+1,j-1) $. Notice that the shape still maintains connectivity. 
	
	Next, we \emph{push} $ \mathcal{S}_{k}$ to include the new constructed \emph{line} segment, which subsequently extends $ \mathcal{S}_{k}$  vertically with an increase of $\sqrt{n}$ in length, in order to occupy $ (i-\sqrt{n}, j-2\sqrt{n}), \ldots, (i-\sqrt{n}, j+k-2\sqrt{n}) $. in phase $k+1$, we obtain a new \textit{connected} shape consists of $ \mathcal{D}_{k+1}$ of length $ n - (k-1)\sqrt{n} $ and $ \mathcal{S}_{k+1}$ of  $ (k+1)\sqrt{n}$, and they both intersect at a common point $ (i - \sqrt{n}-1, j) $, as shown in Figure~\ref{fig:T_a_to_a+2}. Therefore, we derive that \emph{DLC-Extending} in phase $k+1$ decreases the length of $ \mathcal{D}_{k} $  by $ \sqrt{n} $ and increases  $ \mathcal{S}_{k} $ by $ \sqrt{n} $.  
	\begin{figure}[h!t]
		\centering
		\includegraphics[scale=1]{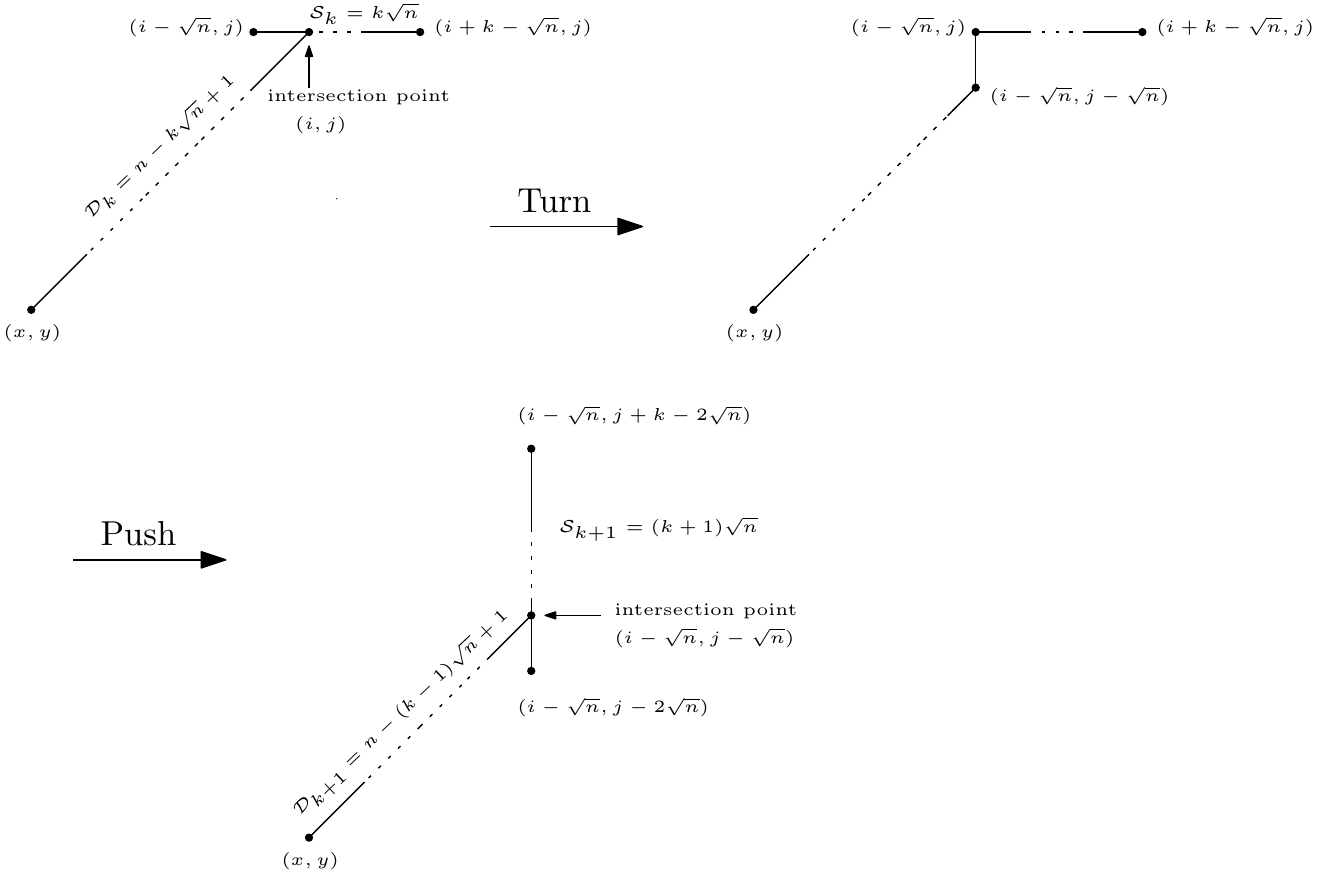}
		\caption{All transformations of \emph{DLC-Extending} on $\mathcal{T}\_shape$ from phase $k$ to $ k+1 $, see Lemma \ref{lem:SEconPhase_of_Stair_To_Line} for more details. }
		\label{fig:T_a_to_a+2}	
	\end{figure} 
\end{proof}

Here, we show that the last phase $\sqrt{n}$ of \emph{DLC-Extending} converts a $\mathcal{T}\_shape$ into a spanning line $S_{L}$.

\begin{lemma}\label{lem:T_a_to_a+2}   	
	Given a diagonal $S_{D}$ of $ n $ nodes and partitioned into $  \sqrt{n} $ segments, \emph{DLC-Extending} transforms a $\mathcal{T}\_shape$ into a spanning line $S_{L}$ by the end of the last phase.	
\end{lemma} 
\begin{proof}
	It follows from Lemmas \ref{lem:FirstSection_of_Stair_To_Line} and \ref{lem:SEconPhase_of_Stair_To_Line}. At any phase $k $, for all $2 \le k < \sqrt{n}$, $\mathcal{D}$ decreases and $\mathcal{S}$ increases by $\sqrt{n}$. Thus, we shall obtain the following in phase $ \sqrt{n}-1 $;
	\begin{align*}
	\mathcal{T}\_shape_{\sqrt{n}-1} &= \mathcal{D}_{\sqrt{n}-1} +  \mathcal{S}_{\sqrt{n}-1}\\ 	
	&=  \big[  n -( \sqrt{n}-1)\sqrt{n}  \big]  +  \big[  (\sqrt{n}-1)\sqrt{n} \big] \\
	&=   \big[   n- (n -\sqrt{n})  \big] + \big[ (n - \sqrt{n})\big]\\
	&= \big[ \sqrt{n} \big] + \big[  (n - \sqrt{n}) \big].
	\end{align*}
	As a result, we acquire $\mathcal{D}_{\sqrt{n}-1}$ of length $\sqrt{n}$ and $ \mathcal{S}_{\sqrt{n}-1} $ of $n - \sqrt{n}$ nodes. At this moment, in the last phase $\sqrt{n}$, we can trivially \emph{turn} and \emph{push} the diagonal $\mathcal{D}_{\sqrt{n}-1} = |\sqrt{n}|$ and include it to $\mathcal{S}_{\sqrt{n}-1}$ by the line formation and pushing mechanism, in order to form the final spanning line $S_{L}$.
	    
\end{proof} 

To sum up, the total running time for all $ \sqrt{n} $ phases of \emph{DLC-Extending} is analysed below.  

\begin{theorem} \label{theo:DLC_Extending}
	Given an initial connected diagonal of $ n $ nodes, \emph{DLC-Extending} solves the {\sc DiagonalToLineConnected} problem in $O(n\sqrt{n})$ steps.
	\label{prop:Sahpe_To_Line_with_connectivty} 
\end{theorem}
\begin{proof}
	By Lemmas \ref{lem:FirstSection_of_Stair_To_Line}, \ref{lem:SEconPhase_of_Stair_To_Line} and \ref{lem:T_a_to_a+2}, \emph{DLC-Extending} carries out transformations in $ \sqrt{n} $ phases. In phase $k$, for all $1 \le k \le \sqrt{n}$, the two main operations, \emph{turn} and \emph{push}, is performed. Here, the first operation \emph{turns} a corresponding segment $l_{k}$ by applying a brute-force line formation on its $ \sqrt{n} $ nodes \emph{(except the base point)} as follows;
	\begin{align*}
	  t_{1} &= 1 + 2 + \ldots + (\sqrt{n}-1)= \sum_{i=1}^{\sqrt{n}-1} i \\ 
	  &=  \frac{\sqrt{n}(\sqrt{n}-1)}{2} = \frac{n -\sqrt{n}}{2}\\
	  &= O (n),   \numberthis \label{BFS}
	\end{align*}
	
  Multiply \refeq{BFS} by $ \sqrt{n} $ to obtain the total \emph{turns} for all $ \sqrt{n} $ phases:
	\begin{align*}
	T_{1} &=  t_{1} \cdot \sqrt{n} = \frac{n -\sqrt{n}}{2} \cdot \sqrt{n} = \frac{n\sqrt{n} - n}{2}\\
	&= O (n\sqrt{n}).   \numberthis \label{Total_BFS}
	\end{align*}
		
	 The second operation of phase $k$  \emph{pushes} the spanning line  by $(k-1)\sqrt{n}$ steps horizontally or vertically in at most (see Lemma \ref{lem:Transfer_Line_H_to_V}):
	 \begin{align*}
	 	t_{2} &= 2 \cdot (k-1)\sqrt{n} = k\sqrt{n} - \sqrt{n}, \numberthis \label{Push_line}
	 \end{align*}
	 
	 Now, observe that the last phase only \emph{turns} and \emph{pushes} the last diagonal segment of length $ \sqrt{n} $ \emph{(dose not push the whole spanning line, see Lemma \ref{lem:T_a_to_a+2})}. Therefore, the total of \emph{push} operations for all $ \sqrt{n} $ phases is the summation of \refeq{Push_line} plus pushing the last segment of length $ \sqrt{n} $:
	 	\begin{align*}
	 	T_{2} &= \sum_{i=1}^{\sqrt{n}-1} t_{2} + 2\sqrt{n} \\
	  &= \sum_{i=1}^{\sqrt{n}-1} i\sqrt{n} +\sqrt{n} = \sqrt{n}  \sum_{i=1}^{\sqrt{n}-1} i + \sqrt{n}= \sqrt{n}  \cdot \frac{n -\sqrt{n}}{2} + \sqrt{n}  = \frac{n\sqrt{n} -n}{2} + \sqrt{n} \\
	 &= O (n).  \numberthis \label{Total_Push_line}
	 \end{align*}
	 Finally, putting \refeq{Total_BFS} and \refeq{Total_Push_line} together,  \emph{DLC-Extending} completes its course in a total cost $T$ computed by, 
    \begin{align*}
    T &= T_{1} + T_{2} \\
    &= \frac{n\sqrt{n} - n}{2} + n +\sqrt{n} = \frac{n\sqrt{n} - n +2n + 2\sqrt{n}}{2} = \frac{n\sqrt{n} - n + 2\sqrt{n}}{2}\\ 
    & = O(n \sqrt{n}).
    \end{align*} 	
\end{proof}

\subsection{An $O(n \log n)$-time Transformation}
\label{sub:nlogn_S}

We now investigate another approach (called \emph{DL-Doubling}) for {\sc DiagonalToLine} (i.e., without necessarily preserving connectivity). The main idea is as follows. The initial configuration can be viewed as $ n $ lines of length 1. We start (in \emph{phases}) to successively double the length of lines (while halving their number) by matching them in pairs through shortest paths, until a single spanning line remains.  
Let the lines existing in each phase be labelled $1,2,3,\ldots$ from top-right to bottom-left. 
In each phase, we shall distinguish two types of lines, \textit{free} and \textit{stationary}, which correspond to the odd  ($ 1, 3, 5, \ldots$) and  even ($2,4,6,\ldots$) lines from top-right to bottom-left, respectively. In any phase, only the \textit{free} lines move, while the \textit{stationary} stay still. 
In particular, in phase $i$, every free line $j$ moves via a shortest path to merge with the next (top-right to bottom-left) stationary line $j+1$. This operation merges two lines of length $k$ into a new line of length $2k$ residing at the column of the stationary line. In general, at the beginning of every phase $i$, $1\leq i\leq \log n$, there are $n/2^{i-1}$ lines of length $2^{i-1}$ each. These are interchangeably free and stationary, starting from a free top-right one, and at distance $2^{i-1}$ from each other. The minimum number of steps by which any free line of length $k_i$, $1\leq k_i\leq n/2$, can be merged with the stationary next to it is roughly at most $4k_i=4\cdot 2^{i}$ (by two applications of turning of Lemma \ref{lem:Transfer_Line_H_to_V}). By the end of phase $i$ (as well as the beginning of phase $i+1$), there will be $n/2^{i}$ lines of length $2^{i}$ each, at distances $2^{i}$ from each other. The total cost for phase $i$ is obtained then by multiplying $n/2^{i}$ free lines, each is paying at most $4\cdot 2^{i}$ to merge with the next stationary, thus, a linear cost in each one of $\log n$ phases in total.    
 
See Figure~\ref{fig:FreeStationary_PhaseOne} for an illustration of \emph{DL-Doubling}.
   
\begin{figure}[h!t]
	\centering
	\includegraphics[scale=0.6]{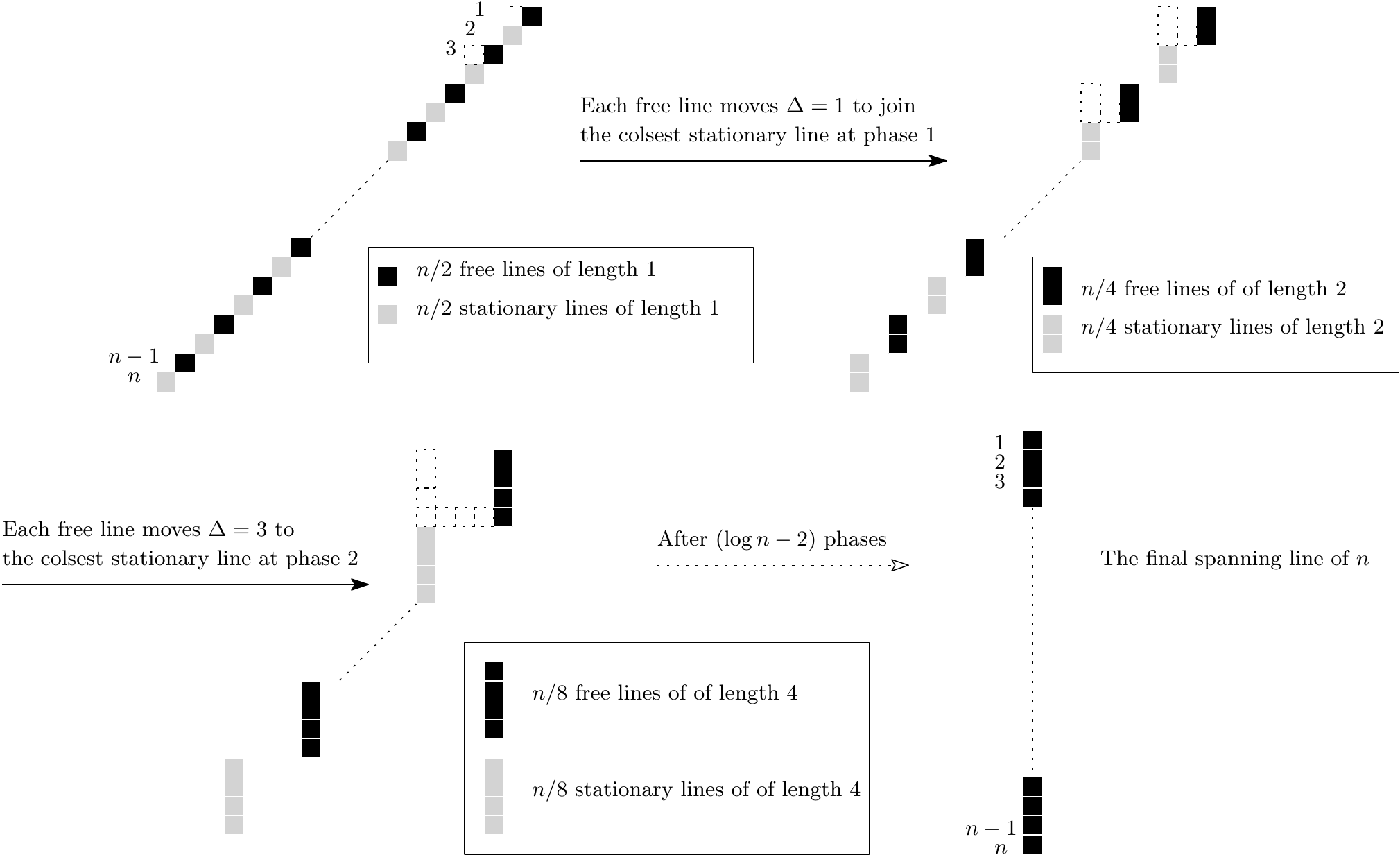}
	\caption{The process of the $O(n \log n)$-time \emph{DL-Doubling}. Nodes reside inside the black and grey cells.}
	\label{fig:FreeStationary_PhaseOne}	
\end{figure}

\subsubsection{Formal Description}
\label{subsubsec:DL-Gathering}  

Consider an initial diagonal line $S_{D}$ of $ n $ nodes occupying cells $(x_1,y_1), \ldots, (x_n, y_n)$, therefore, we can define the odd and even nodes into two different types of line of length 1, \textit{free} and \textit{stationary}. In any phase of \emph{DL-Doubling}, the permission of move is only given to \textit{free} lines, i.e., \textit{stationary} lines cannot move in that phase.  The $ \ceil{\frac{n}{2}}$ nodes in $(x_1,y_1), (x_3,y_3), \ldots, (x_{n-1},y_{n-1})$ are  \textit{free}, on the contrary, $(x_2,y_2), (x_4,y_4), \ldots,  (x_n,y_n)$ are occupied by $ \ceil{\frac{n}{2}}$ \textit{stationary} lines. Since \emph{DL-Doubling} keeps doubling the length of lines from 1 up to $n$, it consequently lasts for $\log n$ phases. In the first phase $i = 1$,  $1\leq i\leq \log n$, each of the $ \ceil{\frac{n}{2}}$ \textit{free} lines of length 1 moves a distance of $\Delta = 1$ to their left to occupy $(x_2,y_1), (x_4,y_3), \ldots, (x_n,y_{n-1})$, in order to merge with the next following $ \ceil{\frac{n}{2}}$ \textit{stationary} lines. As a result, we obtain $\ceil{\frac{n}{2}}$ (\textit{vertical}) lines of length 2, where their bottommost nodes occupying $(x_2,y_1), (x_4,y_3), \ldots, (x_n,y_{n-1})$, as depicted in the top of Figure~\ref{fig:FreeStationary_PhaseOne}. \emph{Due to symmetry, it is sufficient to show transformations occurring on one orientation, i.e., horizontal doubling holds}.

In the second phase, we divide the \textit{vertical} $\ceil{\frac{n}{2}}$ lines of length 2 each into $\ceil{\frac{n}{4}}$  \textit{free} and $\ceil{\frac{n}{4}}$ \textit{stationary}  lines interchangeably, starting from a free top-right one at $ (x_{n-2},y_{n-2}) $, and at distance $\Delta = 2^{2-1} = 2$ from each other. Now,  the $\ceil{\frac{n}{4}}$ \textit{free} lines must move at least $\Delta = 3$ steps to line up with the next following \textit{stationary} line to the left. Consequently, all $\ceil{\frac{n}{4}}$ \textit{free}  joint \textit{stationary} lines, which forms $\ceil{\frac{n}{4}}$  (\textit{vertical}) lines of length 4 each. Eventually, when \emph{DL-Doubling} repeats this process $\log n $ phases, it will form the goal spanning line $S_{L}$ of order $ n $ by the end of phase $\log n$. 

Here, we should also take into account all turns and delays of which a \textit{free} line asks for moving and reallocating above the next following \textit{stationary} line to the left. In particular, each  $ free $ line must first convert from \textit{vertical} into \textit{horizontal}, move one further step to fill the empty cell above the topmost node of the \textit{stationary} and finally converts again into \textit{vertical}. By Lemma \ref{lem:Transfer_Line_H_to_V}, at any phase $i$, for all $1\leq i\leq \log n$, a \textit{free} line of length $k_i$, where $1\leq k_i\leq n/2$, performs at most $4k_i = 4 \cdot 2^i$ steps to merge with closest \textit{stationary} line. Therefore, the following lemma shows that holds in all  $\log n $ phases.

\begin{lemma} \label{lem:Third_Strategy_a_to_a+1}  
		By the end of phase $i$, for all $1 \le i \le \log n $, \emph{DL-Doubling} has created $ n/2^{i} $ lines, each of length $ 2^{i} $, by performing $O( n )$ steps in that phase.			   				   	 	
\end{lemma} 
\begin{proof}	
	We use induction to prove this argument. Following the above reasoning of \emph{DL-Doubling} in Section \ref{subsubsec:DL-Gathering}, by the end of first phase (base case),  each of $ \frac{n}{2} $ \textit{free} lines of length 1 moves only $\Delta = 1$ step to form  $ \frac{n}{2} $ lines of length 2, which costs at most $1 \cdot \frac{n}{2} = O(n)$ steps in total. Hence, the base case is true for phase 1. Now, assume that holds in phase $i$, for all $1 \le i \le \log n $, where each free line of length $2^{i-1}$ pays at most $ 4\cdot 2^{i-1} =  2^{i+1} -3 $ steps to match with the closest \textit{stationary} lines (two applications of Lemma  \ref{lem:Transfer_Line_H_to_V}), therefore, a total cost $ t_{i} $ required to form $\frac{n}{2^{i}}$ lines of length $ 2^{i} $ in phase $i$ is given as follows;  
	\begin{align*}
		t_{i} &= \frac{n}{2^i} \cdot \big( 2^{i+1} -3 \big) \\
		& =  \frac{2^{i+1} \cdot n}{2^i}  - \frac{3n}{2^{i}}  = 2n  -\frac{3n}{2^i} = n \bigg( 2 - \frac{3}{2^i}\bigg)\\
		 & = O(n),       \numberthis  \label{FreeLinesMoveInPhase}
	\end{align*}	
			
	Now, we prove that this must hold in phase $i+1$,
	\begin{align*}
	t_{i+1} &= \frac{n}{2^{i+1}} \cdot \big( 2^{i+2} -3  \big) \\
	& =  \frac{2^{i+2} \cdot n}{2^{i+1}}  - \frac{3n}{2^{i+1}}  = 2n  -\frac{3n}{2^{i+1}} = n ( 2 - \frac{3}{2^{i+1}})\\
	& = O(n).       \numberthis  \label{FreeLinesMoveInPhase_i+1}
	\end{align*}	 
	
	As a result, our assumption is also true for phase $i$, therefore, from \refeq{FreeLinesMoveInPhase} and  \refeq{FreeLinesMoveInPhase_i+1},  we conclude that  \emph{DL-Doubling} pays at most $ O(n) $ steps in each phase.	         
\end{proof}

Utilising  Lemma~\ref{lem:Third_Strategy_a_to_a+1}, we can now formulate the following:

\begin{theorem} \label{prop:Third_Strategy} 
	\emph{DL-Doubling} transforms any diagonal $S_{D}$ of order n into a line $S_{L}$ in $O(n\log n)$ steps.       	
\end{theorem}
\begin{proof}
  \emph{DL-Doubling} constructs the goal spanning line $S_{L}$ of length $n$ in $ \log n $ phases, while $S_{L}$ increases exponentially in each phase. Given that, the total running time $T$ of this transformation is computed by summing all steps of \refeq{FreeLinesMoveInPhase} over all $\log n$ phases, as follows:
	\begin{align*}
	T &= \sum_{i=1}^{\log n} t_{i}\\
	& =\sum_{i=1}^{\log  n} n( 2 - \frac{3}{2^{i}}) \\
	& = 2n \log n - 3n \sum_{i=1}^{\log  n} \frac{1}{2^{i}},  \numberthis  \label{AllFreeLinesMoveInPhase}	 
	\end{align*}

	Let us compute the right summation of \refeq{AllFreeLinesMoveInPhase},
	\begin{align*}
	& \sum_{i=1}^{\log n}\frac{1}{2^{i}} = (\frac{1}{2} + \frac{1}{4} +  \frac{1}{8} + \ldots+ \frac{1}{2^{\log n}}),  \numberthis  \label{AllFreeSUM}	 
	\end{align*}
	
	By multiplying \refeq{AllFreeSUM} by 2 and subtracting it again by itself,
	
	\begin{align*}
	&( 1 + \frac{1}{2} + \frac{1}{4} +  \frac{1}{8} + \ldots  + \frac{1}{2^{\log (n-1)}}) - (\frac{1}{2} + \frac{1}{4} +  \frac{1}{8} + \ldots +  \frac{1}{2^{\log (n-1)}} + \frac{1}{2^{\log (n)}}) \\
	&= 1 - \frac{1}{2^{\log (n)}},   \numberthis  \label{AllFreeSUM_1}	
	\end{align*}      
	
	Then, by  plugging \refeq{AllFreeSUM_1} into \refeq{AllFreeLinesMoveInPhase}, this yields,
	\begin{align*}
		& 2n \log n -  3n (1 - \frac{1}{2^{\log n}}) =   2n \log n  - 3n  + \frac{3n}{2^{\log n}} = n (2 \log n + \frac{3}{2^{\log n}} -3)  \\
		&  = O(n \log n). 
	\end{align*} 	
	
	Thus, it has been proved that \emph{DL-Doubling} transforms  $S_{D}$ of $ n $ nodes into  $S_{L}$ in a total of $ O(n \log n) $ steps.
\end{proof}

\subsection{An $ O(n \log n) $-time Transformation Based on Recursion}
\label{subsec:DLR-Gathering}

An interesting observation for {\sc DiagonalToLine} (i.e., without necessarily preserving connectivity), is that the problem is essentially self-reducible. This means that any transformation for the problem can be applied to smaller parts of the diagonal, resulting in small lines, and then trying to merge those lines into a single spanning line. An immediate question is then whether such recursive transformations can improve upon the $O(n\log n)$ best upper bound established so far. The extreme application of this idea is to employ a full uniform recursion (call it \emph{DL-Recursion}), where $S_D$ is first partitioned into two diagonals of length $n/2$ each, and each of them is being transformed into a line of length $n/2$, by recursively applying to them the same halving procedure. Finally, the top-right half has to pay a total of at most $4(n/2)=2n$ to merge with the bottom-left half and form a single spanning line (and the same is being recursively performed by smaller lines).

More formally, consider a diagonal $S_D$ of $ n $ nodes where the bottom-left and top right nodes occupy $(x_{1}, y_{1})$ and $(x_{n}, y_{n})$, respectively. Then, the goal is to collect all nodes at the leftmost column, say  $  x_{n} $. The collection can be arranged in a recursive way by creating stop points (partitions) on $S_D$ in which each stop point always creates equal partitions of the same length. This can be parametrized by $ \frac{n}{x} $ for each partition, where $ 2 \le x \le n $. For example, if $ x = 2 $, we have a stop point that halves $S_D$ into $ 2 $ partitions of length $ \ceil{\frac{n}{2}} $. As a consequence, the first node on the top  will stop at the middle of  $S_D$ and wait for all nodes to its right to gather at that point (\textit{column}) and then continue directly to the gathering (\textit{column}) $ x_{n} $.

Now, let us repeat the same precess on each of the $ x $ partitions recursively, by considering the partition as a diagonal of length roughly $ \frac{n}{x} -1 $, which is divided into $ x $ sub-partitions each of length $ \frac{n}{x^2} $ roughly. Every recursion shrinks the partitions by a factor of $ \frac{1}{x} $. For example, in the $ x = 2 $ case, we halve the length of the partitions every time we subdivide, therefore, we will end the recursion when it arrives at partitions of length $ 1 $, which will happen after $ \log n $ repetitions. That occurs similarly for the general $ x $ case by simply end after $ \log_x n  $ repetitions. For example, Figure~\ref{fig:Recursion_bound_proof_step_1} demonstrates this procedure. 
	  
		\begin{figure}[h!t]
			\centering
			\includegraphics[scale=0.6]{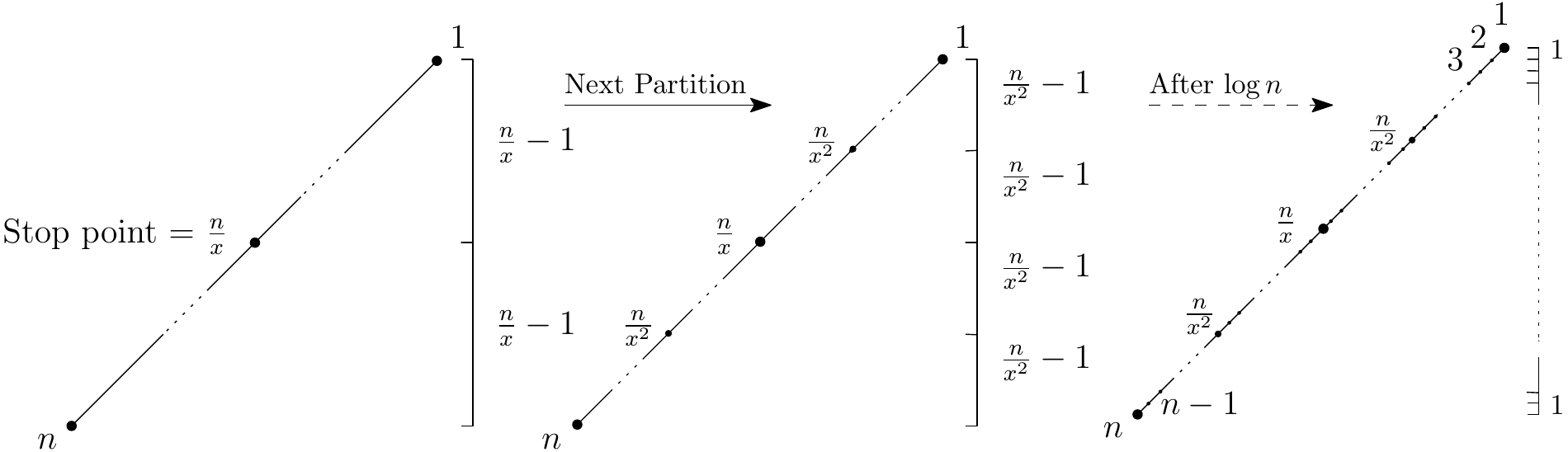}
			\caption{Shows all steps of subdividing the diagonal  $S_D$ recursively by a factor of $\frac{1}{x}$, where $x = 2$.}
			\label{fig:Recursion_bound_proof_step_1}	
		\end{figure}	
	\begin{figure}[h!t]
		\centering
		\includegraphics[scale=0.65]{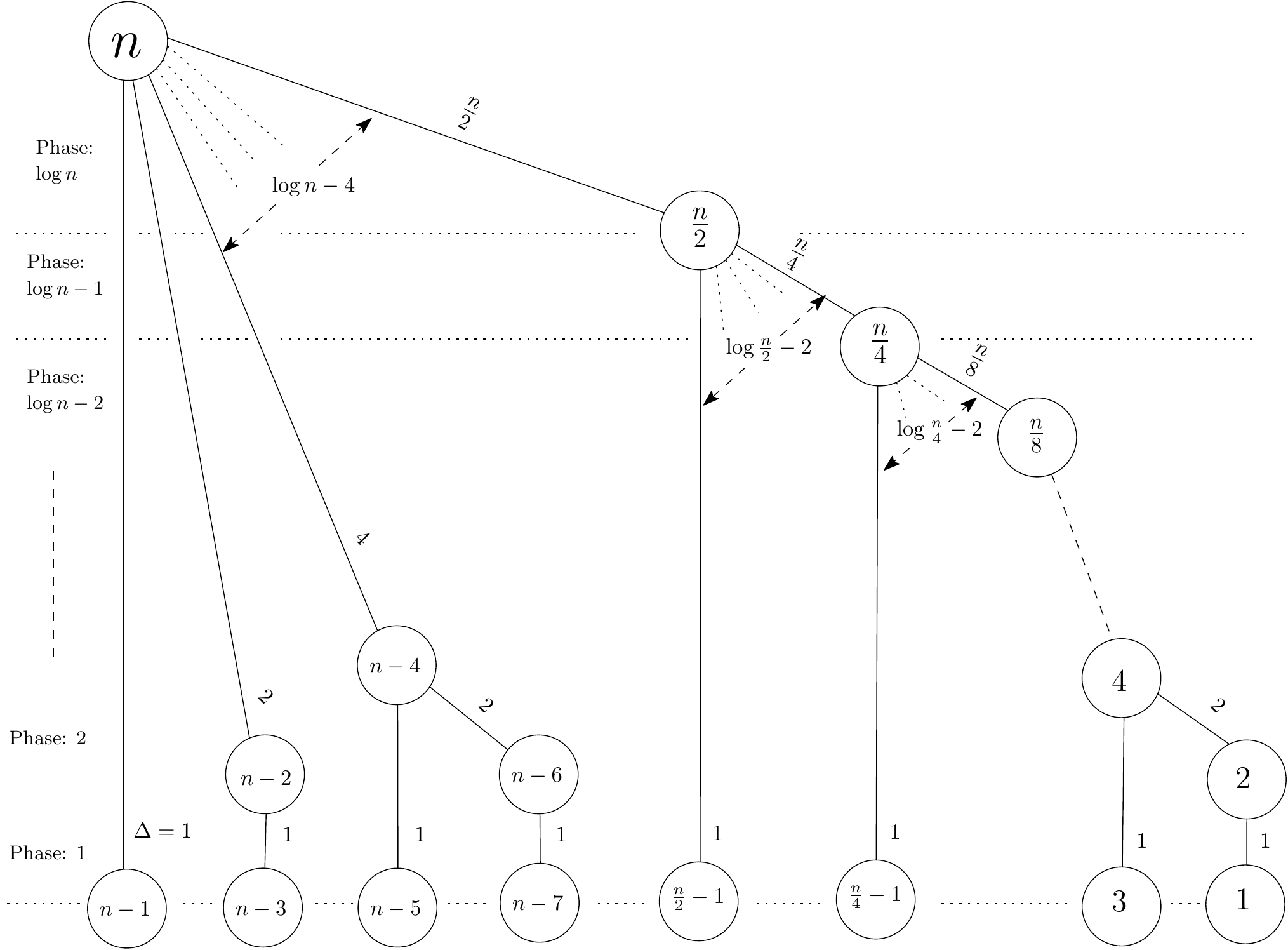}
		\caption{The underlying tree representation of a recursive partitioning of a diagonal. Edges are weighted by the minimum distance ($\Delta$) between nodes. See the text for more explanation.}
		\label{fig:Recursion_Tree_bound.pdf}	
	\end{figure} 	
	
Next, we draw an abstract underlying tree of the partitioning process to trace all necessary computations required to travel from the diagonal  $S_D$ into a bottommost left column $ x_{n} $. Figure \ref{fig:Recursion_Tree_bound.pdf} presents the tree of $ n $ nodes and weighted edges indicate the minimum distances (shortest paths) $\Delta$ between them, with a degree and depth of $\log_2 n$, which is also the number of phases that are needed to cease the segmentation recursively. Here, node $ n $ is the root of the tree occupies the target column at $ x_{n} $, and it has $\log n$ child nodes (\textit{stop points}) $ n-1, n-2, n-4, \ldots ,\frac{n}{2}$ of distances $  \Delta = 1 , 2, 4, \ldots, \frac{n}{2}$, respectively. Now, the distance $\Delta$ between a parent $u$ and child node $v$ defines two basic properties: 1) the number of sub-child nodes (siblings) of $v$, namely each child node $v$ gets $ \log \Delta(u,v)$ sub-child nodes, and 2) the maximum cost by which a child node $v$ requires to merge with its parent $u$, and it is computed by $ (2^{\Delta +1} -3) $ (a reader may consult Lemma~\ref{lem:Third_Strategy_a_to_a+1} and Theorem~\ref{prop:Third_Strategy}). For example, the $\frac{n}{2}$ child node needs $ (2^{\frac{n}{2} +1} -3) $ steps to join its parent node $ n $ in a distance of $\Delta = n - \frac{n}{2} = \frac{n}{2}$, and at the same time, this tells us that this child node is also holding $ \log \frac{n}{2} $ sub-child nodes. On the contrary, $n-1$ node got no sub-child since $ \log 1 = 0 $, so it is a leaf and requires $ (2^{1 + 1} -3) = 1 $ steps to reach its parent node $ n $. The same idea follows for other sub-trees, such as $ \frac{n}{2}, \frac{n}{4}, \frac{n}{8}, \ldots, 2$.

 Having said that, the steps of any transformation strategy that solves the above recursion problem can generally be reordered without affecting the cost, and each stop point takes place in the reordered version. That is because the abstract tree representation remains conveniently invariant, (\emph{inherits the same cost}), for any arrangement by which a transformation can exploit to collect nodes into the target point (\emph{column} or \emph{row}).  However, this recursion proceeds in $ \log n $ \emph{cost} phases, where in each phase $ i $,  for all $1 \le i \le \log n $, we upper bound a cost that the transformation $ A $ pays at most to move all nodes in each phase (even though $ A $ may move them in an entirely different order). In phase 1, $ A $ works on nodes $1, 3, 5, \ldots , n-1$, each of which posses roughly a single node (leaf) and eventually moves all of them. Therefore, no matter in what order they move, $ A $ performs at most $ O(\frac{n}{2}) $ steps to transfer all leaves in the first phase of the tree. 
 
 In the second phase, nodes $2, 4, 6, \ldots, (n-2)$ in the tree are occupied by $ 2 $ nodes (a line of length 2). Those pairs of nodes principally move at some point (even together or as part of other nodes going through them) to its parent nodes (\emph{the next occupied column or row}), with paying a cost of at most $(2^{2 +1} -3) = 5$ each. Repeat the same argument for the rest phases of the tree, the upper bound now is based on the following observation: whenever $ k $ nodes, where $1 \le k < n $, all move as maximum as $(2^{k+1} -3)$ steps to merge with the next stop point (parent node). In other words, suppose that $ k $ nodes are a line of $ k $ nodes, it walks at most $(2^{k+1} -3)$ steps to merge with a line in the next occupied row or column.
		
 Generally, we can say that at any phase $i$ of any transformation $ A $ solves the recursion problem, where $ 0 \le i \le \log n $, there is a node $v$ with other $2^{i}$ nodes occupying the same \emph{column} (\emph{row}), where the distance between $v$ and its parent $u$ (the next occupied \emph{column} (\emph{row}) is $ \Delta(u,v) = 2^{i}  $. Then, our purpose is to upper bound $ A $’s cost for $v$ to reach and merge with $u$; therefore, $v$ walks at most $(2^{i+1} - 3)$ steps to form a line with $ u $. 
		
By analysing the running time of such a uniform recursion, we obtain that it is still $O(n\log n)$, partially suggesting that recursive transformations might not be enough to improve upon $O(n\log n)$ (also possibly because of an $\Omega(n\log n)$ matching lower bound, which is left as an open question). If we denote by $T_k$ the total time needed to split and merge lines of length $k$, then the recursion starts from 1 line incurring $T_n$ and ends up with $n$ lines incurring $T_1$. In particular, we analyse the recurrence relation:
\begin{align*}
 T_n& =2 \cdot T_{n/2}+2n = 2(2 \cdot T_{n/4}+n)+2n = 4 \cdot T_{n/4}+4n = 4(2 \cdot T_{n/8}+n/2)+4n\\
 & = 8T_{n/8}+6n=\cdots=2^i \cdot T_{n/2^i}+2i\cdot n=\cdots=2^{\log n} \cdot T_{n/2^{\log n}}+2(\log n)n 
\end{align*} 
Since $T_1 = 1$, we get,
\begin{align*}
 T_n&= n\cdot T_1+2n(\log n)= n+2n(\log n), \\
& = O(n\log n).
\end{align*} 
	
	Finally, we give the following theorem, 
	
   \begin{theorem}
		\emph{DL-Recursion} transforms any diagonal $S_{D}$ of order $n$ into a line $S_{L}$ of the same order in $O(n\log n)$ steps.
		\label{prop:Upper_Bound_OF_gathering}	
	\end{theorem}


\section{Universal Transformations}
\label{sec:universal}

\subsection{An $O(n \sqrt{n})$-time Universal Transformation}
\label{subsec:Universal_n_root(n)}

In this section, we develop a universal transformation, called \emph{U-Box-Partitioning}, which exploits line movements in order to transform \emph{any} initial connected shape $S_{I}$ into \emph{any} target shape  $S_{F}$ of the same order $n$, in $ O(n \sqrt{n}) $ steps. Due to reversibility (Lemma \ref{lem:Transferability_LineMove}), it is sufficient to show that any initial connected shape $S_I$ can be transformed into a spanning line (implying then that any pair of shapes can be transformed to each other via the line and by reversing one of the two transformations). We maintain our focus on transformations that are allowed to break connectivity during their course. 

Observe that any initial connected shape $S_I$ of order $ n $ can be enclosed in an appropriately positioned $n \times n$ square (called a \textit{box}). Our universal transformation is divided into three phases:\\

\textbf{Phase A}: Partition the $n \times n$ box into $\sqrt{n} \times \sqrt{n}$ sub-boxes ($n$ in total in order to cover the whole $n \times n$ box). For each sub-box move all nodes in it down towards the bottommost row of that sub-box as follows. Start filling in the bottommost row from left to right, then if there is no more space continue to the next row from left to right and so on until all nodes in the sub-box have been exhausted (resulting in zero or more complete rows and at most one incomplete row). Moving down is done via shortest paths (where in the worst case a node has to move distance $2\sqrt{n}$); see Figure \ref{fig:brute-force_line_formation}.\\ 

\textbf{Phase B}: Choose one of the four length-$ n $ boundaries of the $n \times n$ box, say w.l.o.g. the left boundary. This is where the spanning line will be formed. 
Then, transfer every line via a shortest path to that boundary (incurring a maximum distance of $n-\sqrt{n}$ per line). \\

\textbf{Phase C}: Turn all lines (possibly consisting of more than one line on top of each other), by a procedure similar to that of  Figure \ref{fig:S1_nrootn_partition} (e), to end up with a spanning line of $ n $ nodes on the left boundary.\\

However, there are two variants of $\sqrt{n} \times \sqrt{n}$ sub-boxes:  

\begin{enumerate}
	\item \textit{Occupied sub-box}: Denoted by $ s $ and contains $k$ nodes of $S_{I}$, where $1 \le k \le n$.
	\item \textit{Unoccupied sub-box}: An empty sub-box (has no nodes).  
\end{enumerate}

Now, we show some important properties of \textit{occupied sub-boxes}. Given an \textit{occupied sub-box} $s$ of $k$ nodes, where $1 \le k \le n$, therefore, the maximum number of lines which can be formed inside $s$ is at most $\ceil{ \frac{k}{\sqrt{n}}}$. As mentioned above, those $k$ lines can be aligned horizontally at bottommost rows or vertically at leftmost columns of the occupied sub-box, such as Figure \ref{fig:brute-force_line_formation}. The following lemma gives an upper bound on the number of sub-boxes that any connected $S_{I}$ can occupy.  

\begin{lemma} \label{lem:Number_of_Occupied_Sub_Boxs}
Any connected shape $S_{I}$ of order $n$ occupies at most $ O(\sqrt{n}) $ sub-boxes.
\end{lemma}
\begin{proof}
	It follows directly from Corollary \ref{cor:partitioning}, which states that for a given  connected shape $S_{I}$ of $ n $ nodes enclosed by a square box of size $n \times n$ and any uniform partitioning of that box into sub-boxes of dimension $d$, then, it holds that $S_{I}$ can occupy at most $O(\frac{n}{d})$ sub-boxes. Here, \emph{U-Box-Partitioning} is dividing the $n \times n$ square box into $\sqrt{n} \times \sqrt{n}$ sub-boxes of dimension $d = \sqrt{n}$, therefore, $S_{I}$ occupies at most $\frac{n}{\sqrt{n}} = O(\sqrt{n})$ sub-boxes.	
\end{proof}

Below, we prove correctness and analyse the running time of phase A.
\begin{lemma}  \label{lem:brute_force_line_formation_of_Occupied_Sub_Boxs}
		Starting from any connected shape $S_{I}$ of order $n$, Phase A completes in $ O(n\sqrt{n}) $ steps each.    
\end{lemma}
\begin{proof}
	By Lemma \ref{lem:Number_of_Occupied_Sub_Boxs}, let $S_{I}$ be a connected shape of $n$ nodes occupies  $\sqrt{n}$ sub-boxes of size $\sqrt{n} \times \sqrt{n}$ each, and $s \in S_I$ be any \textit{occupied sub-box}  of $ k $ nodes, where $ 1 \le k \le n$. Then, $s$ performs a trivial line formation to collect all $  k $ nodes at its bottommost (\textit{or leftmost})  boundary.  Consider the worst case of a node occupies the top-right corner and wants to gather at the bottom-left of $s$, therefore, it is incurring a distance (\textit{shortest path}) of at most $\Delta = 2\sqrt{n}$ to arrive there. Now, the $k$ nodes fill in the $\sqrt{n}$ bottommost row from left to right, and then start filling in the next row from left to right and so on until all nodes in $s$ is exhausted. For example, Figure~\ref{fig:brute-force_line_formation} shows the line formation at the bottommost rows of a  $6 \times 6$ sub-box of  $k = 11$ nodes. Consequently, $s$ forms at least one \textit{complete} line of length $\sqrt{n}$ or one \textit{incomplete} of less than $\sqrt{n}$. Recall that $S_{I}$ is connected. So there are at most $\sqrt{n}$ \textit{occupied sub-boxes}, which means there are at most $ \frac{n}{\sqrt{n}} = O(\sqrt{n})$ lines (\emph{complete or incomplete}) inside all those \textit{occupied sub-boxes}. This is trivial to prove, assume that  the $\sqrt{n}$ \textit{occupied sub-box} formed more than $\sqrt{n}$ lines, resulting in $|S_{I}| > n$, which is a contradiction.  
	
	\begin{figure}[h!t] 
		\centering
		\includegraphics[scale=0.5]{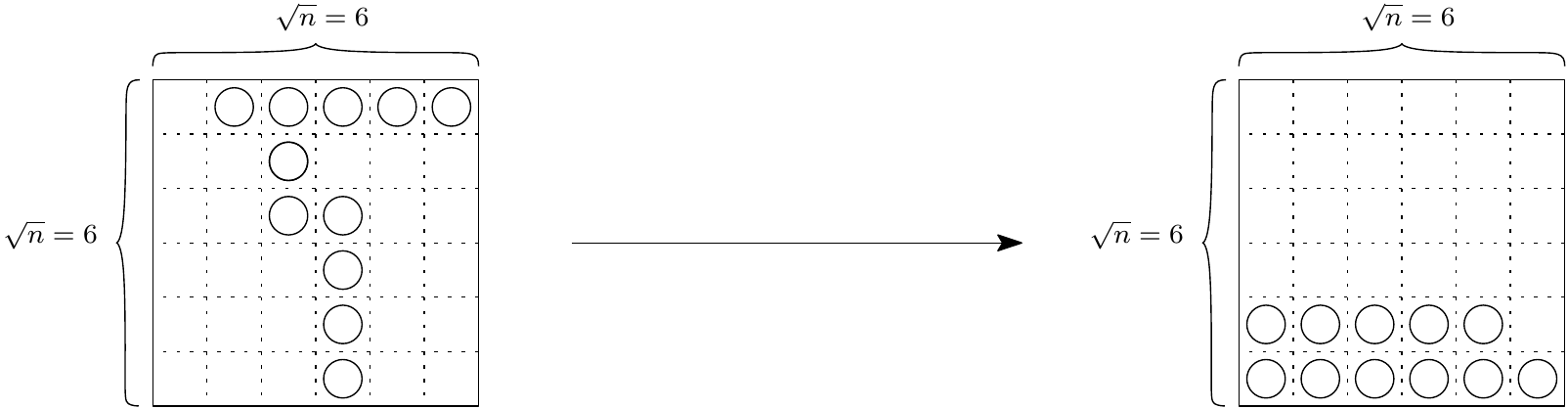}
		\caption{An example of a brute-force line formation to collect all $ k $ nodes at bottommost rows of a sub-box of size $6 \times 6$ containing $k=11$ nodes.}  
		\label{fig:brute-force_line_formation}		       	   		
	\end{figure} 	
	 
	 With that, the total cost $t_{1}$ required to form a line $l$ of $w$ nodes, for all $1 \le l,w \le \sqrt{n}$,  inside $s$ is:
	 \begin{align*}
	 t &   =  w \cdot 2\sqrt{n} = \sqrt{n}   \cdot 2\sqrt{n} = 2n\\
	     & = O(n),   \numberthis \label{nsqrtn_11}       
	 \end{align*}
	  
	Multiply \refeq{nsqrtn_11} by $ \sqrt{n} $, to obtain total steps $T_1$ to form all $ \sqrt{n} $ lines inside all \textit{occupied sub-boxes};
	\begin{align*}
		T_1 &   =   \sqrt{n} \cdot t\\
		&     = \sqrt{n} \cdot 2n = 2n\sqrt{n} \\
		&  = O(n\sqrt{n}).     \numberthis \label{nsqrtn_1}       
	\end{align*}

	Finally, we conclude that Phase A finishes its course in  $O(n\sqrt{n})$ steps.	
\end{proof}

In phase B, set any (length-$n$) boundary of the $n \times n$ square box as the \textit{gathering boundary} of all lines formed inside the \textit{occupied sub-boxes} in phase A. Then, the following lemma computes the total steps required to transfer all those lines to that \textit{gathering boundary}.

\begin{lemma}  \label{lem:Transfer_Lines_to_Gathering_border}   
	Starting from any connected shape $S_{I}$ of order $n$, Phase B completes in $ O(n\sqrt{n}) $ steps.    	
\end{lemma}
\begin{proof}
	It follows from Lemmas \ref{lem:Number_of_Occupied_Sub_Boxs} and \ref{lem:brute_force_line_formation_of_Occupied_Sub_Boxs}. Let $S_{I}$ be a connected shape of order $ n $ enclosed by a $n\times n$ box and then partitioned into $ \sqrt{n} \times \sqrt{n} $ \textit{occupied sub-boxes} of $k$ nodes each, where $ 1 \le k \le n$. By phase A, there are $l$ lines, for all $1 \le l\le \sqrt{n}$, formed inside all $\sqrt{n}$ \textit{occupied sub-boxes}. Without loss of generality, define the left border of the $n\times n$ box as the gathering boundary for all those lines. Now, the distance between any \textit{line} inside \textit{an occupied sub-box} and the defined boundary is no longer than $n - \sqrt{n}$. Thus, the number of steps (distance $\delta$) by which a \textit{line} requires  to reach that boundary is therefore $\delta \le n - \sqrt{n} $. As there are at most $ \sqrt{n} $ lines in the \textit{occupied sub-boxes}, then the total steps $T_2$ for all $l$ lines to transfer and arrive at the left border is at most,
	\begin{align*}
	T_{2} &= l \cdot \delta\\
		&= \sqrt{n} \cdot (n-\sqrt{n}) = n\sqrt{n} - n \\
		&=O(n\sqrt{n}).    \numberthis   \label{nsqrtn_3}  
	\end{align*}	 
\end{proof}

By the end of phase B, all $\sqrt{n}$ lines have transferred and arrived at the length-$n$ gathering boundary, where each contributes a node to that boundary. Therefore, in phase $C$, all  those lines ought to push and include themselves to the length-$n$ border, in order to form the goal spanning line of length $ n $ nodes. Formally, we give the following Lemma.

\begin{lemma} 	\label{lem:Linear_rearrangement_Of_Small_Lines} 
	Consider any length-$n$ boundary and $n$ nodes forming $k$ lines, where $1 \le k \le n$,  that are perpendicular to that boundary. 
	Then, by line movements, the $k$ lines require at most $O(n)$ steps to form a line of length $ n $ on that boundary. This implies that Phase C is completed in $O(n)$ steps.
\end{lemma}
\begin{proof}
	See the example in Figure~\ref{fig:Linear_rearrangement_Of_Small_Lines}, where  the length-$n$ gathering border denoted by the dashed line, and the lines $\{l_{1}, l_{2}, \ldots , l_{k}\}$ of $n$ nodes are depicted by bold black lines, where $1\le k \le \sqrt{n}$. Without doubt, the $n$ nodes shall completely fill up the border of length $\Delta = n$. Now, pick the first line $ l_{1} $ of $k_{1}$ nodes and start to push $k_{1}$  into the topmost point of the boundary, until either 1)  $k_{1}$ are exhausted, or 2) reaching the topmost point of the boundary and still have nodes waiting to be pushed, therefore, it can easily begin to push them down towards the bottommost point of the boundary. By performing the same strategy for all other lines, they shall eventually fill in the length-$n$ border completely by $n$ nodes.  
	    	
	Observe that, in our case, the $ k $ lines connect to length-$n$ boundary by $k$ nodes. Therefore, there is still $n-k$ nodes need to be pushed into the boundary. According to Lemma \ref{lem:Transfer_Line_H_to_V}, each of the $n-k$ nodes requires 2 steps to be included, therefore all $n-k$ nodes shall perform a total of $2(n-k)$ steps to fully filled up the boundary of length $ n $. Following that, for any $k$ lines of $n$ nodes that are connected perpendicularly to a border of length-$n$, \emph{U-Box-Partitioning} pushes the $n$ nodes of $k$ lines into the length-$n$ border by line mechanisms in a total $T_3$ of at most,
	\begin{align*}
	T_{3} &= 2(n-k) = 2(n-\sqrt{n}) \\
	& =  O(n).   \numberthis   \label{nsqrtn_4}  
	\end{align*}

	\begin{figure}[h!t]
		\centering
		\includegraphics[scale=0.6]{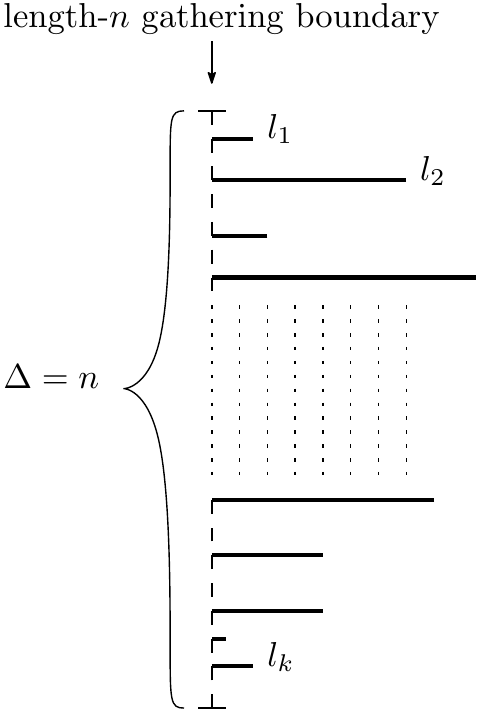}
		\caption{The dashed line indicates the length-$n$ gathering boundary of the $n\times n$ box, whilst the bold black lines represent the $ k $ lines of $n$ nodes.  }
		\label{fig:Linear_rearrangement_Of_Small_Lines}	
	\end{figure} 
\end{proof}

Now, we prove that any connected shape $ S_{I} $ transforms into a line $ S_{L} $ in  at most $ O(n \sqrt{n}) $ steps. 
\begin{lemma} \label{lem:Universal_nroot(n)_S_To_Line}
	\emph{U-Box-Partitioning} transforms any connected shape $S_I$ into a straight line $S_{L}$ of the same order $n$, in $ O(n \sqrt{n}) $ steps.	
\end{lemma}
\begin{proof}
	By the above \cref{lem:Number_of_Occupied_Sub_Boxs,lem:brute_force_line_formation_of_Occupied_Sub_Boxs,lem:Linear_rearrangement_Of_Small_Lines}, we sum \refeq{nsqrtn_1}, \refeq{nsqrtn_3} and \refeq{nsqrtn_4} to compute  the overall moves $T$ as follows;
	\begin{align*}
	T &=   T_1 + T_2 + T_3 \\
	& =    O(n\sqrt{n}) + O(n\sqrt{n}) + O(n) \\
	& =     O(n\sqrt{n}).   \numberthis   \label{nsqrtn_5}  
	\end{align*}	
	
	 This provides an upper bound $ O ( n\sqrt{n}) $ of \emph{U-Box-Partitioning} to transform any arbitrary connected shape $S_I$ into a single spanning line $S_{L}$ of the same number of nodes.				
\end{proof} 

Putting Lemma \ref{lem:Universal_nroot(n)_S_To_Line} and reversibility (Lemma \ref{lem:Transferability_LineMove}) together gives:

\begin{theorem} \label{theorem:Universal_n_root(n)}
	For any pair of connected shapes $S_I$ and $S_F$ of the same order $n$, transformation \emph{U-Box-Partitioning} can be used to transform $S_I$ into $S_F$ (and $ S_F $ into $S_I$) in $ O(n \sqrt{n}) $ steps.
\end{theorem}

\subsection{An $O(n \log n)$-time Universal Transformation}
\label{subsec:Universal_nlogn}

We now present an alternative universal transformation, called \emph{U-Box-Doubling}, that transforms any pair of connected shapes, of the same order, to each other in $O(n \log n)$ steps. Given a connected shape $S_{I}$ of order $ n $, do the following. Enclose $ S_{I} $ into an arbitrary $n \times n$ \emph{box} as in \emph{U-Box-Partitioning} (Section \ref{subsec:Universal_n_root(n)}). For simplicity, we assume that $n$ is a power of 2, but this assumption can be dropped. Proceed in $ \log n $ phases as follows: In every phase $ i $, where  $ 1 \le i \le \log n $, partition the $  n \times n $ box into $ 2^{i} \times 2^{i} $ sub-boxes, disjoint and completely covering the $  n \times n $ box. Assume that from any phase $i-1$, any $2^{i-1} \times 2^{i-1}$ sub-box is either empty or has its $k$, where $0\le k \le 2^{i-1}$, bottommost rows completely filled in with nodes, possibly followed by a single incomplete row on top of them containing $l$, where $1 \le l < 2^{i-1} $, consecutive nodes that are left aligned on that row. This case holds trivially for phase 1 and inductively for every phase. That is, in odd phases, we assume that nodes fill in the leftmost columns of boxes in a symmetric way. Every $2^{i} \times 2^{i}$ sub-box (of phase $ i $) consists of four $ 2^{i-1} \times 2^{i-1}$ sub-boxes  from phase $ i-1 $, each of which is either empty or occupied as described above.

Consider the case where $i$ is odd, thus, the nodes in the  $2^{i-1} \times 2^{i-1}$ sub-boxes are bottom aligned. For every $2^{i} \times 2^{i}$ sub-box, move each line from the previous phase that resides in the sub-box to the left as many steps as required until that row contains a single line of consecutive nodes, starting from the left boundary of the sub-box, as shown in Figure \ref{fig:Filling_LeftMost_SuBn} (a). With a linear procedure similar to that of Figure \ref{fig:Exmpel_folding4} (and of \textit{nice shapes}), start filling in the columns of the $2^{i} \times 2^{i}$ sub-box from the leftmost column and continuing to the right. If an incomplete column remains, push the nodes in it to the bottom of that column; see Figure \ref{fig:Filling_LeftMost_SuBn} (b) for an example. 

\begin{figure}
	\centering
	\subcaptionbox{}
	{\includegraphics[scale=0.7]{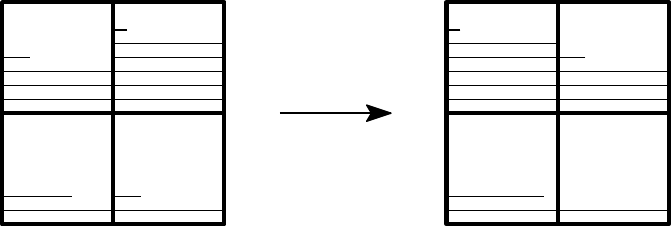}}	 \qquad \qquad \qquad
	\subcaptionbox{}
	{\includegraphics[scale=0.7]{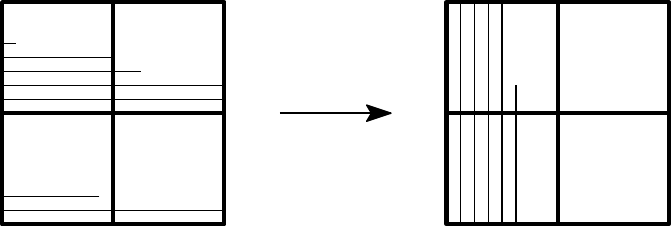}}	
	\caption{(a) Pushing left in each $2^{i} \times 2^{i}$ sub-box. (b) Cleaning the orientation by aligning (filling) the leftmost columns. }
	\label{fig:Filling_LeftMost_SuBn}	
\end{figure}

The case of even $i$ is symmetric, the only difference being that the arrangement guarantee from $i-1$ is left alignment on the columns of the $2^{i-1} \times 2^{i-1}$ sub-boxes and the result will be bottom alignment on the rows of the $2^{i} \times 2^{i}$ sub-boxes of the current phase. This completes the description of the transformation. We first prove correctness:

\begin{lemma} \label{lem:U_nlogn_V_OR_H}
	Starting from any connected shape $S_{I}$ of order $n$, \emph{U-Box-Doubling} forms by the end of phase $\log n $ a line of length $ n $. 
\end{lemma}
\begin{proof}
	In phase $\log n $, the procedure partitions into a single box, which is the whole original $n \times n$ box. Independently of whether gathering will be on the leftmost column or on the bottommost row of the box, as all $n$ nodes are contained in it, the outcome will be a single line of length $n$, vertical or horizontal, respectively. 
\end{proof}

Now, we shall analyse the running time of \emph{U-Box-Doubling}. To facilitate exposition, we break this down into a number of lemmas.

\begin{lemma} \label{lem:U_nlogn_ConnSuB}
	In every phase $ i $, the ``super-shape'' formed by the occupied $2^{i} \times 2^{i}$ sub-boxes is connected.
\end{lemma} 
\begin{proof}
	By induction on the phase number $i$. For the base of the induction, observe that for $i=0$ it holds trivially because the initial $S_{I}$ is a connected shape. Assuming that it holds for phase $i-1$, we shall now prove that it must also hold for phase $i$. By the inductive assumption, the occupied $2^{i-1} \times 2^{i-1}$ sub-boxes form a connected super-shape. Observe that, by the way the original $n \times n$ box is being repetitively partitioned, any box contains complete sub-boxes from previous phases, that is, no sub-box is ever split into more than one box of future phases.
	Additionally, observe that a sub-box is occupied iff any of its own sub-boxes (of any size) had ever been occupied, because nodes cannot be transferred between $2^{i} \times 2^{i}$ sub-boxes before phase $i+1$. Assume now, for the sake of contradiction, that the super-shape formed by $2^{i} \times 2^{i}$ sub-boxes is disconnected. This means that there exists a ``cut'' of unoccupied $2^{i} \times 2^{i}$ sub-boxes as in Figure \ref{fig:Cut_SuB}.
	\begin{figure}[h!t]
		\centering
		\includegraphics[scale=0.5]{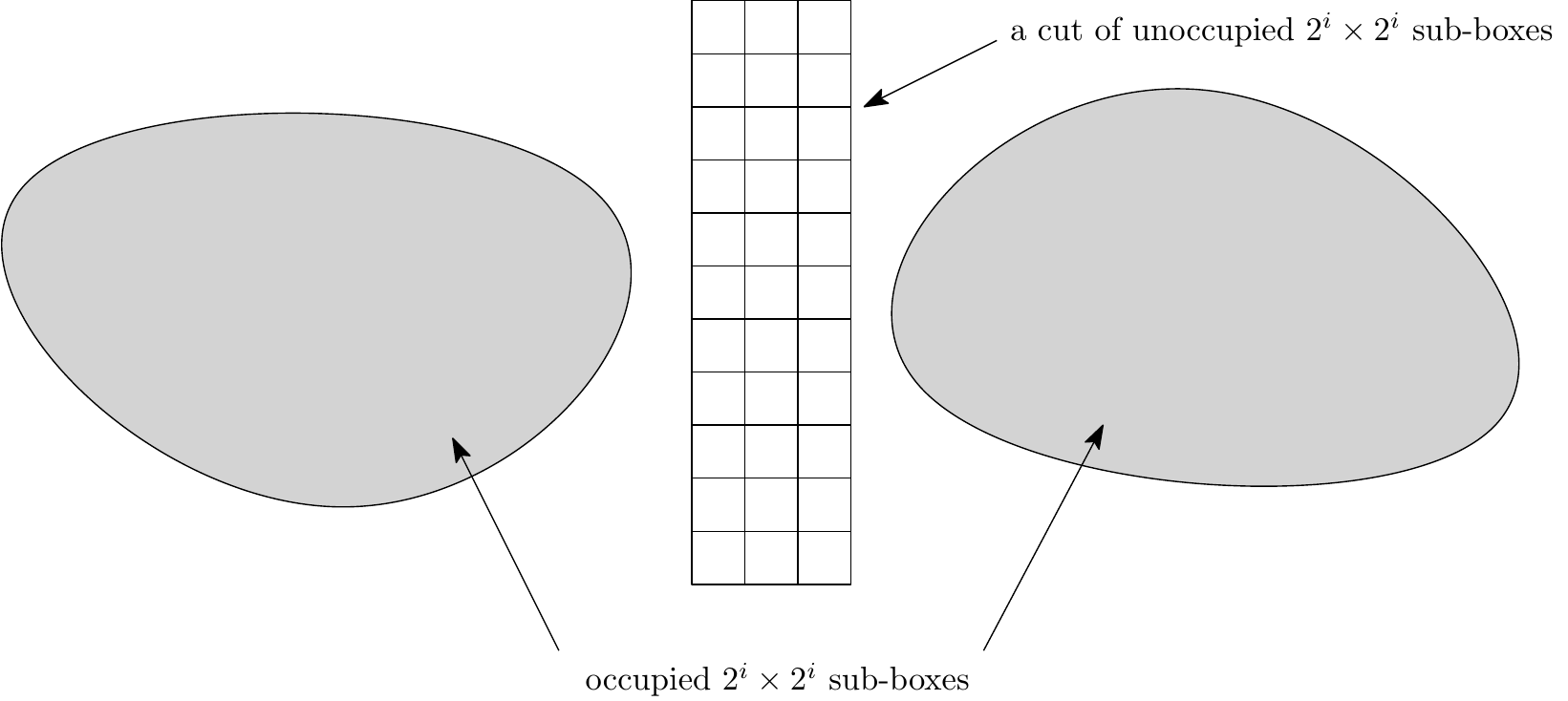}
		\caption{An example of a ``cut'' of unoccupied $2^{i} \times 2^{i}$ sub-boxes.}
		\label{fig:Cut_SuB}	
	\end{figure} 
    Replacing everything by $2^{i-1} \times 2^{i-1}$ sub-boxes, yields that this must also be a cut of $2^{i-1} \times 2^{i-1}$ sub-boxes, because a node cannot have transferred between $2^{i} \times 2^{i}$ sub-boxes before phase $i+1$. But this contradicts the assumption that $2^{i-1} \times 2^{i-1}$ sub-boxes form a connected super-shape. Therefore, it must hold that the $2^{i} \times 2^{i}$ sub-boxes super-shape must have been connected.  	 
\end{proof}

Next, we give an upper bound on the number of occupied sub-boxes in a phase $ i $. 

\begin{lemma} \label{lem:Universal_nlogn}
	Given that \emph{U-Box-Doubling} starts from a connected shape $S_{I} $ of order $n$, the number of occupied sub-boxes in any phase $ i $ is $O(\frac{n}{2^{i}})$.
\end{lemma}
\begin{proof}
First, observe that a $2^{i} \times 2^{i}$ sub-box of phase $ i $ is occupied in that phase iff  $S_{I} $ was originally going through that sub-box. This follows from the fact that nodes are not transferred by this transformation between $2^{i} \times 2^{i} $ sub-boxes before phase $i+1$. Therefore, the $2^{i} \times 2^{i} $ sub-boxes occupied in (any) phase $i$ are exactly the $2^{i} \times 2^{i} $ sub-boxes that the original shape $S_{I} $ would have occupied, thus, it is sufficient to upper bound the number of $2^{i} \times 2^{i} $ sub-boxes that a connected shape of order $ n $  can occupy. Or equivalently, we shall lower bound the number $N_{k}$ of nodes needed to occupy $k$ sub-boxes.

In order to simplify the argument, whenever $S_{I} $ occupies another unoccupied sub-box, we will award it a constant number of additional occupations for free and only calculate the additional distance (in nodes) that the shape has to cover in order to reach another unoccupied sub-box. In particular, pick any node of $S_{I}$ and consider as freely occupied that sub-box and the $ 8  $ sub-boxes surrounding it, as depicted in Figure \ref{fig:Lem_Num_OccSuB} (a). Giving sub-boxes for free can only help the shape, therefore, any lower bound established including the free sub-boxes will also hold for shapes that do not have them (thus, for the original problem). Given that free sub -boxes are surrounding the current node, in order for $S_I$ to occupy another sub-box, at least one surrounding $2^i\times 2^i$ sub-box must be exited. This requires covering a distance of at least $2^i$, through a connected path of nodes.

Once this happens, $S_{I}$ has just crossed the boundary between an occupied sub-box and an unoccupied sub-boxs. Then, by giving it for free at most 5 more unoccupied sub-boxes, $S_{I}$ has to pay another $2^{i}$ nodes to occupy another unoccupied sub-box; see Figure \ref{fig:Lem_Num_OccSuB} (b). We then continue applying this 5-for-free strategy until all $ n $ nodes have been used.

\begin{figure}
	\centering
	\subcaptionbox{}
	{\includegraphics[scale=0.7]{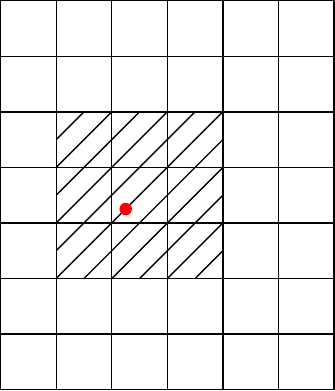}}	 \qquad \qquad \qquad
	\subcaptionbox{}
	{\includegraphics[scale=0.7]{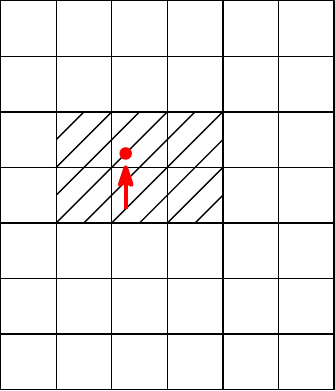}}	
	\caption{(a) A node of shape $S_{I}$ in red and the occupied sub-boxes that we give for free to the shape. (b) The shape just exited the sub-box with arrow entering an unoccupied sub-box. By giving the 5 horizontally dashed sub-boxes for free, a distance of at least $2^{i}$ has to be travelled in order to reach another unoccupied sub-box. }
	\label{fig:Lem_Num_OccSuB}	
\end{figure}

To sum up, the shape has been given $ 8 $ sub-boxes for free, and then for every sub-box covered it has to pay $2^{i}$ and gets 5 sub-boxes. Thus, to occupy $k = 8 + l \cdot 5$ sub-boxes, at least $l \cdot 2^{i}$ nodes are needed, that is,
\begin{align*}
N_{k} &\geq l \cdot 2^{i} \numberthis \label{eqn1}
\end{align*} 
But, that leads to
\begin{align*}
k &= 8 + l \cdot 5 \Rightarrow  l = \dfrac{k-8}{5}.  \numberthis \label{eqn2}
\end{align*} 
Thus,  from \eqref{eqn1} and \eqref{eqn2}:
\begin{align*}
N_{k} &\geq \dfrac{k-8}{5} \cdot 2^{i}. \numberthis \label{eqn3}
\end{align*} 
But shape $S_{I}$ has order $n$, which means that the number of nodes available is upper bounded by $n$, i.e., $N_{k} \le n$, which gives:
\begin{align*}
\dfrac{k-8}{5} \cdot 2^{i}  &\le N_{k} \le n \Rightarrow\\
\dfrac{k-8}{5} \cdot 2^{i}  &\le n \Rightarrow  \dfrac{k-8}{5} \le \dfrac{n}{2^{i}} \Rightarrow\\
k &\le  5 \bigg( \dfrac{n}{2^{i}} \bigg) + 8
\end{align*} 
We conclude that the number of $2^{i} \times 2^{i} $ sub-boxes that can be occupied by a connected shape $S_{I}$, and, thus, also the number of $2^{i} \times 2^{i} $ sub-boxes that are occupied by \emph{U-Box-Doubling} in phase $i$, is at most $ 5 (n/2^{i})+ 8 = O(n/2^{i})$.
\end{proof}

As a corollary of this, we obtain:

\begin{corollary} \label{cor:partitioning}
	Given a uniform partitioning of $n \times n$ square box containing a connected shape $S_{I}$ of order $n$  into $d \times d$ sub-boxes, it holds that $S_{I}$ can occupy at most $O(\frac{n}{d})$ sub-boxes.   
\end{corollary}

We are now ready to analyse the running time of \emph{U-Box-Doubling}. 

\begin{lemma} \label{lem:Universal_nlogn_Upper_bound}
	Starting from any connected shape of $ n $ nodes, \emph{U-Box-Doubling} performs $O(n \log n )$ steps during its course.
\end{lemma}
\begin{proof}
	We prove this by showing that in every phase $i$, $1 \le i \le \log n$, the transformation performs at most a linear number of steps. We partition the occupied $2^{i} \times 2^{i} $ sub-boxes into two disjoint sets, $B_{1} $ and $B_{0}$, where sub-boxes in $B_{1}$ have at least 1 \emph{complete line} (from the previous phase), i.e., a line of length $2^{i-1}$, and sub-boxes in $B_{0}$ have 1 to 4 \emph{incomplete lines}, i.e., lines of length between 1 and $2^{i-1} -1$. For $B_{1}$, we have that $|B_{1}| \le n/2^{i-1}$. Moreover, for every complete line, we pay at most $2^{i-1}$ to transfer it left or down, depending on the parity of $i$. As there are at most $ n/2^{i-1} $ such complete lines in phase $i$, the total cost for this is at most $2^{i} \cdot (n/2^{i-1}) =n$. 
	
	Each sub-box in $B_{1}$ may also have at most 4 incomplete lines from the previous phase, as in Figure \ref{fig:Filling_LeftMost_SuBn} left, where at most two of them may have to pay a maximum of $2^{i-1}$ to be transferred left or down, depending on the parity of $i$ (as the other two are already aligned). As there are at most $ n/2^{i-1} $ sub-boxes in $B_{1}$, the total cost for this is at most $2 \cdot 2^{i-1} \cdot (n/2^{i-1}) = 2n$. 
	
	Therefore, the total cost for pushing all lines towards the required border in $B_{1}$ sub-boxes is at most:
	\begin{align*}
	 n + 2n = 3n. \numberthis \label{eqn11}
	\end{align*}
	For $B_{0}$, we have (by Lemma \ref{lem:Universal_nlogn}) that the total number of occupied sub-boxes in phase $i$ is at most $5 (n/2^{i}) + 8$, therefore, $|B_{0}| \le 5(n/2^{i}) + 8$ (taking into account also the worst case where every occupied sub-box may be of type $B_{0}$). There is again a maximum of 2 incomplete lines per such sub-box that need to be transferred a distance of at most $2^{i-1}$, therefore, the total cost for this to happen in every $B_{0}$ sub-box is at most:
	\begin{align*}
	 2 \cdot 2^{i-1} \bigg( 5 \cdot\frac{n}{2^{i}} +8 \bigg) = 5n + 8\cdot 2^{i} \le 13n. \numberthis \label{eqn12}
	\end{align*}
	By paying the above costs, all occupied sub-boxes have their lines aligned horizontally to their left or vertically to their bottom border, and the final task of the transformation for this phase is to apply a linear procedure in order to fill in the left (bottom) border of the $n \times n $ box. This procedure costs at most $2k$ for every $k$ nodes aligned as above (Lemma \ref{lem:Transfer_Line_H_to_V}), therefore, in total at most:
	 \begin{align*}
	  2n. \numberthis \label{eqn13}
	 \end{align*}
	This completes the operation of \emph{U-Box-Doubling} for phase $i$. Putting \eqref{eqn11}, \eqref{eqn12} and \eqref{eqn13} together, we obtain that the total cost $T_{i}$, in steps, for phase $i$ is,
	\begin{align*}
	T_{i} &\le 3n + 13n + 2n \\
	      &= 18n.  
	\end{align*} 
	As there is a total of $\log n $ phases, we conclude that the total cost $T$ of the transformation is,
		\begin{align*}
	T &\le 18n \cdot \log n \\
	&= O(n \log n).  
	\end{align*} 
\end{proof}

 Finally, together Lemma \ref{lem:U_nlogn_V_OR_H}, Lemma \ref{lem:Universal_nlogn_Upper_bound} and reversibility (Lemma \ref{lem:Transferability_LineMove}) imply that:
 
 \begin{theorem}
 	For any pair of connected shapes $S_{I}$ and $S_{F}$ of the same order $n$, transformation \emph{U-Box-Doubling} can be used to transform $S_{I}$ into $S_{F}$ (and $S_{F}$ into $S_{I}$) in $O(n\log n)$ steps. 
 \end{theorem}

\section{Conclusions}
\label{sec:conclusions}

In this work, we studied a new linear-strength model extending upon the model of \cite{DP04,MICHAIL2019}. The nodes can now move in parallel by translating a line of any length by one position in a single time-step. This model, having the model of \cite{DP04,MICHAIL2019} as a special case, adopts all its transformability results (including universal transformations). Then, our focus naturally turned to investigating if pushing lines can help achieve a substantial gain in performance (compared to the $\Theta(n^2)$ of those models). Even though it can be immediately observed that there are instances in which this is the case (e.g., initial shapes in which there are many long lines, thus, much initial parallelism to be exploited), it was not obvious that this holds also for the worst case. By identifying the diagonal as a potentially worst-case shape (essentially, because in it any parallelism to be exploited does not come for free), we managed to first develop an $O(n\sqrt{n})$-time transformation for transforming the diagonal into a line, then to improve upon this by two transformations that achieve the same bound while preserving connectivity, and finally to provide an $O(n\log n)$-time transformation (that breaks connectivity). Going one step further, we developed two universal transformations that can transform any pair of connected shapes to each other in time $O(n\sqrt{n})$ and $O(n\log n)$, respectively. 

There is a number of interesting problems that are opened by this work. The obvious first target (and apparently intriguing) is to answer whether there is an $o(n\log n)$-time transformation (e.g., linear) or whether there is an $\Omega(n\log n)$-time lower bound matching our best transformations. We suspect the latter, but do not have enough evidence to support or prove it. The tree representation of the problem that we discuss in Section \ref{subsec:DLR-Gathering} (see, e.g., Figure \ref{fig:Recursion_Tree_bound.pdf}), might help in this direction. Moreover, we didn't consider parallel time in this paper. If more than one line can move in parallel in a time-step, then are there variants of our transformations (or alternative ones) that further reduce the running time? In other words, are there parallelisable transformations in this model? In particular, it would be interesting to investigate whether the present model permits an $O(\log n)$ parallel time (universal) transformation, i.e., matching the best transformation in the model of Aloupis \emph{et al.}  \cite{ACD08}. It would also be worth studying in more depth the case in which connectivity has to be preserved during the transformations. In the relevant literature, a number of alternative types of grids have been considered, like triangular (e.g, in \cite{DDGR14}) and hexagonal (e.g., in \cite{WWA04}), and it would be interesting to investigate how our results translate there. Finally, an immediate next goal is to attempt to develop distributed versions of the transformations provided here.



\bibliography{pushing-lines}

\begin{thebibliography}{10}

\bibitem{ABD13}
Greg Aloupis, Nadia Benbernou, Mirela Damian, Erik~D Demaine, Robin Flatland,
  John Iacono, and Stefanie Wuhrer.
\newblock Efficient reconfiguration of lattice-based modular robots.
\newblock {\em Computational geometry}, 46(8):917--928, 2013.

\bibitem{ACD08}
Greg Aloupis, S{\'e}bastien Collette, Erik~D Demaine, Stefan Langerman, Vera
  Sacrist{\'a}n, and Stefanie Wuhrer.
\newblock Reconfiguration of cube-style modular robots using {O}(logn) parallel
  moves.
\newblock In {\em International Symposium on Algorithms and Computation}, pages
  342--353. Springer, 2008.

\bibitem{AADFP06}
Dana Angluin, James Aspnes, Zo{\"e} Diamadi, Michael~J. Fischer, and Ren{\'e}
  Peralta.
\newblock Computation in networks of passively mobile finite-state sensors.
\newblock {\em Distributed Computing}, 18(4):235--253, March 2006.

\bibitem{AAER07}
Dana Angluin, James Aspnes, David Eisenstat, and Eric Ruppert.
\newblock The computational power of population protocols.
\newblock {\em Distributed Computing}, 20(4):279--304, November 2007.

\bibitem{BDFL17}
Aaron~T Becker, Erik~D Demaine, S{\'a}ndor~P Fekete, Jarrett Lonsford, and Rose
  Morris-Wright.
\newblock Particle computation: complexity, algorithms, and logic.
\newblock {\em Natural Computing}, pages 1--21, 2017.

\bibitem{BG15}
Julien Bourgeois and Seth~Copen Goldstein.
\newblock Distributed intelligent mems: progresses and perspectives.
\newblock {\em IEEE Systems Journal}, 9(3):1057--1068, 2015.

\bibitem{BKR04}
Zack Butler, Keith Kotay, Daniela Rus, and Kohji Tomita.
\newblock Generic decentralized control for lattice-based self-reconfigurable
  robots.
\newblock {\em The International Journal of Robotics Research}, 23(9):919--937,
  2004.

\bibitem{CFPS12}
Mark Cieliebak, Paola Flocchini, Giuseppe Prencipe, and Nicola Santoro.
\newblock Distributed computing by mobile robots: Gathering.
\newblock {\em {SIAM} J. Comput.}, 41(4):829--879, 2012.
\newblock \href {http://dx.doi.org/10.1137/100796534}
  {\path{doi:10.1137/100796534}}.

\bibitem{CKLL09}
Alejandro Cornejo, Fabian Kuhn, Ruy Ley-Wild, and Nancy Lynch.
\newblock Keeping mobile robot swarms connected.
\newblock In {\em Proceedings of the 23rd international conference on
  Distributed computing}, DISC'09, pages 496--511, Berlin, Heidelberg, 2009.
  Springer-Verlag.

\bibitem{DFSY15}
Shantanu Das, Paola Flocchini, Nicola Santoro, and Masafumi Yamashita.
\newblock Forming sequences of geometric patterns with oblivious mobile robots.
\newblock {\em Distributed Computing}, 28(2):131--145, April 2015.

\bibitem{DDG18}
Joshua~J Daymude, Zahra Derakhshandeh, Robert Gmyr, Alexandra Porter,
  Andr{\'e}a~W Richa, Christian Scheideler, and Thim Strothmann.
\newblock On the runtime of universal coating for programmable matter.
\newblock {\em Natural Computing}, 17(1):81--96, 2018.

\bibitem{DGMR06}
Xavier D{\'e}fago, Maria Gradinariu, St{\'e}phane Messika, and Philippe
  Raipin-Parv{\'e}dy.
\newblock Fault-tolerant and self-stabilizing mobile robots gathering.
\newblock In {\em International Symposium on Distributed Computing}, pages
  46--60. Springer, 2006.

\bibitem{De01}
Erik~D Demaine.
\newblock Playing games with algorithms: Algorithmic combinatorial game theory.
\newblock In {\em International Symposium on Mathematical Foundations of
  Computer Science}, pages 18--33. Springer, 2001.

\bibitem{DDGR14}
Zahra Derakhshandeh, Shlomi Dolev, Robert Gmyr, Andr{\'e}a~W Richa, Christian
  Scheideler, and Thim Strothmann.
\newblock Brief announcement: amoebot--a new model for programmable matter.
\newblock In {\em Proceedings of the 26th ACM symposium on Parallelism in
  algorithms and architectures (SPAA)}, pages 220--222, 2014.

\bibitem{DGR15}
Zahra Derakhshandeh, Robert Gmyr, Andr{\'e}a~W Richa, Christian Scheideler, and
  Thim Strothmann.
\newblock An algorithmic framework for shape formation problems in
  self-organizing particle systems.
\newblock In {\em Proceedings of the Second Annual International Conference on
  Nanoscale Computing and Communication}, page~21. ACM, 2015.

\bibitem{DGRSS16}
Zahra Derakhshandeh, Robert Gmyr, Andr{\'e}a~W. Richa, Christian Scheideler,
  and Thim Strothmann.
\newblock Universal shape formation for programmable matter.
\newblock In {\em Proceedings of the 28th ACM Symposium on Parallelism in
  Algorithms and Architectures}, pages 289--299. ACM, 2016.

\bibitem{DiLuna2019}
Giuseppe~A. Di~Luna, Paola Flocchini, Nicola Santoro, Giovanni Viglietta, and
  Yukiko Yamauchi.
\newblock Shape formation by programmable particles.
\newblock {\em Distributed Computing}, Mar 2019.
\newblock \href {http://dx.doi.org/10.1007/s00446-019-00350-6}
  {\path{doi:10.1007/s00446-019-00350-6}}.

\bibitem{ALFNVG8}
Giuseppe~Antonio Di~Luna, Paola Flocchini, Giuseppe Prencipe, Nicola Santoro,
  and Giovanni Viglietta.
\newblock Line recovery by programmable particles.
\newblock In {\em Proceedings of the 19th International Conference on
  Distributed Computing and Networking}, ICDCN '18, pages 4:1--4:10, New York,
  NY, USA, 2018. ACM.
\newblock \href {http://dx.doi.org/10.1145/3154273.3154309}
  {\path{doi:10.1145/3154273.3154309}}.

\bibitem{Do12}
David Doty.
\newblock Theory of algorithmic self-assembly.
\newblock {\em Communications of the ACM}, 55:78--88, 2012.

\bibitem{DDL09}
Shawn~M Douglas, Hendrik Dietz, Tim Liedl, Bj{\"o}rn H{\"o}gberg, Franziska
  Graf, and William~M Shih.
\newblock Self-assembly of dna into nanoscale three-dimensional shapes.
\newblock {\em Nature}, 459(7245):414, 2009.

\bibitem{DP04}
Adrian Dumitrescu and J{\'a}nos Pach.
\newblock Pushing squares around.
\newblock In {\em Proceedings of the twentieth annual symposium on
  Computational geometry}, pages 116--123. ACM, 2004.

\bibitem{DSY04a}
Adrian Dumitrescu, Ichiro Suzuki, and Masafumi Yamashita.
\newblock Formations for fast locomotion of metamorphic robotic systems.
\newblock {\em The International Journal of Robotics Research}, 23(6):583--593,
  2004.

\bibitem{DSY04b}
Adrian Dumitrescu, Ichiro Suzuki, and Masafumi Yamashita.
\newblock Motion planning for metamorphic systems: Feasibility, decidability,
  and distributed reconfiguration.
\newblock {\em IEEE Transactions on Robotics and Automation}, 20(3):409--418,
  2004.

\bibitem{FRRS16}
S{\'a}ndor Fekete, Andr{\'e}a~W Richa, Kay R{\"o}mer, and Christian Scheideler.
\newblock Algorithmic foundations of programmable matter ({Dagstuhl Seminar
  16271}).
\newblock In {\em Dagstuhl Reports}, volume~6. Schloss Dagstuhl-Leibniz-Zentrum
  fuer Informatik, 2016.
\newblock Also in \emph{ACM SIGACT News}, 48.2:87-94, 2017.

\bibitem{FPS12}
Paola Flocchini, Giuseppe Prencipe, and Nicola Santoro.
\newblock Distributed computing by oblivious mobile robots.
\newblock {\em Synthesis Lectures on Distributed Computing Theory},
  3(2):1--185, 2012.

\bibitem{GKR10}
Kyle Gilpin, Ara Knaian, and Daniela Rus.
\newblock Robot pebbles: One centimeter modules for programmable matter through
  self-disassembly.
\newblock In {\em Robotics and Automation (ICRA), 2010 IEEE International
  Conference on}, pages 2485--2492. IEEE, 2010.

\bibitem{HD05}
Robert~A Hearn and Erik~D Demaine.
\newblock {PSPACE}-completeness of sliding-block puzzles and other problems
  through the nondeterministic constraint logic model of computation.
\newblock {\em Theoretical Computer Science}, 343(1-2):72--96, 2005.

\bibitem{KCL12}
Ara~N Knaian, Kenneth~C Cheung, Maxim~B Lobovsky, Asa~J Oines, Peter
  Schmidt-Neilsen, and Neil~A Gershenfeld.
\newblock The milli-motein: A self-folding chain of programmable matter with a
  one centimeter module pitch.
\newblock In {\em 2012 IEEE/RSJ International Conference on Intelligent Robots
  and Systems}, pages 1447--1453. IEEE, 2012.

\bibitem{KKM10}
Evangelos Kranakis, Danny Krizanc, and Euripides Markou.
\newblock The mobile agent rendezvous problem in the ring.
\newblock {\em Synthesis Lectures on Distributed Computing Theory},
  1(1):1--122, 2010.

\bibitem{MICHAIL2019}
Othon Michail, George Skretas, and Paul~G. Spirakis.
\newblock On the transformation capability of feasible mechanisms for
  programmable matter.
\newblock {\em Journal of Computer and System Sciences}, 2019.
\newblock \href {http://dx.doi.org/https://doi.org/10.1016/j.jcss.2018.12.001}
  {\path{doi:https://doi.org/10.1016/j.jcss.2018.12.001}}.

\bibitem{MS16a}
Othon Michail and Paul~G. Spirakis.
\newblock Simple and efficient local codes for distributed stable network
  construction.
\newblock {\em Distributed Computing}, 29(3):207--237, 2016.
\newblock \href {http://dx.doi.org/http://dx.doi.org/10.1007/s00446-015-0257-4}
  {\path{doi:http://dx.doi.org/10.1007/s00446-015-0257-4}}.

\bibitem{MS18}
Othon Michail and Paul~G Spirakis.
\newblock Elements of the theory of dynamic networks.
\newblock {\em Communications of the ACM}, 61(2), 2018.

\bibitem{NGY00}
An~Nguyen, Leonidas~J Guibas, and Mark Yim.
\newblock Controlled module density helps reconfiguration planning.
\newblock In {\em Proc. of 4th International Workshop on Algorithmic
  Foundations of Robotics}, pages 23--36, 2000.

\bibitem{RW00}
Paul W.~K. Rothemund and Erik Winfree.
\newblock The program-size complexity of self-assembled squares.
\newblock In {\em Proceedings of the 32nd annual ACM symposium on Theory of
  computing (STOC)}, pages 459--468, 2000.
\newblock \href {http://dx.doi.org/10.1145/335305.335358}
  {\path{doi:10.1145/335305.335358}}.

\bibitem{Ro06}
Paul~WK Rothemund.
\newblock Folding dna to create nanoscale shapes and patterns.
\newblock {\em Nature}, 440(7082):297--302, 2006.

\bibitem{RCN14}
Michael Rubenstein, Alejandro Cornejo, and Radhika Nagpal.
\newblock Programmable self-assembly in a thousand-robot swarm.
\newblock {\em Science}, 345(6198):795--799, 2014.

\bibitem{SMO18}
Masahiro Shibata, Toshiya Mega, Fukuhito Ooshita, Hirotsugu Kakugawa, and
  Toshimitsu Masuzawa.
\newblock Uniform deployment of mobile agents in asynchronous rings.
\newblock {\em Journal of Parallel and Distributed Computing}, 119:92--106,
  2018.

\bibitem{WWA04}
Jennifer~E Walter, Jennifer~L Welch, and Nancy~M Amato.
\newblock Distributed reconfiguration of metamorphic robot chains.
\newblock {\em Distributed Computing}, 17(2):171--189, 2004.

\bibitem{Wi98}
Erik Winfree.
\newblock {\em Algorithmic Self-Assembly of DNA}.
\newblock PhD thesis, California Institute of Technology, June 1998.

\bibitem{WCG13}
Damien Woods, Ho-Lin Chen, Scott Goodfriend, Nadine Dabby, Erik Winfree, and
  Peng Yin.
\newblock Active self-assembly of algorithmic shapes and patterns in
  polylogarithmic time.
\newblock In {\em Proceedings of the 4th conference on Innovations in
  Theoretical Computer Science}, pages 353--354. ACM, 2013.

\bibitem{SY10}
Masafumi Yamashita and Ichiro Suzuki.
\newblock Characterizing geometric patterns formable by oblivious anonymous
  mobile robots.
\newblock {\em Theoretical Computer Science}, 411(26-28):2433--2453, 2010.

\bibitem{YUY16}
Yukiko Yamauchi, Taichi Uehara, and Masafumi Yamashita.
\newblock Brief announcement: pattern formation problem for synchronous mobile
  robots in the three dimensional euclidean space.
\newblock In {\em Proceedings of the 2016 ACM Symposium on Principles of
  Distributed Computing}, pages 447--449. ACM, 2016.

\bibitem{YSS07}
Mark Yim, Wei-Min Shen, Behnam Salemi, Daniela Rus, Mark Moll, Hod Lipson, Eric
  Klavins, and Gregory~S Chirikjian.
\newblock Modular self-reconfigurable robot systems [grand challenges of
  robotics].
\newblock {\em IEEE Robotics \& Automation Magazine}, 14(1):43--52, 2007.

\end{thebibliography}

\appendix

\end{document}